%% file: document.tex
\documentclass[letterpaper,11pt]{article}
\usepackage{microtype}
\usepackage[letterpaper,margin=0.97in]{geometry}
\usepackage[numbers,sort&compress]{natbib}
\usepackage[table]{xcolor}
\usepackage{amsmath}
\usepackage{amssymb}
\usepackage{mathrsfs}
\usepackage{amsthm}
\usepackage{enumitem}
\usepackage{textgreek}
\usepackage{graphicx}
\usepackage{framed}
\usepackage[small]{caption}
\usepackage{booktabs}
\usepackage{tikz-cd}
\usepackage{hyperref}
\usepackage{doi}
\usepackage[textsize=tiny]{todonotes}
\usepackage[framemethod=tikz]{mdframed}
\usepackage{thm-restate}
\usepackage{mathtools}
\usepackage{xspace}
\usepackage[capitalize, nameinlink]{cleveref}
\usepackage{multirow}
\usepackage{cellspace}
\usepackage{array}
\usepackage{tabularx}

\Crefname{remark}{Remark}{Remarks}
\Crefname{observation}{Observation}{Observations}

\theoremstyle{plain}
\newtheorem{theorem}{Theorem}[section]
\newtheorem{lemma}[theorem]{Lemma}

\newtheorem{corollary}[theorem]{Corollary}

\newtheorem{observation}[theorem]{Observation}

\theoremstyle{definition}
\newtheorem{definition}[theorem]{Definition}

\theoremstyle{plain}

\newcounter{open}
\newtheorem{oq}[open]{Open Problem}

\theoremstyle{remark}

\newcommand{\namedref}[2]{\hyperref[#2]{#1~\ref*{#2}}}

\newcommand{\ARB}{\ensuremath{\mathbb{N}_0}}

\newcommand{\capa}{\ensuremath{\mathsf{cap}}}

\newcommand{\eps}{\varepsilon}

\newcommand{\s}{\mspace{1mu}}
\newcommand{\mybox}[1]{\mspace{2mu}{\setlength{\fboxsep}{1.5pt}\color{lightgray}\boxed{\color{black}\scriptstyle #1}}\mspace{2mu}}
\newcommand{\A}{\mathsf{A}}
\newcommand{\B}{\mathsf{B}}
\newcommand{\C}{\mathsf{C}}
\renewcommand{\L}{\mathsf{L}}
\renewcommand{\c}{\mathsf{c}}

\newcommand{\y}{\mathsf{y}}

\newcommand{\Z}{\mathsf{Z}}
\newcommand{\U}{\mathsf{U}}
\newcommand{\W}{\mathsf{W}}
\renewcommand{\S}{\mathsf{S}}
\newcommand{\M}{\mathsf{M}}
\renewcommand{\P}{\mathsf{P}}

\newcommand{\D}{\mathsf{D}}
\newcommand{\Q}{\mathsf{Q}}

\newcommand{\X}{\mathsf{X}}
\newcommand{\Y}{\mathsf{Y}}

\renewcommand{\S}{\mathsf{S}}

\newcommand{\dis}{\ensuremath{\operatorname{disj}}}
\newcommand{\newy}{\ensuremath{\operatorname{new}}}

\newcommand{\fA}{\mathcal{A}}
\newcommand{\fB}{\mathcal{B}}

\newcommand{\fE}{\mathcal{E}}
\newcommand{\fF}{\mathcal{F}}
\newcommand{\fS}{\mathcal{S}}

\DeclareMathOperator{\re}{\mathcal R}
\DeclareMathOperator{\rere}{\overline{\mathcal R}}

\newcommand{\nodeconst}{\ensuremath{\mathcal{N}}}
\newcommand{\edgeconst}{\ensuremath{\mathcal{E}}}
\newcommand{\gen}[1]{\langle #1 \rangle}

\DeclareMathOperator{\poly}{poly}

\newcommand{\LOCAL}{\ensuremath{\mathsf{LOCAL}}\xspace}
\newcommand{\CONGEST}{\ensuremath{\mathsf{CONGEST}}\xspace}

\newcommand{\set}[1]{\left\{#1\right\}}

\newcommand{\ccs}{\mathfrak{C}}
\newcommand{\ccc}{\mathcal{C}}
\newcommand{\level}{\mathrm{level}}
\newcommand{\len}{\mathrm{len}}
\newcommand{\prefix}{\mathrm{prefix}}

\newenvironment{myabstract}
{\list{}{\listparindent 1.5em%
		\itemindent    \listparindent
		\leftmargin    1cm
		\rightmargin   1cm
		\parsep        0pt}%
	\item\relax}
{\endlist}

\newenvironment{mycover}
{\list{}{\listparindent 0pt
		\itemindent    \listparindent
		\leftmargin    1cm
		\rightmargin   1cm
		\parsep        0pt}%
	\raggedright
	\item\relax}
{\endlist}

\newcommand{\myemail}[1]{\,$\cdot$\, {\small #1}}
\newcommand{\myaff}[1]{\,$\cdot$\, {\small #1}\par\smallskip}

\definecolor{darkgreen}{rgb}{0,0.5,0}
\definecolor{darkred}{rgb}{0.4,0,0}
\hypersetup{
	colorlinks=true,
	linkcolor=darkred,
	citecolor=darkgreen,
	filecolor=black,
	urlcolor=[rgb]{0,0.1,0.5},
	pdftitle={$\Delta$-Coloring Plays Hide-and-Seek},
	pdfauthor={Alkida Balliu, Sebastian Brandt, Fabian Kuhn, Dennis Olivetti}
}

\begin{document}

\begin{mycover}
	{\huge\bfseries\boldmath Distributed $\Delta$-Coloring Plays Hide-and-Seek \par}
	\bigskip
	\bigskip
	\bigskip
	
	\textbf{Alkida Balliu}
	\myemail{alkida.balliu@gssi.it}
	\myaff{Gran Sasso Science Institute}
	
	\textbf{Sebastian Brandt}
	\myemail{brandt@cispa.de}
	\myaff{CISPA Helmholtz Center for Information Security}
	
	\textbf{Fabian Kuhn}
	\myemail{kuhn@cs.uni-freiburg.de}
	\myaff{University of Freiburg}
	
	\textbf{Dennis Olivetti}
	\myemail{dennis.olivetti@gssi.it}
	\myaff{Gran Sasso Science Institute}
\end{mycover}
\bigskip

\begin{myabstract}
  We prove several new tight or near-tight distributed lower bounds
  for classic symmetry breaking problems in graphs. As a basic tool, we
  first provide a new insightful proof that any deterministic
  distributed algorithm that computes a $\Delta$-coloring on
  $\Delta$-regular trees requires $\Omega(\log_\Delta n)$ rounds and
  any randomized such algorithm requires $\Omega(\log_\Delta\log n)$
  rounds. We prove this by showing that a natural relaxation of
  the $\Delta$-coloring problem is a fixed point in the round
  elimination framework.
  
  As a first application, we show that our $\Delta$-coloring lower
  bound proof directly extends to arbdefective colorings. An
  arbdefective $c$-coloring of a graph $G=(V,E)$ is given by a
  $c$-coloring of $V$ and an orientation of $E$, where the arbdefect
  of a color $i$ is the maximum number of monochromatic outgoing edges
  of any node of color $i$. We exactly characterize which variants of
  the arbdefective coloring problem can be solved in
  $O(f(\Delta) + \log^*n)$ rounds, for some function $f$, and which of
  them instead require $\Omega(\log_\Delta n)$ rounds for
  deterministic algorithms and $\Omega(\log_\Delta \log n)$ rounds for
  randomized ones.

  As a second application, which we see as our main contribution, we
  use the structure of the fixed point as a building block to prove
  lower bounds as a function of $\Delta$ for problems that, in some
  sense, are \emph{much easier} than $\Delta$-coloring, as they can be
  solved in $O(\log^* n)$ deterministic rounds in bounded-degree
  graphs.  More specifically, we prove lower bounds as a function of
  $\Delta$ for a large class of distributed symmetry breaking
  problems, which can all be solved by a simple sequential greedy
  algorithm. For example, we obtain novel results for the fundamental problem
  of computing a $(2,\beta)$-ruling set, i.e., for computing an
  independent set $S\subseteq V$ such that every node $v\in V$ is
  within distance $\leq \beta$ of some node in $S$. We in particular
  show that $\Omega(\beta\Delta^{1/\beta})$ rounds are needed
  even if initially an $O(\Delta)$-coloring of the graph is
  given. With an initial $O(\Delta)$-coloring, this lower bound is
  tight and without, it still nearly matches the existing
  $O(\beta\Delta^{2/(\beta+1)}+\log^* n)$ upper bound. The new
  $(2,\beta)$-ruling set lower bound is an exponential improvement
  over the best existing lower bound for the problem, which was proven
  in [FOCS '20]. As a special case of the lower bound,
  we also obtain a tight linear-in-$\Delta$ lower bound for computing
  a maximal independent set (MIS) in trees. While such an MIS lower
  bound was known for general graphs, the best previous MIS lower
  bounds for trees was $\Omega(\log\Delta)$.  Our lower bound even applies to a much more general family of problems that allows for almost arbitrary combinations of natural constraints from coloring problems, orientation problems, and independent set problems, and provides a single unified proof for known and new lower bound results for these types of problems.
  
  All of our lower bounds as a function of $\Delta$ also imply substantial lower bounds
  as a function of $n$. For instance, we obtain that the maximal
  independent set problem, on trees, requires
  $\Omega(\log n / \log \log n)$ rounds for deterministic algorithms,
  which is tight.
\end{myabstract}

\clearpage
\setcounter{page}{0}
\thispagestyle{empty}
\tableofcontents

\clearpage

\input{introduction}

\subsection{Our Results}
\label{sec:ourresults}

\definecolor{lightblue}{rgb}{0.68, 0.85, 0.9}
\newcommand{\highlight}{\cellcolor{lightblue!20}}

\begin{table*}
	\caption{Some of our results and how they compare with the state of the art. Whenever a condition contains $\eps_0$, then $\eps_0$ has to be chosen as a sufficiently small positive constant. The term ``easy'' indicates that the results follow by applying standard techniques.}
	\label{tab:results}    
\begin{tabularx}{\textwidth}{X l l Sl c c r}
	\midrule
	\small\begin{tabular}[c]{@{}l@{}}Problem\\ \\ \\\\ \end{tabular}&  &  & \small\begin{tabular}[c]{@{}l@{}}Complexity\\ \\ \\\\ \end{tabular} & \small\begin{tabular}[c]{@{}l@{}}Holds \\ on\\ trees\\\\\end{tabular} & \small\begin{tabular}[c]{@{}l@{}}Holds\\ on\\ general\\ graphs\end{tabular} & \small\begin{tabular}[c]{@{}l@{}}Ref.\\ \\ \\ \\ \end{tabular}                               \\ \midrule
	\multirow{3}{*}{\footnotesize\begin{tabular}[c]{@{}l@{}}$(2,\beta)$-ruling set,\\as a function of $n$,\\ for $\beta$ small enough\end{tabular}}                                                                                                                                                                    &                                       &  det.     & $O\big(\log^5 n\big)$                                                         & \checkmark                                                  & \checkmark                                                            & \cite{GGR2020}                          \\
	&                                           &       & $O\Big(\frac{\log n}{\beta \log\log n}\Big)$                          & \checkmark                                                  &                                                                       & \cite{BarenboimE10}                     \\
	&                                           &       & $\Omega\Big(\sqrt{\frac{\log n}{\beta\log\log n}}\Big)$               & \checkmark                                                  & \checkmark                                                            & \cite{balliurules}                      \\
	&                                           &       &\highlight{}$\Omega\Big(\frac{\log n}{\beta \log\log n}\Big)$                     &\highlight{}\checkmark                                                  &\highlight{}\checkmark                                                            &\highlight{}\emph{\textbf{new}}                     \\\cmidrule{3-7} 
	&                                      &   rand.    & $O\big(\beta\log^{\frac{1}{\beta+1}} n\big)$                                  & \checkmark                                                  &                                                                       & \begin{tabular}[c]{@{}r@{}}\cite{Barenboim2016}\\\cite{ghaffari16improved}\end{tabular} \\
	&                                           &       & $\Omega\Big(\sqrt{\frac{\log \log n}{\beta\log\log\log n}}\Big)$      & \checkmark                                                  & \checkmark                                                            & \cite{balliurules}                      \\
	&                                           &       &\highlight{}$\Omega\Big(\frac{\log\log n}{\beta \log\log\log n}\Big)$             &\highlight{}\checkmark                                                  &\highlight{}\checkmark                                                            &\highlight{}\emph{\textbf{new}}                     \\ \midrule
	\multirow{5}{*}{\begin{tabular}[c]{@{}l@{}}$(2,\beta)$-ruling set,\\as a function of $\Delta$,\\ for $\beta$ small enough,\\ given an \\$O(\Delta)$-coloring\end{tabular}}                                                                                  &                                       &   det.    & $O\big(\beta\Delta^{1/\beta}\big)$                                            & \checkmark                                                  & \checkmark                                                            & \cite{SEW13}                            \\ \cmidrule{3-7} 
	&                             &  \multirow{2}{*}{\begin{tabular}[c]{@{}l@{}}det.\\\&\\rand.\end{tabular}}      & $\Omega\Big(\frac{\log\Delta}{\beta\log\log\Delta}\Big)$              & \checkmark                                                  & \checkmark                                                            & \cite{balliurules}                      \\
	&                                           &       &\highlight{}$\Omega\big(\beta\Delta^{1/\beta}\big)$                                       &\highlight{}\checkmark                                                  &\highlight{}\checkmark                                                            &\highlight{}\emph{\textbf{new}}                     \\\\[-2pt] \midrule
	\multirow{2}{*}{\begin{tabular}[c]{@{}l@{}}$\alpha$-arbdefective\\ $c$-coloring\end{tabular}}                                                                                                                                                                                                                 & $(\alpha+1)c > \Delta$                    & det.  & $O\big(\Delta + \log^* n\big)$                                                & \checkmark                                                  & \checkmark                                                            & easy                                \\ \cmidrule{2-7} 
	& $(\alpha+1)^2 c \le \varepsilon_0 \Delta$ & det.  & $\Omega\big(\log_\Delta n\big)$                                               & \checkmark                                                  & \checkmark                                                            & easy                                \\ \cmidrule{3-7} 
	&                                           & rand. & $\Omega\big(\log_\Delta \log n\big)$                                          & \checkmark                                                  & \checkmark                                                            & easy                                \\ \cmidrule{2-7} 
	& $(\alpha+1) c \le \varepsilon_0 \Delta$   & det.  & $\Omega\big(\log_\Delta n\big)$                                               &                                                             & \checkmark                                                            & easy                                \\\cmidrule{3-7} 
	&                                           & rand. & $\Omega\big(\log_\Delta \log n\big)$                                          &                                                             & \checkmark                                                            & easy                                \\ \cmidrule{2-7} 
	& $(\alpha+1)c \le \Delta$                  & det.  &\highlight{}$\Omega\big(\log_\Delta n\big)$                                               &\highlight{}\checkmark                                                  &\highlight{}\checkmark                                                            &\highlight{}\emph{\textbf{new}}                     \\ \cmidrule{3-7} 
	&                                           & rand. &\highlight{}$\Omega\big(\log_\Delta \log n\big)$                                          &\highlight{}\checkmark                                                  &\highlight{}\checkmark                                                            &\highlight{}\emph{\textbf{new}}                     \\ \midrule
	\multirow{3}{*}{\begin{tabular}[c]{@{}l@{}}$\alpha$-arbdefective\\$c$-colored\\$\beta$-ruling set,\\as a function of $\Delta$,\\ for $\beta$ small enough,\\ given an $\alpha$-arbdef.\\ $O\big(\frac{\Delta}{\alpha+1}\big)$-coloring\end{tabular}} &                                       &  det.     &\highlight{}$O\Big(\beta\big(\frac{\Delta}{(\alpha+1) c}\big)^{1/\beta}\Big)$     &\highlight{}\checkmark                                                  &\highlight{}\checkmark                                                            &\highlight{}\emph{\textbf{new}}                     \\ \cmidrule{3-7} 
	&                              &    \multirow{2}{*}{\begin{tabular}[c]{@{}l@{}}det.\\\&\\rand.\end{tabular}} &\highlight{}$\Omega\Big(\beta\big(\frac{\Delta}{(\alpha+1)c}\big)^{1/\beta}\Big)$ &\highlight{}\checkmark                                                  &\highlight{}\checkmark                                                            &\highlight{}\emph{\textbf{new}}                     \\
	&                                           &       &                                                                       &                                                             &                                                                       &                                         \\
	&                                           &       &                                                                       &                                                             &                                                                       & \multicolumn{1}{l}{}                    \\[11pt]
	&                                           &       &                                                                       &                                                             &                                                                       & \multicolumn{1}{l}{}                    \\ \midrule
\end{tabularx}
\end{table*}

Our main lower bound result applies to an entire family of problems for which we gave a high level description before \Cref{sec:problems} and that is defined formally in \Cref{sec:family}. We now present the most important results that we obtain as corollaries for the most natural variants of our general problem family (see \Cref{tab:results} for a partial summary of our results and a comparison with the state of the art).

We first consider arbdefective colorings. We prove that we can exactly characterize which variants of arbdefective coloring can be solved in $O(f(\Delta) + O(\log^* n))$ for some function $f$, and which of them require $\Omega(\log_\Delta n)$ for deterministic algorithms and $\Omega(\log_\Delta \log n)$ for randomized ones. In particular, we show the following (follows from \Cref{thm:ubarbdef,})
\begin{theorem}
	The $\vec{d}$-arbdefective $c$-coloring problem can be solved in $O(\Delta + \log^* n)$ in the deterministic \LOCAL model if $\vec{d}$ is $1$-relaxed, while otherwise it requires $\Omega(\log_\Delta n)$ rounds for deterministic algorithms and $\Omega(\log_\Delta \log n)$ rounds for randomized algorithms.
\end{theorem}
By allowing the same arbdefect on each color, we obtain the following corollary.
\begin{corollary}\label{cor:arbdefcol}
	The $d$-arbdefective $c$-coloring problem can be solved in $O(\Delta + \log^* n)$ rounds in the deterministic \LOCAL model if $(d+1)c > \Delta$, while if $(d+1)c \le \Delta$ then it requires $\Omega(\log_\Delta n)$ rounds for deterministic algorithms and $\Omega(\log_\Delta \log n)$ rounds for randomized algorithms.
\end{corollary}

We now focus on ruling sets and some of their variants. Consider a variant of ruling sets where nodes are required to prove that there is a node in the set at distance at most $\beta$.
For this variant, it was shown in \cite{balliurules}  that, given a proper $c$-coloring, one can solve the problem in $t$ rounds where $t$ is the minimum value such that ${t + \beta \choose \beta} \ge c$. 
We prove that this upper bound is tight, assuming that a proper $(\Delta + 1)$-coloring is given as input to the nodes. In particular, we show that this problem requires at least $t+1$ rounds, where $t$ is the maximum value such that ${t + \beta \choose \beta} < \Delta$. 

For the standard variants of ruling sets, which do not require to
prove that there is a node in the set at distance at most $\beta$, we
lose an additive $\beta$ term in the lower bound, but we still obtain
tight results for $\beta \leq \eps_0\log \Delta$ for some constant
$\eps_0>0$.  In particular, these results exponentially improve the
lower bounds in \cite{balliurules}, and are asymptotically tight in
the case in which an $O(\Delta)$-coloring is provided to the
nodes. Notice that a stronger lower bound for the case where we are
not given an $O(\Delta)$-coloring, would directly imply an
$\omega(\log^* n)$ lower bound for $O(\Delta)$-vertex coloring (and
thus also for $(\Delta+1)$-vertex coloring), which is a long-standing
open question.

Our lower bound does not only hold for ruling sets, but for the more general  $\alpha$-arbdefective $c$-colored $\beta$-ruling sets that we defined in \Cref{sec:problems}. While the case of $\beta=0$ has been handled by \Cref{cor:arbdefcol}, in the remaining cases the following holds.
\begin{restatable}{theorem}{lbcoldom}\label{thm:lbarbcolrs}
  Let $t$ be the maximum value such that ${t + \beta \choose \beta} < \frac{\Delta}{c(\alpha+1)}$.
  The $\alpha$-arbdefective $c$-colored $\beta$-ruling set problem requires $\Omega(\min\{t, \log_\Delta n\} - \beta)$ in the deterministic \LOCAL model and $\Omega(\min\{t, \log_\Delta \log n\} - \beta)$ in the randomized \LOCAL model.
\end{restatable}

\noindent As proven in \Cref{cor:arbdefrsdelta} (see also \Cref{tab:results}),
\Cref{thm:lbarbcolrs} implies that for sufficiently small $\beta$,
computing an $\alpha$-arbdefective $c$-coloring $\beta$-ruling set
requires $\Omega\big(\min\big\{\beta (\frac{\Delta}{c(1+\alpha)})^{1/\beta}, \log_\Delta n\big\}\big)$ rounds deterministically
and $\Omega\big(\min\big\{\beta (\frac{\Delta}{c(1+\alpha)})^{1/\beta}, \log_\Delta \log n\big\}\big)$ rounds with randomization. For
details and additional implications of the above theorem, we refer to
\Cref{sec:corollaries}. The theorem in particular also implies that
computing an MIS requires $\Omega\big(\frac{\log n}{\log\log n}\big)$
rounds in the deterministic \LOCAL model, even in trees. Notice that
this result is tight, as in \cite{BarenboimE10} it has been shown that
MIS on trees can be solved in $O\big(\frac{\log n}{\log\log n}\big)$
deterministic rounds. The following theorem shows that our general
lower bound for arbdefective colored ruling sets is tight if we are
initially given an appropriate coloring as input.

\begin{restatable}{theorem}{ubcoldom}\label{thm:ubarbcolrs}
   Let $G=(V,E)$ be a graph and let $\alpha\geq 0$, $c\geq 1$, $\beta\geq 1$, and $C\geq c$. If an $\alpha$-arbdefective $C$-coloring of $G$ is given, then an $\alpha$-arbdefective $c$-colored $\beta$-ruling set of $G$ can be computed deterministically in the minimum time $t$ such that ${t + \beta \choose \beta}\geq \lceil\frac{C}{c}\rceil$ in the \LOCAL model.
\end{restatable}

\noindent \Cref{thm:ubarbcolrs} implies that if an $\alpha$-arbdefective $C$-coloring of $G$ is given as input, then an $\alpha$-arbdefective $c$-colored $\beta$-ruling set of $G$ can be computed deterministically in $O\big(\beta\big(\frac{C}{c}\big)^{1/\beta}\big)$ rounds (cf.\ \Cref{cor:basicarbruling}). By combining with the arbdefective coloring algorithm given by \Cref{lemma:relaxedarbcoloring}, this implies that an $\alpha$-arbdefective $c$-colored $\beta$-ruling set can be computed in $O\big(\beta\big(\frac{C}{(\alpha+1)\sqrt{c}}\big)^{2/(\beta+1)}\big)$ rounds in the deterministic \LOCAL model (cf.\ \Cref{cor:basicarbruling2}). We refer the reader to \Cref{sec:corollaries,sec:algorithms} for a full list of results that we obtain.

\subsection{Additional Related Work}

In the following, we list additional relevant previous work. In light of the vast amount of related work on the techniques and problems that we study in this paper, we restrict the discussion to the results that are most relevant in the context of the present paper.

\paragraph{Round Elimination.}
Round elimination  has proven to be an extremely useful
tool for showing lower bounds in the distributed setting, \cite{Brandt2016,chang16exponential,chang18complexity,Brandt2019,BalliuHOS19,Olivetti2019,Balliu2019,trulytight,balliurules,BBKOmis,binary,FraigniaudPaz20}. Although we
started to widely understand the technique only recently, already
Linial's classic $\Omega(\log^* n)$ lower bound for $3$-coloring a
cycle \cite{Linial1992} can be seen as a round elimination proof.
The clean concept of round elimination, as we use it nowadays, has
been introduced in 2016 in \cite{Brandt2016}, where the authors prove a randomized
lower bound of $\Omega(\log_\Delta \log n)$ rounds for computing a
sinkless orientation or a $\Delta$-coloring.
Shortly after, Chang, Kopelowitz, and Pettie~\cite{chang16exponential} lifted the lower bound to an
$\Omega(\log_\Delta n)$-round lower bound for deterministic
algorithms. Later, the lower bound was also extended to the $(2\Delta-2)$-edge coloring problem in \cite{chang18complexity}.


In \cite{Brandt2019}, Brandt further refined the technique by
giving an \emph{automatic} way to perform round elimination. That is,
given a problem $\Pi$, he showed how to mechanically construct a
problem $\Pi'$ that is exactly one round easier than $\Pi$.
The new automatic round elimination technique, together with a tool by Olivetti
\cite{Olivetti2019} that implements it, gave rise to a better
understanding of several important problems in distributed
computing. As a highlight, \cite{Balliu2019} used the technique to
prove asymptotically tight-in-$\Delta$ lower bounds for computing a
maximal matching (MM) and as an immediate corollary also for computing
an MIS. \cite{trulytight} further improved these results by showing
that MM on bipartite $2$-colored graphs requires exactly $2\Delta-1$
rounds. While the lower bounds for MM hold on trees, the MIS lower
bounds only hold on more general graphs. Moreover, the above lower
bounds cannot be directly extended to ruling sets, since it is known
that ruling sets on the line graph are much easier than
MIS~\cite{KuhnMW18}.

Balliu, Brandt, and Olivetti~\cite{balliurules} made progress on this question and used the automatic round elimination technique to show the first $\omega(\log^* n)$ lower bound for ruling sets and the first such lower bound for MIS on trees. Unfortunately, their lower bounds are exponentially worse than the best known upper bounds: they for example show that, on trees, MIS requires $\Omega(\log \Delta / \log \log \Delta)$ rounds.
Subsequently, \cite{BBKOmis} showed that by making use of an input edge coloring, one can apply the round elimination technique to improve and simplify the lower bound for MIS on trees, obtaining a lower bound of $\Omega(\log\Delta)$ rounds. Unfortunately, this simpler proof does not extend to ruling sets, but the authors show that it can be used to prove lower bounds for a class of problems called bounded outdegree dominating sets (that are the same as arbdefective $1$-colored $1$-ruling sets).

\paragraph{Maximal Independent Set.}

The problem of computing an MIS is one of the most fundamental symmetry breaking problems and it has been extensively studied in the \LOCAL model (for example, see \cite{Awerbuch89, KarpW85, Luby1986, Alon1986, Linial1992, Naor1991, panconesi96decomposition, Kuhn2004, BarenboimE10, SchneiderW10, LenzenW11, Barenboim2016, Barenboim2013, barenboim14distributed, ghaffari16improved, Rozhon2020,GGR2020, DISC17_MIS,Ghaffari19congest,ghaffariPortman19}). Already in the 80s, the works of Luby~\cite{Luby1986} and Alon, Babai, and Itai~\cite{Alon1986} showed that with randomization, an MIS can be computed in $O(\log n)$ rounds. The first non-trivial deterministic upper bounds for MIS were also designed in the late 80s, when \cite{Awerbuch89} introduced a $2^{O(\sqrt{\log n\log\log n})}$-round \LOCAL algorithm for a more general tool, called network decomposition. Later, Panconesi and Srinivasan~\cite{panconesi96decomposition} improved this runtime to $2^{O(\sqrt{\log n})}$, and very recently the deterministic complexity of computing a network decomposition (and an MIS) has been improved to polylogarithmic in $n$ by Rozho\v{n} and Ghaffari~\cite{Rozhon2020}. An improved variant of this algorithm by Ghaffari, Grunau, and Rozho\v{n}~\cite{GGR2020} provides the current best distributed deterministic MIS algorithm with a round complexity of $O(\log^5 n)$. In addition to analyzing the complexity of computing an MIS solely as a function of $n$, there is also a large body of work that studies runtimes expressed as $f(\Delta) + g(n)$. In 2008, Barenboim, Elkin, and Kuhn~\cite{barenboim14distributed} showed that MIS can be computed deterministically in $O(\Delta+\log^* n)$ deterministic rounds in the \LOCAL model. Furthermore, the authors presented randomized algorithms with a round complexity of $O(\log^2\Delta) + 2^{O(\sqrt{\log\log n})}$, which was improved to $O(\log\Delta)+2^{O(\sqrt{\log\log n})}$ by Ghaffari~\cite{ghaffari16improved}. As a result of the network decomposition algorithm of \cite{Rozhon2020,GGR2020}, the complexity has been further improved to $O(\log\Delta + \log^5\log n)$.

On the lower bound side, we have known since the late 80s and early 90s that MIS requires $\Omega(\log^*n)$ rounds, even for randomized algorithms, which is due to the works of Linial~\cite{Linial1992} and Naor~\cite{Naor1991}. The first $\omega(\log^* n)$ lower bound for MIS was shown in 2004 by Kuhn, Moscibroda, and Wattenhofer~\cite{Kuhn2004}, who proved a lower bound of $\Omega\Bigl(\min\Bigl\{\frac{\log\Delta}{\log\log\Delta},\sqrt{\frac{\log n}{\log\log n}}\Bigr\}\Bigr)$ rounds, even for randomized algorithms. Recently, this lower bound has been improved and complemented by the round elimination based maximal matching lower bound of \cite{Balliu2019}, which implies that computing an MIS requires $\Omega\Bigl(\min\Bigl\{\Delta, \frac{\log n}{\log \log n}\Bigr\}\Bigr)$ rounds for deterministic algorithms and $\Omega\Bigl(\min\Bigl\{\Delta, \frac{\log \log n}{\log \log \log n}\Bigr\}\Bigr)$ rounds for randomized algorithms. As discussed, these results do not hold in trees and the first lower bounds for MIS in trees were recently proven in \cite{balliurules,BBKOmis}. In this paper, we prove that the same MIS lower bounds that were proven in \cite{Balliu2019} for general graphs also hold for computing MIS in trees.

The MIS problem has also been studied in specific class of graphs and in particular for trees (e.g., \cite{BarenboimE10, LenzenW11, Barenboim2016, ghaffari16improved}). For example, on trees, barenboim and Elkin~\cite{BarenboimE10} showed that MIS can be solved in $O(\log n / \log\log n)$ deterministic rounds, and Ghaffari~\cite{ghaffari16improved} showed that it can be solved in $O(\sqrt{\log n})$ randomized rounds. As a function of $\Delta$ and $n$, MIS on trees can be solved in $O(\log\Delta + \log\log n/\log\log\log n)$ randomized rounds~\cite{Barenboim2016,ghaffari16improved}.

\paragraph{Ruling Sets.}
Ruling sets were introduced in the late 80s in \cite{Awerbuch89}, where the authors used them as a subroutine in their network decomposition algorithm. Since then, there have been several works that studied the complexity of ruling sets (see, e.g., \cite{Awerbuch89, SchneiderW10symmetry, Gfeller07, SEW13, BishtKP13, Barenboim2016, ghaffari16improved, KuhnMW18, KP12,PPPR017,HKN16}), and that used ruling sets as a subroutine for solving problems of interest (see, e.g., \cite{Barenboim2016, ghaffari16improved, GhaffariHKM18, ChangLP18}). \cite{Awerbuch89} in particular showed that a $(2,O(\log n))$-ruling set can be computed in time $O(\log n)$. Moreover, as shown by Schneider, Elkin, und Wattenhofer~\cite{SEW13}, the algorithm of \cite{Awerbuch89} can be adapted to deterministically compute, for general $\beta\geq 1$, a $(2,\beta)$-ruling set in time $O(\beta \Delta^{2/\beta} + \log^* n)$ (or in time $O(\beta C^{1/\beta})$ if an initial $C$-coloring is given).

 On the randomized side, we know that $(2,\log \log n)$-ruling sets
 can be computed in $O(\log \log n)$ randomized rounds, which
 follows by combining the works of Gfeller and Vicari~\cite{Gfeller07} and Schneider, Elkin, and Watternhofer~\cite{SEW13}. Further improvements were obtained in \cite{KP12,BishtKP13,Barenboim2016,ghaffari16improved}. In combination
 with the network decomposition result of \cite{Rozhon2020,GGR2020},
 the algorithm of Ghaffari~\cite{ghaffari16improved} implies that a
 $(2,\beta)$-ruling set can be computed in
 $O(\beta\log^{1/\beta} \Delta) + \poly(\log\log n)$ rounds in the
 randomized \LOCAL model. On trees, a $(2,\beta)$-ruling set can be
 computed deterministically in time
 $O\big(\frac{\log n}{\beta \log\log n}\big)$ by combining a coloring
 algorithm of Barenboim and Elkin~\cite{BarenboimE10} with the
 $O(\beta C^{1/\beta})$-round algorithm for a given $C$-coloring. With
 randomization, in trees, it is possible to compute a
 $(2,\beta)$-ruling set in time
 $O\big(\beta\log^{\frac{1}{\beta+1}}n\big)$ by combining techniques
 of \cite{Barenboim2016} with techniques of \cite{ghaffari16improved}.
 
 On the lower bound side, it was known since the late 80s and early 90s that computing a $(2, \beta)$-ruling set requires $\Omega(\log^* n)$ rounds up to some $\beta \in \Theta(\log^* n)$, even for randomized algorithms, and this directly follows from the MIS lower bounds of \cite{Linial1992} and \cite{Naor1991} for paths and rings and it therefore also holds on trees. The only previous $\omega(\log^* n)$ lower bounds for ruling sets were obtained in \cite{balliurules}. There, the authors showed that for sufficiently small $\beta$, any deterministic algorithm that computes a $(2,\beta)$-ruling set requires $\Omega\left(\min \left\{  \frac{\log \Delta}{\beta \log \log \Delta}  ,  \log_\Delta n \right\} \right)$ rounds, while any randomized algorithm requires $\Omega\left(\min \left\{  \frac{\log \Delta}{\beta \log \log \Delta}  ,  \log_\Delta \log n \right\} \right)$ rounds.

\paragraph{Arbdefective Colorings.}
The arbdefective coloring problem was introduced by Barenboim and Elkin~\cite{BarenboimE11}, and it has proved to be a very useful tool that is used in several state-of-the-art distributed algorithms for computing a proper coloring of a graph~\cite{BarenboimE11,barenboim16sublinear,fraigniaud16local,MausTonoyan20,GhaffariKuhn20,Kuhn20}. The aforementioned paper shows that in graphs of arboricity $a$, a $b$-arbdefective $O(a/b)$-coloring can be deterministically computed in $O(\min\{a,(a/b)^2\} \log n)$ rounds. The paper uses this to recursively decompose a graph into sparser subgraphs, which can then be colored more efficiently. Subsequently, Barenboim~\cite{barenboim16sublinear} gave an algorithm that computes an $O(\alpha)$-arbdefective $O(\Delta\log\Delta/\alpha)$-coloring in $O(\Delta\log^2\Delta/\alpha + \log^* n)$ rounds. Barenboim~\cite{barenboim16sublinear} and Fraigniaud, Heinrich, and Kosowski~\cite{fraigniaud16local} used this algorithm to obtain deterministic $(\Delta+1)$-list-coloring algorithms with a round complexity that is sublinear in $\Delta$. Barenboim, Elkin, and Goldenberg~\cite{BEG18} improved on the arbdefective coloring of \cite{barenboim16sublinear} by achieving an $\alpha$-arbdefective coloring with $O(\Delta/\alpha)$-coloring in time $O(\Delta/\alpha + \log^* n)$ (the same result was later also proved in different ways in \cite{GhaffariKuhn20} and \cite{Maus21}).

\section{Road Map}
\paragraph{Preliminaries.}
We start, in \Cref{sec:preliminaries}, by providing some preliminaries. In particular, we define the model of computing, we define the class of \emph{locally checkable} problems, and we formally define the round elimination framework.

\paragraph{\boldmath A Fixed Point for $\Delta$-coloring.}
In \Cref{sec:deltacoloring}, we prove an $\Omega(\log_\Delta n)$-round deterministic and $\Omega(\log_\Delta \log n)$-round randomized lower bound for the $\Delta$-coloring problem in the \LOCAL model. These results are already known from prior work \cite{chang16exponential,Brandt2016}. The novelty of our proof lies on the fact that we apply the round elimination technique directly to a more natural generalization of $\Delta$-coloring. The previous proofs are also based on round elimination, however the relaxation of $\Delta$-coloring (which is called sinkless coloring) in the proofs of \cite{chang16exponential,Brandt2016} is more related to the sinkless orientation problem than to $\Delta$-coloring. Note that the primary goal of \cite{Brandt2016} is to prove a lower bound for the sinkless orientation problem. While the proof of \cite{Brandt2016} also implies that $\Delta$-coloring requires $\Omega(\log_\Delta \log n)$ randomized rounds, the proof does not reveal enough of the structure of $\Delta$-coloring to be used to prove the lower bounds for MIS or ruling sets. As already described in \Cref{sec:intro}, we define a relaxation of $\Delta$-coloring, which we prove to be a non-trivial fixed point under the round elimination framework. This proof can be seen as a special case of the more general proof that we will show afterwards. We also explain why a fixed point proof for the $\Delta$-coloring problem is helpful in proving lower bounds for problems that are \emph{much easier} than the $\Delta$-coloring problem. In particular, we show that one reason for why it is hard to prove lower bounds for, e.g., ruling sets is that the $c$-coloring problem, for $c \le \Delta$, appears as a subproblem of ruling sets when applying the round elimination technique. The $c$-coloring problem does not behave nicely under the round elimination framework, and this makes also ruling sets to not behave nicely under this framework. Relaxing $c$-coloring to a non-trivial fixed point makes it behave much better, and the same then helps for ruling sets as well.

\paragraph{A Problem Family.}
In \Cref{sec:family} we formally define the family of problems for which we prove lower bounds. We also show that this family of problems is rich enough to contain relaxations of many interesting natural problems: $\Delta$-coloring, arbdefective colorings, maximal independent set, ruling sets, and arbdefective colored ruling sets. We also show that this family of problems contains even more problems: in general, it gives a way to interpolate between colorings and ruling sets.

\paragraph{Lower Bound in the Port Numbering Model.}
In \Cref{sec:sequence}, we are going to use the round elimination technique to prove lower bounds for the problems of the family in the port numbering model (a model that is strictly weaker than \LOCAL). In particular, we are going to show that some problems in this family are strictly easier than others. Then, by showing that there exists a sequence of $T$ problems, all non-zero round solvable, such that each problem in the sequence is strictly harder than the next problem in the sequence, we directly get that the first problem of the sequence requires at least $T$ rounds. This would give us a lower bound for the port numbering model, but what we really want is a lower bound for the \LOCAL model. An important technical detail that we will need in order to be able to lift such a result from the port numbering model to the \LOCAL model is the following: we need to be able to describe each problem of the sequence with a small number of labels. Hence, in this section we also show that the number of labels of each problem is small enough.

\paragraph{\boldmath Lifting Port Numbering Lower Bounds to the \LOCAL model.}
In \Cref{sec:lifting}, we show that a long sequence of problems satisfying the previously described requirements directly gives a lower bound for the \LOCAL model. Here we use known results, but we strengthen them a bit in order to support a larger number of labels, compared to the one used in previous works.

\paragraph{Lower Bound Results.}
In \Cref{sec:corollaries}, we formalize the connection between the
problems in the family and some natural problems. We show that the
generic lower bound obtained by combining the results of
\Cref{sec:sequence} and \Cref{sec:lifting} implies lower bounds for
many natural problems. For example, in this section we show that the
$(2,\beta)$-ruling set problem requires
$\Omega(\beta \Delta^{1/\beta})$ rounds, unless $\Delta$ (as a
function of $n$) is very large.

\paragraph{Upper Bounds.}
While the fixed point that we provide in \Cref{sec:deltacoloring} is a
relaxation of the $\Delta$-coloring problem, we prove in
\Cref{sec:corollaries} that it is also a relaxation of some variants
of arbdefective coloring. We show that this fixed point exactly
characterizes which variants of arbdefective coloring require
$\Omega(\log_\Delta n)$ rounds. In particular, in
\Cref{sec:algorithms} we show that all the variants of arbdefective
coloring for which we do not obtain a lower bound in
\Cref{sec:corollaries} can be solved in $O(\Delta + \log^* n)$. In
\Cref{sec:algorithms} we also provide tight upper bounds for
$\alpha$-arbdefective $c$-colored $\beta$-ruling sets.

\paragraph{Open Problems.}
We conclude, in \Cref{sec:open}, with some open problems. As we have
seen, fixed point based proofs for hard problems may help in proving
lower bounds for much easier problems, so one natural question is
whether such a fixed point based proof exists for all problems
that require $\Omega(\log_\Delta n)$ deterministic rounds in the
\LOCAL model.

\section{Preliminaries}\label{sec:preliminaries}

\subsection{\boldmath \LOCAL Model}
The model of distributed computing considered in this paper is the well-known \LOCAL model. In this model, the computation proceeds in synchronous rounds, where, at each round, nodes exchange messages with neighbors and perform some local computation. The size of the messages and the computational power of each node are not bounded. Each node in an $n$-node graph has a unique identifier (ID) in $\{1,\dotsc,\poly n\}$.  Initially, each node knows its own ID, its own degree, the maximum degree $\Delta$ in the graph, and the total number of nodes in the graph. Nodes execute the same distributed algorithm, and upon termination, each node outputs its local output (e.g., a color). Nodes correctly solve a distributed problem if and only if their local outputs form a correct global solution (e.g., a proper coloring of the graph). The complexity of an algorithm that solves a problem in the \LOCAL model is measured as the number of rounds it takes such that all nodes give their local output and terminate. Since the size of the messages is not bounded, each node can share with its neighbors all what it knows so far, and hence a $T$ round algorithm in this model can be seen as a mapping of $T$-hop neighborhoods into local outputs. In the randomized version of the \LOCAL model, in addition, each node has access to a stream of random bits. The randomized algorithms that we consider are Monte Carlo ones, that is, a randomized algorithm of complexity $T$ must always terminate within $T$ rounds and the result it produces must be correct with high probability, i.e., with probability at least $1 - 1/n$. The \LOCAL model is quite a strong model of computation, hence lower bounds in this model widely apply on other weaker models as well (such as, on the well-studied \CONGEST model of distributed computation, where the size of the messages is bounded by $O(\log n)$ bits).

\paragraph{Port Numbering Model.}
For technical reasons, in order to show our lower bounds for the \LOCAL model, we first show lower bounds for the weaker \emph{port numbering} (PN) model, and then we lift these results to the \LOCAL model. As in the \LOCAL model, in the PN model the size of the messages and the computational power of a node are unbounded. Differently from the \LOCAL model, in the PN model nodes do not have IDs. Instead, nodes are equipped with a port numbering, that is, each node has assigned a port numbering to all its incident edges. If we denote with $\deg(v)$ the degree of a node $v$, then each incident edge of node $v$ has a number, or a \emph{port}, in $\{1,\dotsc,\deg(v)\}$ assigned in an arbitrary way, such that for any two incident edges it holds that their port number is not the same. For technical reasons, we will also assume that edges are equipped with a port numbering in $\{1, 2\}$, that is, for each edge, we have an arbitrary assignment of port numbers to each of its endpoints, such that, endpoints of the same edge have different port numbers. This essentially results in an arbitrary consistent orientation of the edges. Notice that this technical detail makes the PN model only stronger, hence our lower bounds directly apply on the classical PN model as well.

\paragraph{Special Port Numbering.}
There have been some cases in previous works where, in order to show lower bounds, the authors have made use of a specific input given to the graph (note that this would only make the task of finding lower bounds potentially harder). For example, \cite{BBKOmis} showed lower bounds for MIS and out-degree dominating sets on trees in the case where we are given a $\Delta$-edge coloring in input. Hence, for the sake of being as general as possible, it is convenient to show that some of our main statements (such as \Cref{thm:lifting}) hold also when the ports satisfy some constraints. More precisely, consider a hypergraph where nodes have degree $\Delta$ and hyperedges have rank $\delta$. Let $C^n$ and $C^h$ be two sets, where $C^n \subseteq \{1,\ldots,\delta\}^\Delta$ and $C^h \subseteq \{1,\ldots,\Delta\}^\delta$. The set $C^n$ is the \emph{node port numbering} constraint, and restricts the possible ports that hyperedges incident to the nodes are allowed to have. The set $C^h$ is the \emph{hyperedge port numbering} constraint, and restricts the possible ports that nodes incident to the hyperedges are allowed to have. The node port constraint and the hyperedge port constraint are defined as follows.
\begin{itemize}
	\item Node port constraint $C^n$: Consider a node $v$, and let $h_i$ be the hyperedge connected to port $i$ of node $v$, for $1 \le i \le \Delta$. Let $p(h_i,v)$ be the hyperedge port that connects $h_i$ to $v$, then $(p(h_1,v),\dotsc,p(h_\Delta,v))$ must be contained in $C^n$.
	
	\item Hyperedge port constraint $C^h$: Consider a hyperedge $h$ and let $v_i$ be the node connected to port $i$ of hyperedge $h$, for $1 \le i \le \delta$. Let $p(v_i,h)$ be the node port that connects $v_i$ to $h$, then $(p(v_1,h),\dotsc,p(v_\delta,h))$ must be contained in $C^h$.
\end{itemize}
If $C^n$ and $C^h$ contain all possible tuples, we say that the port numbering is \emph{unconstrained}.

As already mentioned, some specific constrained port assignments can encode inputs, and an example is $\Delta$-edge coloring. Consider a graph $G$. We can define $C^n$ and $C^h$ in such a way that they imply a $\Delta$-edge coloring of $G$. Let $C^n$ be unconstrained, and let $C^h = \{(i,i) ~|~ 1 \le i \le \Delta\}$. In this case, if port $i$ of node $v$ is connected to an edge $e = \{u,v\}$, then the port of $u$ connected to $e$ is also $i$. Hence, we can interpret this port assignment as a coloring, where edge $e$ has color $i$. We may assume that also in the \LOCAL model we have a port numbering assignment, and this assignment may satisfy some constraints $C^{\mathrm{port}}$. Lower bounds obtained in this setting are even stronger.

\subsection{Black-White Formalism}
In this work, we prove lower bounds on graphs, and we provide a general theorem that is useful for lifting lower bounds from the port numbering model to the \LOCAL model. In order to make such a theorem as general as possible, we let it support also problems defined on \emph{hypergraphs}. In particular, we consider problems defined in the black-white formalism, that has already been used in the literature to formally define so-called \emph{locally checkable problems} on these kinds of graphs.

 In this formalism, a problem $\Pi = \Pi_{\Delta,\delta}$ is described by a triple $(\Sigma_{\Pi},\nodeconst_{\Pi},\edgeconst_{\Pi})$, where $\Sigma_{\Pi}$ is a set of labels, and $\nodeconst_{\Pi}$ and $\edgeconst_{\Pi}$ are sets of tuples of size $\Delta$ and $\delta$, respectively, where each element of the tuples is a label from $\Sigma_{\Pi}$. That is, $\nodeconst_{\Pi} \subseteq (\Sigma_{\Pi})^\Delta$, and $\edgeconst_{\Pi} \subseteq (\Sigma_{\Pi})^\delta$. $\nodeconst_{\Pi}$ and $\edgeconst_{\Pi}$ are called \emph{node constraint} and \emph{hyperedge constraint}, respectively. In some cases, we may also refer to them as \emph{white constraint} and \emph{black constraint}, respectively.

Given a hypergraph $G = (V,E)$, and the set of node-hyperedge pairs $P = \{(v,e) \in V \times E ~|~ v \in e \}$, solving $\Pi$ means assigning to each element of $P$ an element from $\Sigma_{\Pi}$ such that:
\begin{itemize}
	\item Let $v$ be any node of degree exactly $\Delta$. Let $L = (\ell_1,\ldots,\ell_\Delta)$ be the labels assigned to the node-hyperedge pairs incident to $v$. A permutation of $L$ must be in $\nodeconst_\Pi$.
	\item Let $e$ be any hyperedge of rank exactly $\delta$. Let $L = (\ell_1,\ldots,\ell_\delta)$ be the labels assigned to the node-hyperedge pairs incident to $e$. A permutation of $L$ must be in $\edgeconst_\Pi$.
\end{itemize}
The complexity of $\Pi$ is the minimum number of rounds required for nodes to produce a correct output.

A special case of the black-white formalism is called node-edge formalism, and is given by considering $\delta=2$. Note that, in this case, there must be two labels assigned to each edge, one for each endpoint. Since any hypergraph can be seen as a $2$-colored bipartite graph, and vice versa, then the black-white formalism allows us to define problems on $2$-colored bipartite graphs as well. In this case, we need to assign a label to each edge, and we have constraints on the labeling of the edges incident on white nodes of degree $\Delta$, and on edges incident on black nodes of degree $\delta$.

We use regular expressions to concisely describe the allowed configurations of a problem. For example, we can write $(\W \s [\X \Y] \s \Z) ~|~  \Z^3$ to denote the constraint $\{(\W,\X,\Z), (\W,\Y,\Z), (\Z, \Z, \Z)\}$. As a shorthand, we may also write this constraint as $\{\W \s \X \s \Z, \W \s \Y \s \Z, \Z \s \Z \s \Z\}$.
Recall that the order does not matter, and hence this constraint also allows the configuration $\W \s \Z \s \X$.
We call parts of regular expressions of the form $[\X \Y]$ \emph{disjunctions}, and while expressions of the form $\W \s [\X \Y] \s \Z$ denote sets of configurations, we may refer to them as \emph{condensed configurations}. If a specific configuration can be obtained by choosing a label in each disjunction, then we say that the configuration is \emph{contained} in the condensed configuration. For example, the configuration $\W \s \X \s \Z$ is contained in the condensed configuration $\W \s [\X \Y] \s \Z$.

The black-white formalism is powerful enough to be able to encode many problems for which the solution can be checked locally, i.e., in $O(1)$ rounds. In this formalism, only nodes of degree $\Delta$ and hyperedges of degree $\delta$ are constrained, while other nodes and hyperedges are not required to satisfy the constraints, but note that proving lower bounds in this restricted setting only makes the lower bounds stronger. We may refer to all problems that can be expressed in this formalism as \emph{locally checkable} problems. We now provide some examples of locally checkable problems.

\paragraph{\boldmath Example: Sinkless Orientation in $2$-Colored Graphs.}
Assume we have a $2$-colored graph. We require nodes of degree $\Delta$ to solve the sinkless orientation problem: nodes must orient their edges, and every node of degree $\Delta$ must have at least one outgoing edge. Nodes of smaller degree are unconstrained. We can define $\Sigma_{\Pi} = \{\B,\W\}$. Since the graph is $2$-colored, then every edge has a black and a white endpoint. An edge labeled $\B$ denotes that the edge is oriented towards the black node, while an edge labeled $\W$ denotes that the edge is oriented towards the white node.
Hence, we can express the requirement of having at least one outgoing edge as follows. The white constraint $\nodeconst_{\Pi}$ is defined as $\B \s [\B\W]^{\Delta-1}$, while the black constraint $\edgeconst_{\Pi}$ is defined as $\W \s [\B\W]^{\Delta-1}$. In other words, white nodes must have at least one edge labeled $\B$, meaning that it is oriented towards a black neighbor, and then all other edges are unconstrained. Similarly, black nodes must have at least one edge labeled $\W$.

\paragraph{Example: Maximal Independent Set.}
Assume we have a graph, where we require nodes of degree $\Delta$ to solve the maximal independent set problem. In order to encode this problem, one may try to use two labels, one for nodes in the set, and the other for nodes not in the set. Unfortunately, this is not possible, as with only two labels it is not possible to encode the maximality requirement~\cite{binary}. Hence, we use three labels, and we define $\Sigma_{\Pi} = \{\M, \P, \U\}$. We define the node constraint $\nodeconst_{\Pi}$ as $\M^\Delta ~|~ \P \s \U^{\Delta-1}$, and the edge constraint $\edgeconst_{\Pi}$ as $\M \s [\U \P] ~|~ \U \s \U$.

In other words, nodes can be part of the independent set and output $\M^\Delta$, or can be outside of the independent set, and in this case they have to prove that they have at least one neighbor in the independent set, by outputting $\P \s \U^{\Delta-1}$. The edge constraint only allows $\P$ to be paired with an $\M$, and hence nodes outputting $\P \s \U^{\Delta-1}$ must have at least one neighbor outputting $\M^\Delta$. Also, nodes of the independent set cannot be neighbors, because $\M \s \M$ is not allowed. Then, every other configuration is allowed by $\edgeconst_{\Pi}$.

\subsection{Round Elimination}
The round elimination technique works as follows. Assume we are given a problem $\Pi_1$, for which we want to prove a lower bound. We show that, given $\Pi_i$, we can construct a problem $\Pi_{i+1}$ that is at least one round easier than $\Pi_i$, unless $\Pi_i$ is already $0$ rounds solvable. If we are able to perform this operation for $T$ times, and prove that for all $i \le T$, the problem $\Pi_i$ is not $0$-rounds solvable, then this implies that $\Pi_1$ is at least $T-1$ rounds harder than $\Pi_T$. Since $\Pi_T$ is not $0$-rounds solvable, this implies that $\Pi_1$ requires at least $T$ rounds.

In order to define this sequence of problems, we exploit \cite[Theorem 4.3]{Brandt2019}, that gives a mechanical way to define $\Pi_{i+1}$ as a function of $\Pi_i$, such that if $\Pi_i$ requires exactly $T$ rounds, then $\Pi_{i+1}$ requires exactly $\max\{T-1,0\}$ rounds.
In its general form, this theorem works for any locally checkable problem defined in the black-white formalism, and given a problem $\Pi$ with complexity $T$, it first constructs an intermediate problem $\Pi'$, and then it constructs a problem $\Pi''$ that can be proved to have complexity $T-1$.

We denote by $\re(\cdot)$ the procedure that can be applied on a problem $\Pi$ to obtain the intermediate problem $\Pi'$, and by $\rere(\cdot)$ the procedure that can be applied on $\Pi'$ to obtain the problem $\Pi''$. Note that $\Pi'' = \rere(\re(\Pi))$.
In \cite{Brandt2019}, $\re(\Pi = (\Sigma_\Pi, \nodeconst_\Pi, \nodeconst_\Pi)) = \Pi' = (\Sigma_{\Pi'}, \nodeconst_{\Pi'}, \edgeconst_{\Pi'})$ is defined as follows. 
\begin{itemize}
	\item $\edgeconst_{\Pi'}$ is defined as follows. Let $C$ be the maximal set such that for all $(L_1, \ldots, L_\delta) \in C$ it holds that, for all $i$, $L_i \in 2^{\Sigma_{\Pi}} \setminus \{\emptyset\}$, and for all $(\ell_1, \ldots, \ell_\delta) \in L_1 \times \ldots \times L_\delta$, it holds that a permutation of $(\ell_1, \ldots, \ell_\delta)$ is in $\edgeconst_\Pi$. In this case, we say that $(L_1, \ldots, L_\delta)$ \emph{satisfies the universal quantifier}. The set $\edgeconst_{\Pi'}$ is obtained by removing all \emph{non-maximal} configurations from $C$, that is, all configurations $(L_1, \ldots, L_\delta) \in C$ such that there exists another configuration $(L'_1, \ldots, L'_\delta) \in C$ and a permutation $\phi$, such that $L_i \subseteq L'_{\phi(i)}$ for all $i$, and there exists at least one $i$ such that the inclusion is strict.
	\item $\Sigma_{\Pi'} \subseteq 2^{\Sigma_{\Pi}}$ contains all the sets that appear at least once in $\edgeconst_{\Pi'}$.
	\item $\nodeconst_{\Pi'}$ is defined as follows. It contains all configurations $(L_1, \ldots, L_\Delta)$ such that, for all $i$, $L_i \in \Sigma_{\Pi'}$, and there exists $(\ell_1, \ldots, \ell_\Delta) \in L_1 \times \ldots \times L_\Delta$, such that that a permutation of $(\ell_1, \ldots, \ell_\Delta)$ is in $\nodeconst_\Pi$. In this case, we say that $(L_1, \ldots, L_\Delta)$ \emph{satisfies the existential quantifier}.
\end{itemize}
We may refer to the computation of $\edgeconst_{\Pi'}$ as \emph{applying the universal quantifier}, and to the computation of $\nodeconst_{\Pi'}$ as \emph{applying the existential quantifier}. When applying the universal quantifier, we may use the term \emph{by maximality} to argue that, in order for a configuration to be maximal, some label must appear in a set of labels.
The hard part in computing $\Pi'$ is applying the universal quantifier. In fact, there is an easy way to compute $\nodeconst_{\Pi'}$, that is the following. Start from all the configurations allowed by $\nodeconst_{\Pi}$, and for each configuration add to $\nodeconst_{\Pi'}$ the condensed configuration obtained by replacing each label $\y$ by the disjunction of all label sets in $\Sigma_{\Pi'}$ that contain $\y$.

The operator $\rere(\cdot)$ is defined in the same way, but the roles of $\nodeconst$ and $\edgeconst$ are reversed, that is, the universal quantifier is applied on $\nodeconst_{\Pi'}$, and the existential quantifier is applied on $\edgeconst_{\Pi'}$.

In \cite[Theorem 4.3]{Brandt2019}, the following is proved.
\begin{theorem}[\cite{Brandt2019}, rephrased]\label{thm:rethm}
	Let $T > 0$. Consider a class $\mathcal{G}$ of hypergraphs with girth at least $2T+2$, and some locally checkable problem $\Pi$. Then, there exists an algorithm that solves problem $\Pi$ in $\mathcal{G}$ in $T$ rounds if and only if there exists an algorithm that solves $\rere(\re(\Pi))$ in $T-1$ rounds.
\end{theorem}

Note that, in order to prove lower bounds, when computing $\edgeconst_{\Pi'}$, it is not necessary to show that any possible choice over the sets allowed by $\edgeconst_{\Pi'}$ is a configuration allowed by $\edgeconst_{\Pi}$, but it is sufficient to show that there are no maximal configurations that satisfy the universal quantifier that are not in $\edgeconst_{\Pi'}$. In other words, when proving lower bounds, we can add allowed configurations to $\edgeconst_{\Pi'}$, as this is only making $\Pi'$ easier. Hence, we introduce the notion of a relaxation of a configuration (already introduced in~\cite{balliurules}).
\begin{definition}\label{def:noderelax}
	Consider two node configurations $\fA = \A_1 \s \dots \s \A_\Delta$ and $\fB = \B_1 \s \dots \s \B_\Delta$, where the $\A_i$ and $\B_i$ are sets of labels from some label space.
	We say that \emph{$\fA$ can be relaxed to $\fB$} if there exists a permutation $\rho \colon \{ 1, \dots, \Delta \} \to \{ 1, \dots, \Delta \}$ such that $\A_i \subseteq \B_{\rho(i)}$ for all $1 \leq i \leq \Delta$.
	For simplicity, if not indicated otherwise, we will assume that the sets in the two configurations are ordered such that $\rho$ is the identity function. We define relaxations of edge configurations analogously.
\end{definition}
Given two constraints $C$ and $C'$, if all the configurations allowed by $C$ can be relaxed to configurations in $C'$, we say that $C'$ is a \emph{relaxation} of $C$. Also, given two problems $\Pi_1 = (\Sigma_1,\nodeconst_1,\edgeconst_1)$ and  $\Pi_2 = (\Sigma_2,\nodeconst_2,\edgeconst_2)$, if $\nodeconst_2$ is a relaxation of $\nodeconst_1$ and $\edgeconst_2$ is a relaxation of $\edgeconst_1$, then we say that $\Pi_2$ is a relaxation of $\Pi_1$.

\paragraph{\boldmath Example: Sinkless Orientation in $2$-Colored Graphs.}
Recall that this problem can be expressed as follows.
\begin{align*}
 \nodeconst_{\Pi} &= \B \s [\B\W]^{\Delta-1}\\
 \edgeconst_{\Pi} &=\W \s [\B\W]^{\Delta-1}
\end{align*}
We start by computing $\edgeconst_{\Pi'}$. We need to satisfy the universal quantifier, and this implies that it must not be possible to pick $\B$ from $\Delta$ different sets, since we would obtain $\B^\Delta$, that is not allowed by $\edgeconst_{\Pi}$. But we can also notice that, as long as we are forced to pick $\W$ from at least one set, then everything else is allowed. Hence, the universal quantifier is satisfied for all configurations of the form $\{\W\} \s  [\{\W\}  \{\B\}  \{\W,\B\}]^{\Delta-1}$. If we discard non-maximal configurations, we obtain $\{\W\} \s  \{\W,\B\}^{\Delta-1}$. 

We can now rename the sets to make it easier to read the constraint, by using the following mapping:
\begin{align*}
\{\W\} &\rightarrow \W\\
\{\W,\B\} &\rightarrow \B
\end{align*}
 We obtain $\edgeconst_{\Pi'} = \W \s \B^{\Delta-1}$. In order to apply the existential quantifier on $\nodeconst_{\Pi}$, it is enough to replace each label with the disjunction of all the sets that contain that label. Note that $\B$ is only contained in the new label $\{\W,\B\}$, that is renamed to $\B$, while $\W$ is contained in both $\{\W\}$ and $\{\W,\B\}$, that are renamed to $\W$ and $\B$, respectively. Hence, $\nodeconst_{\Pi'} = \{\W,\B\} \s [\{\W\} \{\W,\B\}]]^{\Delta-1} = \B \s [\B \W]^{\Delta-1}$. Summarizing, we obtain the following:
\begin{align*}
	\Sigma_{\Pi'} &= \{\B,\W\} \\ 
	\edgeconst_{\Pi'} &= \W \s \B^{\Delta-1}\\
	\nodeconst_{\Pi'} &= \B \s [\B \W]^{\Delta-1}
\end{align*}

\paragraph{Relation Between Labels.}
Consider a problem $\Pi$ and let $C$ be one of its constraints (either node or hyperedge) containing tuples of size $d$. It can be useful to relate the labels of $\Pi$, according to how easier it is to use a label instead of another, according to $C$.

Consider two arbitrary labels $\X$, $\Y$ of $\Pi$. We say that $\X$ is \emph{at least as strong as} $\Y$ according to $C$, if the following holds. For all configurations $(\ell_1,\ldots,\ell_i,\ldots,\ell_d) \in C$ such that there exists an $i$ such that $\ell_i = \Y$, a permutation of the configuration $(\ell_1,\ldots,\X,\ldots,\ell_d)$ is in  $C$. That is, for all configurations in $C$ containing $\Y$, if we replace an arbitrary amount of $\Y$ with $\X$, then we obtain a configuration that is still in $C$. If the constraint is clear from the context, we may omit it and just say that $\X$ is at least as strong as $\Y$. If $\X$ is at least as strong as $\Y$, we also say that $\Y$ is \emph{at least as weak as} $\X$, or $\Y \le \X$. If $\X$ is at least as strong as $\Y$, but $\Y$ is not at least as strong as $\X$, then $\X$ is \emph{stronger} than $\Y$, and $\Y$ is \emph{weaker} than $\X$, that we may also denote as $\Y < \X$. Intuitively, stronger labels are easier to use, if we only consider the requirements of $C$, because we can always replace a weaker label by a stronger one. 

It can be useful to illustrate the strength of the labels by using a directed graph. The \emph{diagram} of $\Pi$ according to $C$ is a directed graph where nodes are labels of $\Pi$, and there is an edge $(\Y,\X)$ if and only if:
\begin{itemize}
	\item $\X \neq \Y$;
	\item $\X$ is at least as strong as $\Y$;
	\item there is no label $\Z$ different from $\X$ and $\Y$ that is stronger than $\Y$ and weaker than $\X$.
\end{itemize}
In other words, we put edges connecting weaker labels to stronger labels, but we omit relations that can be obtained by transitivity. Note that the obtained graph is not necessarily connected, but it is directed acyclic, unless $\Pi$ contains two labels of equal strength. We may refer to the diagram according to $\nodeconst$ as the \emph{node diagram}, and to the diagram according to $\edgeconst$ as the \emph{hyperedge diagram} (or \emph{edge diagram}, in the case of graphs). The \emph{successors} of some label $\ell$ w.r.t.\  some constraint $C$ are defined to be all labels that are at least as strong as $\ell$ according to $C$, or in other words all labels reachable from $\ell$ in the diagram according to $C$. 

When applying the round elimination technique, and in particular when performing the universal quantifier, it will be useful to consider the diagram of the constraint on which the universal quantifier is applied to. As we will see later, this diagram limits the possible labels appearing in the new constraint. 
As an example, \Cref{fig:mis} depicts the edge diagram of the MIS problem, that is, the diagram according to $\edgeconst_{\Pi} = \M \s [\U \P] ~|~ \U \s \U$. There is an edge from $\P$ to $\U$ because the only edge configuration containing $\P$ is $\M \s \P$, and if we replace $\P$ with $\U$, we obtain $\M \s \U$, which is also allowed by $\edgeconst_{\Pi}$. Conversely, there is no edge from $\U$ to $\P$, because $\U \s \U$ is allowed, but by replacing an arbitrary amount of $\U$ with $\P$ we obtain configurations not in $\edgeconst_{\Pi}$.

\begin{figure}[h]
	\centering
	\includegraphics[width=0.14\textwidth]{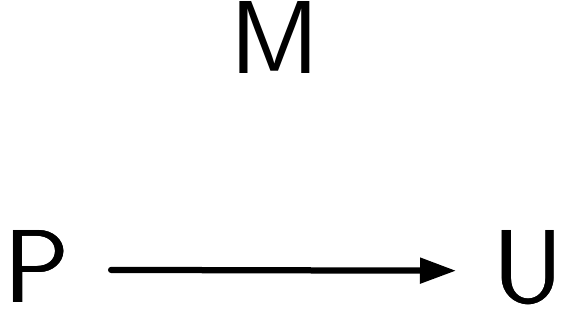}
	\caption{The edge diagram of the MIS problem: $\U$ is stronger than $\P$, and there is no relation between labels $\M$ and $\P$, and between $\M$ and $\U$.}
	\label{fig:mis}
\end{figure}

\paragraph{Additional Notation.}
We now introduce some additional notation that can be useful to concisely describe a set of labels.
Given a set $\{ \ell_1, \dots, \ell_p \} \subseteq \Sigma_\Pi$, we denote by $\gen{\ell_1, \dots, \ell_p}$ the set that contains all labels $\ell$ in $\Sigma_\Pi$ satisfying that there exists an $i$ such that $\ell$ is at least as strong as $\ell_i$. In other words, $\gen{\ell_1, \dots, \ell_p}$ contains the labels $\ell_1, \ldots, \ell_p$, plus all labels that are reachable from at least one of them in the diagram of $\Pi$ (i.e., all the successors of $\ell_1, \ldots, \ell_p$).

In order to properly define $\gen{}$, we need to specify to which constraint the strength relation refers to. If the argument of $\gen{}$ contains labels of some problem $\Pi$ on which we are applying the operator $\re(\cdot)$, then $\gen{}$ refers to its hyperedge constraint, while if the argument of $\gen{}$ contains labels of an intermediate problem $\Pi'$, that is, a problem on which we are applying the operator $\rere(\cdot)$, then $\gen{}$ refers to its node constraint. In other words, we always consider the constraint on which we apply the universal quantifier. That is, if we start from $\Pi$ and compute $\Pi' = \re(\Pi)$, then $\gen{}$ refers to $\edgeconst_{\Pi}$, while when computing $\Pi'' = \rere(\Pi')$, then  $\gen{}$ refers to $\nodeconst_{\Pi'}$.

When computing $\rere(\re(\Pi))$ it will be useful to use expressions such as $\gen{\gen{\ell}}$, where $\ell$ is a label of $\Pi$. This expression represents a set of sets of labels from $\Sigma_\Pi$, where $\gen{\ell}$ is computed according to $\edgeconst_{\Pi}$, and $\gen{\gen{\ell}}$ is then taken according to $\nodeconst_{\Pi'}$.

A set $S = \{ \ell_1, \dots, \ell_p \} \subseteq \Sigma_\Pi$ is called \emph{right-closed} if $S = \gen{\ell_1, \dots, \ell_p}$. That is, $S$ is right-closed if and only if for each label $\ell_i$ contained in $S$ also all successors of $\ell_i$ in the diagram (taken w.r.t.\ the constraint on which the universal quantifier is applied to) are contained in $S$.
In \cite{balliurules}, the following statements have been proved.
\begin{observation}[\cite{balliurules}]\label{obs:rcs}
	Consider an arbitrary collection of labels $\A_1, \dots, \A_p \in \Sigma_\Pi$.
	If $\{ \A_1, \dots, \A_p \} \in \Sigma_{\re(\Pi)}$, then the set $\{ \A_1, \dots, \A_p \}$ is right-closed (w.r.t.\ $\edgeconst_{\Pi}$).
	If $\{ \A_1, \dots, \A_p \} \in \Sigma_{\rere(\Pi')}$, then the set $\{ \A_1, \dots, \A_p \}$ is right-closed (w.r.t.\ $\nodeconst_{\Pi'}$).
\end{observation}
\begin{observation}[\cite{balliurules}]\label{obs:subsetarrow}
	Let $\U, \W \in  \Sigma_{\re(\Pi)}$ be two sets satisfying $\U \subseteq \W$.
	Then $\W$ is at least as strong as $\U$ according to $\nodeconst_{\re(\Pi)}$.
	In particular, for any label $\A \in \Sigma_\Pi$ such that $\gen{\A} \in \Sigma_{\re(\Pi)}$, every set $\X \in \Sigma_{\re(\Pi)}$ containing $\A$ is contained in $\gen{\gen{\A}}$.
	
	Analogous statements hold for $\rere(\cdot)$ instead of $\re(\cdot)$, by considering the label strength according to $\edgeconst_{\rere(\cdot)}$.
\end{observation}
For a set $S = \{ \ell_1, \dots, \ell_p \} \subseteq \Sigma_\Pi$ of labels, we denote by $\dis(S)$ the disjunction $[\ell_1 \dots \ell_p]$. For instance, $\dis(\gen{\ell})$ is the disjunction of all labels that are at least as strong as $\ell$.

\section{$\Delta$-Coloring Fixed Point and Its Applications}\label{sec:deltacoloring}
In this section, we define a problem $\Pi_\Delta$ that is at least as easy as the $\Delta$-coloring problem, and show that $\Pi_\Delta$ is a non-trivial fixed point under the round elimination framework. It is known that a non-trivial fixed point directly implies an $\Omega(\log_\Delta n)$ lower bound for deterministic algorithms, and an $\Omega(\log_\Delta \log n)$ lower bound for randomized algorithms, in the \LOCAL model, but we show this formally in \Cref{sec:lifting}.
 Then, we also describe the relation of $\Pi_\Delta$ with the maximal independent set problem.

The problem $\Pi_\Delta$ is defined as follows. Consider the set of colors $\ccs = \{1, \ldots, \Delta\}$. The label set $\Sigma_{\Pi_{\Delta}}$ of $\Pi_\Delta$ is defined as $\Sigma_{\Pi_{\Delta}} = \{ \ell(\ccc) ~|~ \ccc \in 2^\ccs \}$. The node constraint $\nodeconst_{\Pi_{\Delta}}$ of $\Pi_\Delta$ contains the following allowed configurations:
\[
	\ell(\ccc)^{\Delta - x} \s \ell(\emptyset)^x, \text{ for all $\ccc \in 2^\ccs \setminus \{\emptyset\}$, where $x = |\ccc|-1$. }
\]
The edge constraint $\edgeconst_{\Pi_{\Delta}}$ of $\Pi_\Delta$ contains the following allowed configurations:
\[
	\ell(\ccc_1) \s \ell(\ccc_2), \text{for all $\ccc_1,\ccc_2 \in 2^\ccs$ such that $\ccc_1 \cap \ccc_2 = \emptyset$.}
\]
In \Cref{ssec:fixpoint}, we will prove that $\Pi_\Delta$ is a fixed point. Note that $\Pi_\Delta$ is not $0$-round solvable, since all configurations allowed by the node constraint contain at least one label $\L=\ell(\ccc)\neq \ell(\emptyset)$, and $\L\s\L\notin\edgeconst_{\Pi_{\Delta}}$. We now provide an explicit example of $\Pi_\Delta$ where $\Delta=3$, and we give some intuition behind its definition.

\paragraph{\boldmath Example: $\Pi_3$.}
Consider the case where $\Delta=3$, and let us use letters to denote colors, that is, $\ccs = \{\A,\B,\C\}$. To simplify the notation, we rename each label $\ell(\{x_1, \ldots, x_k\})$ corresponding to a non-empty set of colors as $\mybox{x_1 \ldots x_k}$. Also, we rename $\ell(\emptyset)$ as $\X$. For example, the label $\ell(\{\A,\B\})$ becomes $\mybox{\A\B}$.
The problem $\Pi_3$ is defined as follows:

\begin{equation*}
	\begin{aligned}
		\begin{split}
			\nodeconst_{\Pi_{3}}\text{:}\\
			&\mybox{\A} &\s& \mybox{\A} &\s& \mybox{\A} \\
			&\mybox{\B} &\s& \mybox{\B} &\s& \mybox{\B} \\
			&\mybox{\C} &\s& \mybox{\C} &\s& \mybox{\C} \\
			&\mybox{\A\B} &\s& \mybox{\A\B} &\s& \X \\
			&\mybox{\A\C} &\s& \mybox{\A\C} &\s& \X \\
			&\mybox{\B\C} &\s& \mybox{\B\C} &\s& \X \\
			&\mybox{\A\B\C} &\s& \X &\s& \X \\\\
		\end{split}
		\qquad
		\begin{split}
			\edgeconst_{\Pi_{3}}\text{:}\\
			&\X &\s& [\X\mybox{\A}\mybox{\B}\mybox{\C}\mybox{\A\B}\mybox{\B\C}\mybox{\A\C}\mybox{\A\B\C}] \\
			&\mybox{\A} &\s& [\X\mybox{\B}\mybox{\C}\mybox{\B\C}]\\
			&\mybox{\B} &\s& [\X\mybox{\A}\mybox{\C}\mybox{\A\C}]\\
			&\mybox{\C} &\s& [\X\mybox{\A}\mybox{\B}\mybox{\A\B}]\\
			&\mybox{\A\B} &\s& [\X\mybox{\C}]\\
			&\mybox{\A\C} &\s& [\X\mybox{\B}]\\
			&\mybox{\B\C} &\s& [\X\mybox{\A}]\\
			&\mybox{\A\B\C} &\s& \X\\
		\end{split}
	\end{aligned}
\end{equation*}
The intuition behind this problem is the following. We start from the standard definition of $3$-coloring, and we add allowed configurations that ``reward'' nodes that have many neighbors of the same color. For example, if a node $v$ outputs the configuration $\mybox{\A\B} \s \mybox{\A\B} \s \X$, it implies that $v$ has at least $2$ neighbors that output neither $\A \s \A \s \A$ nor $\B \s \B \s \B$. We can think of $v$ as being colored \emph{both $\A$ and $\B$}. Node $v$ is rewarded for this, and it is allowed to output $\X$ on one of its ports. Label $\X$ can be seen as a wildcard: a label that is compatible with everything. In particular, this means that $\Pi_3$ can be seen as a variant of arbdefective coloring, where nodes that control many colors are allowed to have a larger arbdefect. 

\paragraph{Relation with MIS}
It is known that if a problem can be relaxed to a problem that is a fixed point under the round elimination framework, then it requires $\Omega(\log_\Delta n)$ rounds in the \LOCAL model for deterministic algorithms, and $\Omega(\log_\Delta \log n)$ rounds for randomized algorithms.
Hence, it is perhaps surprising that a fixed point can be helpful to prove lower bounds for problems that can be solved in $O(f(\Delta) + \log^*n)$ for some function $f$, that is, problems that are much easier. We now provide the intuition behind this phenomenon, by considering the case of MIS (similar observations can be made for ruling sets and for arbdefective colored ruling sets).

Already in \cite{balliurules}, it has been shown that, when applying the round elimination technique on the MIS problem for many times, we get a problem that looks like the following:
\begin{itemize}
	\item There is a part of the problem that can be naturally described, that is, after doing $k$ steps of round elimination, the problem becomes a relaxation of the $k$-coloring problem.
	\item There is a part of the problem that cannot be naturally described, that seems to be just an artifact of the round elimination technique. In particular, after $k$ steps, the obtained problem also contains what one would obtain by applying the round elimination technique on the $(k-1)$-coloring problem.
\end{itemize}
Unfortunately, by applying the round elimination technique on coloring problems, we usually obtain problems whose description grows exponentially \emph{at each step} of round elimination. This makes it infeasible to apply the technique for more than just a few steps. Hence, we need to find a way to avoid this exponential growth, that is, find a relaxation of $k$-coloring such that if we apply the technique, we don't get a problem that is much larger. This is exactly what a fixed point for $k$-coloring (for $k \le \Delta$) gives us, since, by applying round elimination on this problem, we obtain the problem itself. Hence, the idea that we use to prove a lower bound for MIS (and other problems) is the following:
\begin{itemize}
	\item By applying the round elimination technique on MIS, at each step, the number of colors grows by $1$, and intuitively this means that the MIS problem requires $\Omega(\Delta)$, since $k$-coloring is a hard problem, for $k\le\Delta$. To prove this, we need to show that the ``unnatural'' part of the problem does not help to solve the problem faster.
	\item The unnatural part would grow exponentially at each step, but we relax it in the same way as we do in the $k$-coloring fixed point. In this way, at step $k+1$, we obtain a $(k+1)$-coloring, plus something that can be relaxed to the configurations allowed in the fixed point of the $(k+1)$-coloring problem. This allows us to \emph{concisely} describe the problem obtained at each step.
\end{itemize}
Hence, the problem obtained by applying the round elimination technique on MIS for $k$ times can be described as follows: there is a $k$-coloring component, plus something that makes the colors grow by $1$ at each new step, plus something related to applying the round elimination technique on $k'$-coloring for $k' < k$. In order to obtain our linear-in-$\Delta$ lower bound for MIS, we essentially embed, into the description of the intermediate problems, a proof that $k$-coloring is hard.

\subsection{\boldmath A Fixed Point For $\Delta$-Coloring}\label{ssec:fixpoint}
In order to simplify the proofs, instead of showing that $\Pi_\Delta$ is a non-trivial fixed point, we do the following. We define a problem $\Pi$, that is essentially equivalent to $\Pi_\Delta$, where the only difference is that the node constraint of $\Pi$ is defined in a slightly different way. Then, we show that $\Pi$ is a non-trivial fixed point. Hence, problem $\Pi$ is not $0$-round solvable, and it is a relaxation of the $\Delta$-coloring problem, i.e., a solution for $\Delta$-coloring is also a solution for $\Pi$, but the reverse does not necessarily hold. Then, we show that, by applying two round elimination steps on $\Pi$, we obtain the same problem $\Pi$, implying that $\Pi$ is a non-trivial fixed point. 
 
\paragraph{\boldmath Problem $\Pi$.}
As in the case of $\Pi_\Delta$, the label set is defined as $\Sigma_{\Pi}=\{\ell(\ccc)~|~\ccc\in 2^\ccs\}$, where $\ccs=\{1,\dotsc,\Delta\}$. The edge constraint $\edgeconst_{\Pi}$ is also defined in the same way. It contains all the configurations of the following form:
\[
\ell(\ccc_1) \s \ell(\ccc_2), \text{for all $\ccc_1,\ccc_2 \in 2^\ccs$ such that $\ccc_1 \cap \ccc_2 = \emptyset$.}
\]
The node constraint $\nodeconst_{\Pi}$ is defined slightly differently. It contains all the configurations of the following form:
\begin{align*}
	\ell(\ccc_1) \s \ell(\ccc_2)\s \ldots \s \ell(\ccc_\Delta),& \text{ such that }\\
	\exists k \in \{1,\dotsc,\Delta\} \text{ and } 1\le i_1< i_2 < \dotsc <i_k\le \Delta,& \text{ such that }\\
	|\ccc_{i_1}\cap \ccc_{i_2}\cap\dotsc\cap \ccc_{i_k}|\ge \Delta -k+1.&
\end{align*}
On the one hand, notice that $\nodeconst_{\Pi}$ allows strictly more configurations than the node constraint of $\Pi_\Delta$. On the other hand, given a solution of $\Pi_\Delta$, nodes can solve $\Pi$ in $0$ rounds by outputting the intersection of the $k$ sets on $k$ ports, and the empty set on the remaining ones. Hence, $\Pi$ and $\Pi_\Delta$ are equivalent. We start by observing some properties of the configurations of $\Pi$.

\begin{observation}\label{obs:supersets}
	If $\ell(\ccc_1) \s \ell(\ccc_2)\s \ldots \s \ell(\ccc_\Delta)$ is a configuration in $\nodeconst_{\Pi}$, then also $\ell(\ccc'_1) \s \ell(\ccc'_2)\s \ldots \s \ell(\ccc'_\Delta)$, where, for all $1\le j\le\Delta$, $\ccc_j\subseteq \ccc'_j$, is a configuration in $\nodeconst_{\Pi}$.
\end{observation}

\begin{observation}\label{obs:subsets}
	If $\ell(\ccc_1) \s \ell(\ccc_2)$ is a configuration in $\edgeconst_{\Pi}$, then also $\ell(\ccc'_1) \s \ell(\ccc'_2)$, where, for all $1\le j\le2$, $\ccc'_j\subseteq \ccc_j$, is a configuration in $\edgeconst_{\Pi}$.
\end{observation}

\paragraph{\boldmath Labels of $\Pi'$.} Let $\Pi'=\re(\Pi)$. By definition, $\Sigma_{\Pi'} \subseteq 2^{\Sigma_{\Pi}}$, that is, $\Sigma_{\Pi'}$ consists of a set of sets of labels in $\Sigma_{\Pi}$. Before characterizing the node and edge constraints of $\Pi'$, we perform a renaming of the labels of problem $\Pi'$, in the following way. Each label set $\{\ell(\ccc_1),\ldots,\ell(\ccc_h)\}$ is renamed to label $\ell(\ccc)=\ell(\bigcup_{1\le j\le h}\ccc_j)$, that is, a label that consists of sets of colors is mapped to a label that consists of the union of the sets of colors. This renaming gives a label-set that contains labels of the form $\ell(\ccc)$, where $\ccc\in 2^\ccs$ and $\ccs=\{1,\dotsc,\Delta\}$. Hence, after this renaming, $\Sigma_{\Pi'}=\Sigma_{\Pi}$. Note that different original labels, such that the union of their sets gives the same set, are mapped to the same label.

\paragraph{\boldmath Edge Constraint of $\Pi'$.}
We now characterize the edge constraint $\edgeconst_{\Pi'}$ of $\Pi'$. 
First of all, by definition of the round elimination framework, $\edgeconst_{\Pi'}$ contains all configurations $\L_1 \s \L_2$ such that, for any choice in $\L_1$, and for any choice in $\L_2$, we get a configuration that is in $\edgeconst_{\Pi}$. Hence, under the renaming that we performed, $\edgeconst_{\Pi'}$ must contain at least all the configurations in $\edgeconst_{\Pi}$. We claim that $\edgeconst_{\Pi'}$ does not contain any additional allowed configuration.

Suppose, for a contradiction, that there exists a configuration $\ell(\ccc_1)\s \ell(\ccc_2)$ in $\edgeconst_{\Pi'}$ such that $\ccc_1\cap \ccc_2 \neq \emptyset$. This means that, before the renaming, there exists a configuration $\L_1\s\L_2$ in $\edgeconst_{\Pi'}$ such that $\L_1$ contains a set $S_1$ containing $\ell(\ccc'_1)$, and $\L_2$ contains a set $S_2$ containing $\ell(\ccc'_2)$, such that $\ccc'_1\cap\ccc'_2\neq \emptyset$. Hence, we get that there exists a choice resulting in the configuration $\ell(\ccc'_1)\s \ell(\ccc'_2)$, where $\ccc'_1\cap \ccc'_2 \neq \emptyset$, which does not satisfy any configuration in $\edgeconst_{\Pi}$, which is a contradiction. Hence, $\edgeconst_{\Pi'}=\edgeconst_{\Pi}$.

\paragraph{\boldmath Node Constraint of $\Pi'$.}
We now characterize the node constraint $\nodeconst_{\Pi'}$ of $\Pi'$. 
As shown in \Cref{sec:preliminaries}, the node constraint $\nodeconst_{\Pi'}$ can be constructed in the following way: take all the configurations allowed by the node constraint $\nodeconst_{\Pi}$, and replace each label $\ell(\ccc)$ with the disjunction of all the labels of $\Pi'$ that, before the renaming, contain $\ell(\ccc)$. Note that, after the renaming, the labels that contain $\ell(\ccc)$ are exactly all the labels $\ell(\ccc')$ such that $\ccc \subseteq \ccc'$. Hence, from each configuration $\ell(\ccc_1) \s \ldots \s \ell(\ccc_\Delta)$ in $\nodeconst_{\Pi}$, we obtain the configuration itself, plus all the configurations $\ell(\ccc'_1) \s \ldots \s \ell(\ccc'_\Delta)$ such that for all $1\le j\le\Delta$, $\ccc_j \subseteq \ccc'_j$. By Observation \ref{obs:supersets}, these configurations are also in $\nodeconst_{\Pi}$. Hence $\nodeconst_{\Pi'}=\nodeconst_{\Pi}$.

\paragraph{\boldmath Labels of $\Pi''$.} 
In order to prove that $\Pi$ is a fixed point, we show that problem  $\Pi''$, defined as $\Pi''=\re(\Pi')=\rere(\re(\Pi))$, after a specific renaming, is the same as problem $\Pi$.
By definition, $\Sigma_{\Pi''} \subseteq 2^{\Sigma_{\Pi'}}$. We rename the labels in $\Sigma_{\Pi''}$ as follows. Each label set $\{\ell(\ccc_1),\ldots,\ell(\ccc_h)\}$ is renamed to label $\ell(\ccc)=\ell(\bigcap_{1\le j\le h}\ccc_j)$, that is, a label that consists of sets of colors is mapped to a label that consists of the intersection of the sets of colors. This renaming gives a label-set that contains labels of the form $\ell(\ccc)$, where $\ccc\in 2^\ccs$ and $\ccs=\{1,\dotsc,\Delta\}$. Hence, after this renaming, $\Sigma_{\Pi''}=\Sigma_{\Pi}$. Note that different original labels, such that the intersection of their sets gives the same set, are mapped to the same label. We  prove that, under this renaming, $\Pi'' = \Pi$.

\paragraph{\boldmath Node Constraint of $\Pi''$}
We now characterize the node constraint $\nodeconst_{\Pi''}$ of $\Pi''$, and show that it is equal to $\nodeconst_{\Pi}$.
First of all, by definition of the round elimination framework, before renaming, $\nodeconst_{\Pi''}$ contains all and only configurations $\L_1 \s \ldots \s \L_\Delta$, such that, for any choice $(\ell(\ccc_1),\dotsc,\ell(\ccc_\Delta)) \in \L_1\times\ldots\times\L_\Delta$, we get a configuration that is in $\nodeconst_{\Pi'}=\nodeconst_{\Pi}$. We claim that, after the proposed renaming, $\nodeconst_{\Pi''} = \nodeconst_{\Pi}$.

We prove our claim by performing multiple partial renamings.
Let $\L_1 \s \ldots \s \L_\Delta$ be an arbitrary configuration in $\nodeconst_{\Pi''}$. Consider an arbitrary label $\L_i = \{\ell(\ccc_1), \ldots, \ell(\ccc_k)\}$ in the configuration. Create the label $\L'_i = \{\ell(\ccc_1 \cap \ccc_2), \ell(\ccc_3), \ldots, \ell(\ccc_k)\}$, that is, replace $\ell(\ccc_{1}), \ell(\ccc_{2})$ by their intersection. In the following lemma, we prove that also the configuration obtained by replacing $\L_i$ with $\L'_i$ satisfies that, for any choice over this configuration, we get a configuration that is in $\nodeconst_{\Pi'}=\nodeconst_{\Pi}$. If we repeat this process many times, until every set $\L_i$ contains a single label, then we get the exact same result that we would have obtained by directly replacing each set of labels by their intersection. But this implies that, under the proposed renaming, we get exactly the same configurations that appear in $\nodeconst_{\Pi'}$, and hence our claim follows.

	\begin{lemma}\label{lem:intersection}
		Let $\L_1\s\L_2\s\dotsc\s\L_\Delta$ be a configuration in $\nodeconst_{\Pi''}$. Let $\L_t = \{\ell(\ccc_t^1), \ldots, \ell(\ccc_t^h)\}$
		be an arbitrary label in the configuration, and let $\ell(\ccc')=\ell(\ccc_t^{1} \cap \ccc_t^{2})$. If we replace the label set $\L_t$ with $\L'_t = \{\ell(\ccc'), \ell(\ccc_t^3), \ldots, \ell(\ccc_t^h)\}$, we get a configuration that is in $\nodeconst_{\Pi''}$.
	\end{lemma}
	\begin{proof}
		W.l.o.g., suppose $t=\Delta$. Let $\L_{\Delta}^{1\cap 2} = \{ \ell(\ccc')\}$, and let $\L_{\Delta}^{3+} =  \{\ell(\ccc_\Delta^3), \ldots, \ell(\ccc_\Delta^h)\}$. 
		Since by assumption $\L_1\s\dotsc\s \L_{\Delta-1} \s \L_{\Delta}^{3+}$ already satisfies the universal quantifier of the round elimination framework, we only need to prove that $\L_1\s\dotsc\s \L_{\Delta-1}\s \L_{\Delta}^{1\cap 2}$ is a configuration in $\nodeconst_{\Pi''}$, that is, for any choice $(\ell(\ccc_1),\dotsc,\ell(\ccc_{\Delta-1}))\in\L_1\times\dotsc\times\L_{\Delta-1}$, we get that $\ell(\ccc_1)\s\dotsc\s\ell(\ccc_{\Delta -1})\s\ell(\ccc')$ is a configuration in $\nodeconst_{\Pi}$. 		
		We prove our claim by contradiction. Hence, 
		assume that there exists a choice $C' = \ell(\ccc_1)\s\dotsc\s\ell(\ccc_{\Delta -1})\s\ell(\ccc')$ that is not a configuration in $\nodeconst_{\Pi}$.
		
		Let $\ccc_\Delta = \ccc^1_\Delta$.
		By assumption, for any choice in $\L_1\s\L_2\s\dotsc\s\L_\Delta$, we get a configuration in $\nodeconst_{\Pi}$. Hence, by the definition of $\nodeconst_{\Pi}$, we get the following:
		\begin{align*}
			\exists k \in \{1,\dotsc,\Delta\} \text{ and } 1\le i_1< \dotsc <i_k\le \Delta,& \text{ such that }\\
			|S_1 = \ccc_{i_1}\cap \dotsc\cap \ccc_{i_k}|\ge \Delta -k+1.&
		\end{align*}
		Similarly, if we set $\ccc_\Delta = \ccc^2_\Delta$, we get the following:
		\begin{align*}
			\exists k' \in \{1,\dotsc,\Delta\} \text{ and } 1\le i'_1 < \dotsc <i'_k\le \Delta,& \text{ such that }\\
			|S_2 = \ccc_{i'_1}\cap \dotsc\cap \ccc_{i'_{k'}}|\ge \Delta -k'+1.&
		\end{align*}	
		Let $s_1$ be the number of sets $\ccc_i$, $i\in\{1,\dotsc,\Delta-1\}$, such that $S_1 \subseteq \ccc_i$ and $S_2 \not\subseteq \ccc_i$. Similarly, let $s_2$ be the number of sets $\ccc_{i'}$, $i'\in\{1,\dotsc,\Delta-1\}$, such that $S_2 \subseteq \ccc_{i'}$ and $S_1 \not\subseteq \ccc_{i'}$. Also, let $s$ be the number of labels $\ccc_j$, $j\in\{1,\dotsc,\Delta-1\}$, such that $S_1\cup S_2 \subseteq \ccc_j$.
		
		Consider the configuration $C = \ell(\ccc_1)\s\dotsc\s\ell(\ccc_{\Delta -1})\s\ell(\ccc^1_\Delta)$. Depending on $\ccc^1_\Delta$, the number of sets that are supersets of $S_1$ is either $s+s_1$, or $s+s_1+1$. Hence, $s+s_1+1$ is an upper bound on the number of sets that are supersets of $S_1$. Since, by assumption, $C$ is a valid configuration of $\nodeconst_\Pi$, then $|S_1| \ge \Delta - (s+s_1+1) +1$. 
		By applying the same reasoning on $S_2$, we obtain the following:
		\begin{align}
			|S_1| &\ge \Delta - (s+ s_1), \text{ and }\label{eq:s1}\\
			|S_2| &\ge \Delta - (s+ s_2).\label{eq:s2}
		\end{align}
		Also, consider the configuration $C'$. If $|S_1 \cup S_2|\ge \Delta - s + 1$, then, since we can pick $s$ sets that are supersets of $S_1 \cup S_2$, we get that $C'$ satisfies the requirements for being a configuration in $\nodeconst_\Pi$, reaching a contradiction. Hence, we obtain the following:
		\begin{align}
			|S_1 \cup S_2 | &\le \Delta - s.\label{eq:union}
		\end{align}
		By combining \Cref{eq:s1}, \Cref{eq:s2}, and \Cref{eq:union}, we get that:
		\begin{align}\label{eq:intersection1}
			|S_1\cap S_2| &=|S_1| + |S_2| - |S_1\cup S_2|\\
			&\ge (\Delta -s -s_1) + (\Delta -s -s_2) - (\Delta -s)\nonumber\\
			&= \Delta -s -s_1 - s_2\nonumber.
		\end{align}
		If $S_1 \nsubseteq \ccc^1_\Delta$, then all the sets of $C$ that are supersets of $S_1$ are part of $\ccc_1,\ldots,\ccc_{\Delta-1}$, but this implies that also in $C'$ there are still $k$ sets that are supersets of $S_1$, and hence that $C'$ satisfies the requirements for being a configuration in $\nodeconst_\Pi$, reaching a contradiction. Hence, $S_1 \subseteq \ccc^1_\Delta$. For the same reason, $S_2 \subseteq \ccc^2_\Delta$. We thus get that $S_1 \cap S_2 \subseteq \ccc'$.
		If  $|S_1 \cap S_2|\ge \Delta - (s + s_1 + s_2 + 1 )+ 1$, then, since from $C'$ we can pick $s+s_1+s_2+1$ sets that are supersets of $S_1 \cap S_2$, we get that that $C'$ satisfies the requirements for being a configuration in $\nodeconst_\Pi$, reaching a contradiction. Hence, we obtain the following:
		\begin{align}
			|S_1 \cap S_2| < \Delta - s - s_1 - s_2.\label{eq:intersection2}
		\end{align}
		But \Cref{eq:intersection1} is in contradiction with \Cref{eq:intersection2}. Hence, $C'$ is a valid configuration.
	\end{proof}

\paragraph{\boldmath Edge Constraint of $\Pi''$.}
We now characterize the edge constraint $\edgeconst_{\Pi''}$ of $\Pi''$. 
As shown in \Cref{sec:preliminaries}, the edge constraint $\edgeconst_{\Pi''}$ can be constructed in the following way: take all the configurations allowed by the edge constraint $\edgeconst_{\Pi'} = \edgeconst_{\Pi}$, and replace each label $\ell(\ccc)$ with the disjunction of all the labels of $\Pi''$ that, before the renaming, contain $\ell(\ccc)$. Note that, after the renaming, the labels that contain $\ell(\ccc)$ are exactly all the labels $\ell(\ccc')$ such that $\ccc' \subseteq \ccc$. Hence, from each configuration $\ell(\ccc_1) \s \ell(\ccc_2)$ in $\edgeconst_{\Pi}$, we obtain the configuration itself, plus all the configurations $\ell(\ccc'_1) \s \ell(\ccc'_2)$ such that for all $1\le j\le 2$, $\ccc'_j \subseteq \ccc_j$. By Observation \ref{obs:subsets}, these configurations are also in $\edgeconst_{\Pi}$. Hence $\edgeconst_{\Pi''}=\edgeconst_{\Pi}$.

\section{The Problem Family}\label{sec:family}
In this section, we describe the family of problems that we use to prove our lower bounds. We first start from the formal definition, and then we provide some intuition. Finally, we explain the relation between these problems and some other natural problems.

\subsection{Problem Definition}
Consider a vector $z$ of $\beta+1$ non-negative integers $z_0,\ldots, z_\beta$ and let $\len(z) = \beta$. We define $|z| = z_0 + \dots + z_\beta = |\ccs|$, where $\ccs = \{ C_{i,j} ~|~ 0 \le i \le \beta, 1 \le j \le z_i \}$ is a set of colors. Given a color $C_{i,j} \in \ccs$, we define $\level(C_{i,j}) = i$, and given a set of colors $\ccc \in 2^\ccs \setminus \{\emptyset\}$, we define $\level(\ccc) = \max\{ i ~|~ C_{i,j} \in \ccc \}$, where $\level(\emptyset) = -1$. Finally, let $\prefix(z)$ be the inclusive prefix sum of $z$, that is, $\prefix(z)_i = \sum_{j\le i} z_j$. Assume that $|z|\le\Delta$, then problem $\Pi_{\Delta}(z)$ is defined as follows.

\paragraph{Labels.}
The label set of problem $\Pi_{\Delta}(z)$ is defined as 
$\Sigma_{\Delta}(z) = \mathcal{P} \cup \mathcal{S}\cup \{\X\}$, where:
\begin{itemize}
	\item $\mathcal{P}=\{ \P_i, \U_i ~|~ 1 \le i \le \beta \}$;
	\item $\mathcal{S}=\{ \ell(\ccc) ~|~ \ccc \in 2^\ccs \setminus  \{\emptyset\}\}$.
\end{itemize}

\paragraph{Node Constraint.}
The node constraint $\nodeconst_{\Delta}(z)$ of problem $\Pi_{\Delta}(z)$ contains the following allowed configurations for the nodes:
\begin{itemize}
	\item $\ell(\ccc)^{\Delta-x} \s \X^x$, for each $\ccc \in 2^\ccs \setminus  \{\emptyset\}$, where $x = |\ccc|-1$;
	\item $\P_i \s \U^{\Delta-1}_i$, for each $1 \le i \le \beta$. 
\end{itemize}

\paragraph{Edge Constraint.}
The edge constraint $\edgeconst_{\Delta}(z)$ of problem $\Pi_{\Delta}(z)$ contains the following allowed configurations for the edges:
\begin{itemize}
	\item $\ell(\ccc_1) \s \ell(\ccc_2)$, for all $\ccc_1, \ccc_2 \in 2^\ccs \setminus  \{\emptyset\}$ such that $\ccc_1 \cap\ccc_2 = \emptyset$;
	\item $\U_i \s \U_j$, for each $1 \le i,j \le \beta$;
	\item $\U_i \s \P_j$, for each $1 \le i < j \le \beta$;
	\item $\U_i \s \ell(\ccc)$, for each $\ccc \in 2^\ccs \setminus  \{\emptyset\}$ and $1 \le i \le \beta$; 
	\item $\P_i \s \ell(\ccc)$, for all $\ccc \in 2^\ccs \setminus  \{\emptyset\}$ such that $\level(\ccc) < i$;
	\item $\X \s \L$, for each $\L \in \Sigma_{\Delta}(z)$.
\end{itemize}

\paragraph{Edge Diagram.}
In order to compute $\re(\Pi_\Delta(z))$, it will be helpful to know the relation between the labels of $\Pi_\Delta(z)$ according to the edge constraint $\edgeconst_{\Delta}(z)$. The following relation derives directly from the definition of $\edgeconst_{\Delta}(z)$. \Cref{fig:diagram} shows an example, for $\Pi_\Delta([1,2,1])$.
\begin{observation}\label{lem:edgediagpi}
	The labels of $\Pi_\Delta(z)$ satisfy all and only the following strength relations w.r.t.\ $\edgeconst_{\Delta}(z)$. 
	\begin{itemize}
		\item $\U_i < \U_j$, if $j < i$.
		\item $\P_i < \P_j$, if $i < j$.
		\item $\P_i < \U_j$, for all $i,j$.
		\item $\P_i < \ell(\ccc)$, if, for all $C \in \ccc$, $\level(C) \ge i$.
		\item $\ell(\ccc) < \ell(\ccc')$, if $\ccc' \subseteq \ccc$.
		\item $\ell(\ccc) < \U_i$, if $\level(\ccc) \ge i$.
		\item $\L < \X$, for all $\L \in \Sigma_\Delta(z)$.
	\end{itemize}
\end{observation}
\begin{figure}[h]
	\centering
	\includegraphics[width=0.5\textwidth]{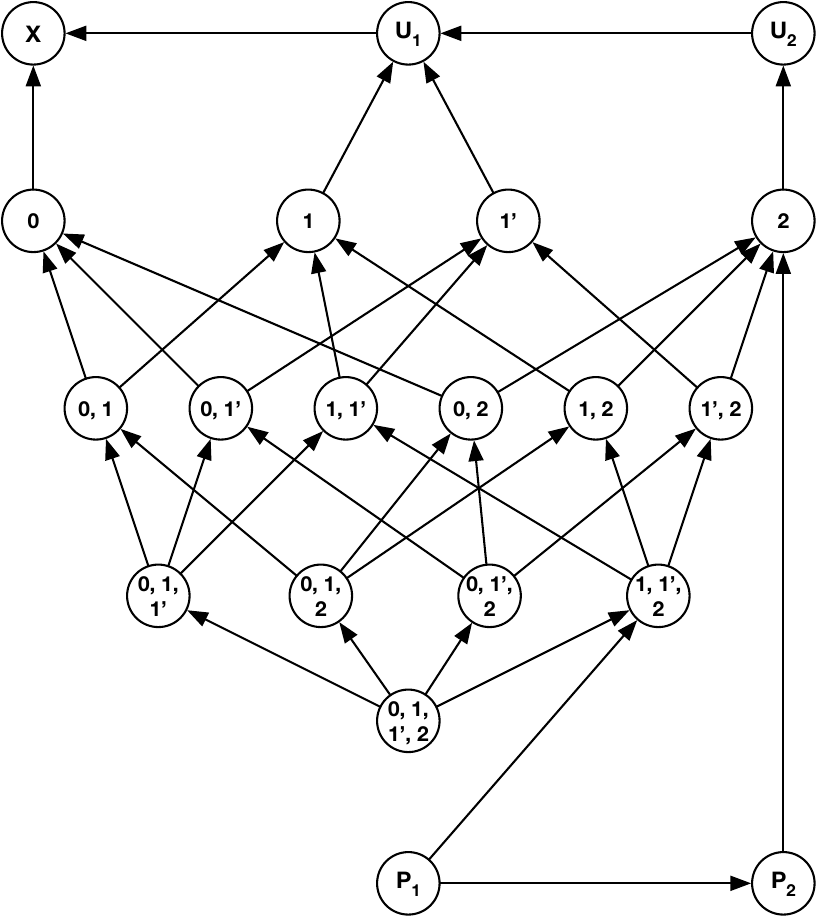}
	\caption{The edge diagram of $\Pi_\Delta([1,2,1])$. Color $0$ is in level $0$, colors $1$ and $1'$ are in level $1$, and color $2$ is in level $2$. For the sake of readability, we omit $\ell(\cdot)$.}
	\label{fig:diagram}
\end{figure}

\subsection{Intuition behind These Problems}
The problem definition contains two main components, a part related to coloring, and a part related to pointers. We start by describing how to interpret the coloring part. Then, we will discuss how pointers can interact with each other, and with the colors. 

\paragraph{Coloring Part.}
Let $\ccs$ be a color space and assume that it is split into $\beta+1$ levels. We can think of color $C_{i,j}$ to be the $j$-th color of level $i$. Nodes can either be colored (nodes outputting a configuration containing $\ell(\ccc)$, for some $\ccc$), or uncolored (nodes outputting a configuration containing $\P_i$, for some $i$). Colored nodes may even have multiple colors assigned ($\ccc$ is a non-empty \emph{set} of colors). Colored nodes do not need to satisfy the requirements of a proper coloring, but it is enough to satisfy the requirements of an arbdefective coloring, that is, they can mark $x$ edges as outgoing (using label $\X$), and the nodes reachable through those edges are allowed to have the same color as them. If a node is colored with multiple colors, then, ignoring edges that are marked with $\X$ by at least one endpoint, the coloring constraints must be satisfied for \emph{all} colors assigned to that node: if a node is red \emph{and} blue, then no neighbor must be red, and no neighbor must be blue (in fact, the edge constraint requires that if two nodes output $\ell(\ccc_1) \s \ell(\ccc_2)$, then $\ccc_1$ and $\ccc_2$ must be disjoint).
The arbdefect $x$ allowed for a node depends on the number of colors assigned to that node: if a node has a single color assigned, then $x=0$, while if a node has two colors assigned, then $x=1$, and so on. In general, the node constraint allows $x = |\ccc|-1$.
 The idea here is that if a node manages to get more than $1$ color assigned, then it is rewarded by being allowed to have a larger arbdefect.
 Note that the node constraint requires to have \emph{exactly} $x$ ports marked $\X$, but this is just a technicality: if a node wants to output $\X$ on \emph{at most} $x$ ports, then it can always change the output on some arbitrary ports to $\X$ and obtain \emph{exactly} $x$ ports marked $\X$, since $\X$ is compatible with any other label.

\paragraph{Pointers Part.}
Uncolored nodes must output pointers, and there are $\beta$ types of pointers ($\P_1, \ldots, \P_\beta$). We can think of pointer $\P_i$ to be of level $i$. The edge constraint allows pointers to point to either:
\begin{itemize}
	\item nodes outputting a pointer of strictly lower level ($\P_j$ is compatible with $\U_i$ only if $i<j$), or
	\item nodes outputting a color set $\ccc$ such that all the colors of $\ccc$ are of lower levels ($\P_i \s \ell(\ccc)$ is allowed only if $\level(\ccc) < i$, that is, the maximum level of the colors in $\ccc$ is strictly less than $i$).
\end{itemize}
Except for these constraints, all the other ports of nodes outputting pointers are unconstrained: $\U_i$ can be neighbor with all other $\U$ labels, with all color sets, and with $\X$.

Hence, pointers can be interpreted as follows. Pointers must form chains of length at most $\beta$, and these chains must reach colored nodes. Not all types of chains are allowed: if the chain terminates by pointing to a color of level $i$, then the chain must have length at most $\beta - i$.

\subsection{Relation with Natural Problems}
We now provide some informal examples of what problems our family can capture. We make this more formal in \Cref{sec:corollaries}.
\paragraph{$\Delta$-coloring.}
Consider the case where $\beta=0$ (that is, we do not have pointers at all), and $\ccs = \{1, \ldots, \Delta\}$. We can see that the obtained problem is a relaxation of the $\Delta$-coloring problem.

\paragraph{Arbdefective Colorings.}
When $\beta=0$, some variants of arbdefective coloring are captured as well. For example, if the color space contains $\Delta$ colors, then there exist $\lfloor \Delta/2 \rfloor$ disjoint sets of colors of size $2$. If nodes output sets of size $2$, then they are allowed to have arbdefect $1$. Hence, by letting each node of color $0 \le i < \lfloor \Delta/2 \rfloor$ output the configuration $\ell(\{C_{0,2i+1}, C_{0,2i+2}\})^{\Delta-1} \s \X$, we get a problem that is a relaxation of the $1$-arbdefective  $\lfloor \Delta/2 \rfloor$-coloring problem.

\paragraph{Ruling Sets.}
Consider the problem $\Pi$ in the family where the color space contains a single color, in level $0$, and there are $\beta$ possible pointers. In a solution for the $(2,\beta)$-ruling set problem, nodes in the set form an independent set, and nodes not in the set are at distance at most $\beta$ from nodes in the set. Hence, nodes can spend $\beta$ rounds to solve $\Pi$: nodes in the set output the single color of level $0$, and node not in the set output pointer chains based on their distance from the nodes in the set.

\paragraph{\boldmath $\alpha$-Arbdefective $c$-Colored $\beta$-Ruling Sets.}
This problem is a variant of ruling sets, where the nodes in the set do not need to satisfy the independence requirement, but they are only required to be $c$-colored with arbdefect at most $\alpha$. Note that, if we set $\alpha = 0$ and $c=1$, we get the $(2,\beta)$-ruling set problem.
Hence, we can consider a problem of the family where the color space contains $c(1+\alpha)$ colors, all in level $0$, that allows us to have $c$ sets of colors with arbdefect $\alpha$, and a solution for $\alpha$-arbdefective $c$-colored $\beta$-ruling sets can be converted in $\beta$ rounds into a solution for the problem in the family.

\paragraph{Other Problems.}
In all the examples made until now, all colors of the color space were in level $0$, but our family of problems is more general. For example, our family contains a relaxation of the following problem.
Nodes can either be colored or uncolored. There are two groups of colors. One of the two groups is allowed to have some arbdefect, while the other group is allowed some different arbdefect. 
Uncolored nodes must either be at distance at most $3$ from a node colored colored with a color from the first group, or at distance at most $2$ from a node colored with a color from the second group.

\section{Lower Bounds}\label{sec:sequence}
In this section, we prove some properties about the problems in our family. In particular, we show that there exists a long sequence of problems, all contained in the family, such that the first problem of the sequence is a relaxation of a problem for which we want to prove a lower bound, and then for each problem in the sequence it holds that:
\begin{itemize}
	\item it is not $0$-round solvable;
	\item it is at least $1$ round harder than the next problem in the sequence;
	\item it can be described with a small amount of labels.
\end{itemize}
By combining the first two properties, we obtain that a sequence of $T$ problems implies a lower bound of $T$ rounds for the first problem in the port numbering model. As we will see later, in order to lift this lower bound to the \LOCAL model, we also need the third property: the number of labels used to describe the problems in the family should stay small.

In order to construct this sequence, we start in \Cref{ssec:relation} by computing $\Pi' = \re(\Pi_{\Delta}(z))$, and prove some properties about it. While it is easy to check that,  if $|z| \le \Delta$, then $\Pi_\Delta(z)$ contains at most $2^\Delta + 2 \len(z) \le 2^\Delta(1 + \len(z))$ labels, we also make sure that the intermediate problem $\Pi'$ can also be described with that amount of labels.

We continue, in \Cref{sec:keytec}, where we show the core part of our proof, that is, we apply the universal quantifier on $\re(\Pi_\Delta(z))$.

Then, in \Cref{sec:relaxrerere}, we show that $\rere(\re(\Pi_\Delta(z)))$ can be relaxed to $\Pi_\Delta(\prefix(z))$, implying that  if $\Pi_{\Delta}(z)$ can be solved in $T$ rounds, then $\Pi_{\Delta}(\prefix(z))$ can be solved in at most $T-1$ rounds.

By combining \Cref{ssec:relation}, \Cref{sec:keytec}, and \Cref{sec:relaxrerere}, we obtain the following lemma.
\begin{lemma}\label{lem:pn1step}
	Let $z' = \prefix(z)$, and assume that $|z'| \le \Delta$ if $\len(z') = 0$, and $|z'| \le \Delta-1$ otherwise.
	Let $\Pi'=\re(\Pi_{\Delta}(z))$. Then, $\Pi_{\Delta}(z')$ is a relaxation of $\rere(\Pi')$. The number of labels of $\Pi_{\Delta}(z)$, $\Pi'$, and $\Pi_{\Delta}(z')$ is upper bounded by $2^\Delta(1 + \len(z))$.
\end{lemma}

We continue, in \Cref{ssec:zero}, by showing that, unless $|z|$ is too large, $\Pi_{\Delta}(z)$ cannot be solved in $0$ rounds in the port numbering model. In this way, we obtain all the requirements mentioned before.

We conclude, in \Cref{ssec:sequence}, by combining all the ingredients and give a lower bound on the length of the problem sequence.

\subsection{\boldmath Computing $\re(\Pi_\Delta(z))$}\label{ssec:relation}
Recall the definition of $\Pi_\Delta(z)$ (see \Cref{sec:family}), and define $\U_0 := \ell(\emptyset) := \X$.
The following lemma computes $\re(\Pi_{\Delta}(z))$.

\begin{lemma}\label{lem:repi}
	The label set of problem $\re(\Pi_{\Delta}(z))$ is
	\begin{align*}
		\Sigma_{\re} := \{ \gen{\P_i, \ell(\ccc)} \mid 1 \leq i \leq \beta, \ccc \in 2^{\ccs}, \level(\ccs \setminus \ccc) \leq i - 1 \} \cup \{ \gen{\U_i, \ell(\ccc)} \mid \ccc \in 2^{\ccs}, \level(\ccc) \leq i \leq \beta \}.
	\end{align*}
	Its edge constraint $\edgeconst_{\re}$ consists of the following configurations:
	\begin{align*}
		&\gen{\P_i, \ell(\ccc)} \s \gen{\U_{i-1}, \ell(\ccs \setminus \ccc)} &\textrm{ for all $1 \leq i \leq \beta$ and $\ccc \in 2^{\ccs}$ satisfying } \level(\ccs \setminus \ccc) \leq i - 1,& \textrm{ and}\\
		&\gen{\U_{\beta}, \ell(\ccc)} \s \gen{\U_{\beta}, \ell(\ccs \setminus \ccc)} &\textrm{ for all $\ccc \in 2^{\ccs}$}.&
	\end{align*}
	Its node constraint $\nodeconst_{\re}$ consists of all configurations $\B_1 \s \dots \s \B_{\Delta}$ with labels $\B_i \in \Sigma_{\re}$ such that there exists some configuration $ \A_1 \dots \A_{\Delta} \in \nodeconst_{\Delta}(z)$ with $\A_i \in \B_i$ for all $1 \leq i \leq \Delta$. 
\end{lemma}
\begin{proof}
	Consider an arbitrary configuration $\B_1 \s \B_2$ in the edge constraint $\edgeconst_{\re}$ of $\re(\Pi_{\Delta}(z))$.
	We will show that $\B_1 \s \B_2$ is of the form stated in the lemma.
	We consider two cases.
	
	Consider first the case that at least one of $\B_1$ and $\B_2$ (which, by the definition of $\re(\cdot)$ are sets of labels from $\Sigma_{\Delta}(z)$) contains a label from $\{ \P_1, \dots, \P_\Delta\}$.
	Let $i$ be the smallest index for which $\P_i$ is contained in $\B_1$ or $\B_2$, and assume w.l.o.g.\ that $\P_i \in \B_1$.
	Let $\ccc \in 2^\ccs$  be the set of all colors $C \in \ccs$ such that there exists a set $\ccc'$ containing $C$ for which $\ell(\ccc') \in \B_1$ (possibly $\ccc = \emptyset$). In other words, $\ccc$ contains all the colors that appear in $\B_1$.
	Observe that $\B_2$ cannot contain any label $\ell(\ccc'')$ such that $\ccc'' \cap \ccc \neq \emptyset$ since otherwise there exist $\ccc', \ccc'' \in 2^\ccs$ such that $\ell(\ccc') \in \B_1$, $\ell(\ccc'') \in \B_2$, and $\ccc' \cap \ccc'' \neq \emptyset$, which is a contradiction to the definition of $\re(\Pi_{\Delta}(z))$ (as $\ccc' \cap \ccc'' \neq \emptyset$ implies that $\ell(\ccc') \s \ell(\ccc'') \notin \edgeconst_\Delta(z)$).
	Similarly, for any $1 \leq j \leq \beta$, set $\B_2$ does not contain $\P_j$ since $\B_1$ contains $\P_i$, and $\P_i \s \P_j \notin \edgeconst_\Delta(z)$.
	Since, for any $\A \in \Sigma_{\Delta}(z) \setminus \left( \{ \ell(\ccc'') \mid \ccc'' \cap \ccc \neq \emptyset \} \cup \{ \P_1, \dots, \P_\Delta \} \right)$, we have $\ell(\ccc) \s \A \in \edgeconst_\Delta(z)$, it follows that $\ell(\ccc) \in \B_1$, by maximality.
	Hence, $\B_1$ contains $\ell(\ccc)$ and $\P_i$, but not any $\P_j$ with $j < i$ (by the definition of $i$) and not any $\ell(\ccc'')$ with $\ccc'' \nsubseteq \ccc$ (by the definition of $\ccc$).
	 By \Cref{lem:edgediagpi} and the fact that by \Cref{obs:rcs} $\B_1$ is right-closed w.r.t.\ $\edgeconst_{\Delta}(z)$, it follows that the labels in $\B_1$ are exactly the labels from $\Sigma_\Delta(z)$ that are successors of $\ell(\ccc)$ (w.r.t.\ $\edgeconst_{\Delta}(z)$) or $\P_i$, i.e., we have $\B_1 = \gen{\P_i, \ell(\ccc)}$.
	Note further that the right-closedness of $\B_1$, together with \Cref{lem:edgediagpi} and the fact that $\P_i \in \B_1$, implies that $\B_1$ contains $\ell(\{C\})$ for all $C \in \ccs$ satisfying $\level(C) \geq i$.
	Hence, $\level(\ccs \setminus \ccc) \leq i - 1$, by the definition of $\ccc$.

	By the definition of $\re(\cdot)$ (in particular by maximality), the set $\B_2$ consists of precisely those labels $\A \in \Sigma_{\Delta}(z)$ that satisfy $\A \s \A' \in \edgeconst_{\re}$ for all $\A' \in \B_1$.
	By going through the labels in $\Sigma_{\Delta}(z)$ and checking whether this property holds (by going through the list of edge configurations of $\Pi_\Delta(z)$), we see that $\B_2 = \{ \ell(\ccc') \mid \ccc' \in 2^\ccs\setminus\{\emptyset\}, \ccc' \subseteq \ccs \setminus \ccc \} \cup \{ \U_1, \dots, \U_{i-1} \} \cup \{ \X \}$.
	Note that, for any nonempty $\ccc' \subseteq \ccs \setminus \ccc$ and any $i \leq j \leq \beta$, we have $\ell(\ccc') \s \P_j \in \edgeconst_{\re}$ since $\level(\ccs \setminus \ccc) \leq i - 1$, as observed above.
	Now, by \Cref{lem:edgediagpi}, the fact that by \Cref{obs:rcs} $\B_2$ is right-closed  w.r.t.\ $\edgeconst_{\Delta}(z)$, and the fact that $\level(\ccs \setminus \ccc) \leq i - 1$, we obtain that $\B_2 = \gen{\U_{i-1}, \ell(\ccs \setminus \ccc)}$.
	It follows that $\B_1 \s \B_2 = \gen{\P_i, \ell(\ccc)} \s \gen{\U_{i-1}, \ell(\ccs \setminus \ccc)}$, where $\level(\ccs \setminus \ccc) \leq i - 1$, as desired.

	Consider now the second case, namely that both $\B_1$ and $\B_2$ do not contain any label from $\{ \P_1, \dots, \P_\Delta\}$.
	Again, let $\ccc \in 2^\ccs$  be the set of all colors $C \in \ccs$ such that there exists a set $\ccc'$ containing $C$ for which $\ell(\ccc') \in \B_1$.
	Similarly to above, we obtain that $\B_2$ does not contain any label $\ell(\ccc'')$ such that $\ccc'' \cap \ccc \neq \emptyset$, and that $\ell(\ccc) \in \B_1$.
	Also, observe that, by maximality, both $\B_1$ and $\B_2$ contain $\U_\beta$ since $\U_\beta \s \A \in \edgeconst_{\re}$ for any $\A \in \Sigma_\Delta(z) \setminus \{\P_1, \dots, \P_\Delta\}$, and $\B_1$ and $\B_2$ do not contain any label from $\{ \P_1, \dots, \P_\Delta\}$.
	
	Since $\B_1$ contains $\ell(\ccc)$ and $\U_\beta$, but not any $\P_j$ and not any $\ell(\ccc'')$ with $\ccc'' \nsubseteq \ccc$, we obtain $\B_1 = \gen{\U_{\beta}, \ell(\ccc)}$, by \Cref{lem:edgediagpi} and the fact that by \Cref{obs:rcs} $\B_1$ is right-closed  w.r.t.\ $\edgeconst_{\Delta}(z)$.
	Moreover, with an analogous argumentation for computing $\B_2$ as in the first case, we obtain $\B_2 = \gen{\U_{\beta}, \ell(\ccs \setminus \ccc)}$.
	Hence, also in this second case $\B_1 \s \B_2$ is of a form given in the lemma statement.
	
	It is straightforward to verify that the edge configurations given in the lemma statement are indeed contained in $\edgeconst_{\re}$, by checking that any two labels picked from the two sets in the considered configuration indeed form an edge configuration of $\Pi_\Delta(z)$, and that for both sets it holds that any label that is not in the set cannot be added to the set without violating this property.
	Hence, the edge constraint $\edgeconst_{\re}$ is precisely as given in the lemma.

	Moreover, by collecting all labels (i.e., sets of labels from $\Pi_{\Delta}(z)$) that appear in $\edgeconst_{\re}$, we obtain precisely the label set $\Sigma_{\re}$ given in the lemma.
	Finally, the statement about the node constraint $\nodeconst_{\re}$ follows from the definition of $\re(\cdot)$.
\end{proof}

We now prove a property about the labels of $\re(\Pi_{\Delta}(z))$.
\begin{lemma}\label{lem:abc123}
	Assume $|z| \le \Delta$ if $\len(z) = 0$, and $|z| \le \Delta-1$ otherwise.
	For all $\A, \B \in \Sigma_{\re}$ it holds that, if $\A \le \B$ according to $\nodeconst_{\re}$, then $\A \subseteq \B$.
\end{lemma}
\begin{proof}
	Assume that $\A \not\subseteq \B$, and let $a \in \Sigma_\Delta(z)$ be a label contained in $\A \setminus \B$. We prove that $\A \not\le \B$. For a contradiction, assume that $\A \le \B$.
	
	Consider some configuration $a \s \ell_2 \s \ldots \s \ell_\Delta \in \nodeconst_\Delta(z)$. Note that, by \Cref{lem:repi} and by \Cref{lem:edgediagpi}, for each label $\ell \in \Sigma_{\Delta}(z)$, there is the label $\gen{\ell}$ in $\Sigma_{\re}$.
	Hence, by \Cref{lem:repi}, the configuration  $\A \s \gen{\ell_2} \s \ldots \s \gen{\ell_\Delta}$ is in $\nodeconst_{\re}$. Since $\A \le \B$, then the configuration $\B \s \gen{\ell_2} \s \ldots \s \gen{\ell_\Delta}$ must also be in $\nodeconst_{\re}$.
	
	By the definition of $\nodeconst_{\re}$, there must exist a configuration $b \s \ell'_2 \s \ldots \ell'_\Delta \in \nodeconst_\Delta(z)$ that is in $\B \times \gen{\ell_2} \times \ldots \times \gen{\ell_\Delta}$. Since, by \Cref{obs:rcs}, $\B$ is right-closed, then the fact that $a \notin \B$ implies that also all labels that are at least as weak as $a$ (according to $\edgeconst_{\Delta}(z)$) are not in $\B$. This implies that $b$ is not as least as weak as $a$ according to $\edgeconst_{\Delta}(z)$, since otherwise $a$ would be a successor of $b$, and $b$ being in $\B$ would imply $a$ being in $\B$. Also, since $\ell'_i \in \gen{\ell_i}$ for all $2 \le i \le \Delta$, then we also have that $\ell_i \le \ell'_i$.
	
	Summarizing, there must be two configurations, $a \s \ell_2 \s \ldots \s \ell_\Delta$ and $b \s \ell'_2 \s \ldots \ell'_\Delta$, that are both contained in $\nodeconst_\Delta(z)$, satisfying that, according to $\edgeconst_{\Delta}(z)$, $b \not\le a$ (implying that either $a < b$, or that $a$ and $b$ are uncomparable), and $\ell_i \le \ell'_i$ for all  $2 \le i \le \Delta$.
	We show that all configurations of $\nodeconst_\Delta(z)$ do not satisfy these requirements, reaching a contradiction.
	
	In $\nodeconst_\Delta(z)$ there are two types of configurations: $\ell(\ccc)^{\Delta-x} \s \X^x$ (color configurations), and $\P_i \s \U^{\Delta-1}_i$ (pointer configurations). We consider three cases separately. 
	\begin{itemize}
		\item Both configurations are color configurations, $C_1 = \ell(\ccc_1)^{\Delta-x_1} \s \X^{x_1}$ and $C_2 = \ell(\ccc_2)^{\Delta-x_2} \s \X^{x_2}$. Assume w.l.o.g.\ that $x_1 \le x_2$. By assumption on $z$, $x_2 \le \Delta - 1$. Also by the definition of $\nodeconst_\Delta(z)$, $|\ccc_1| \le |\ccc_2|$. 
		 In this case, $a$ cannot be $\X$ (since $b$ must be a label strictly stronger than $a$ or uncomparable with $a$, but $\X$ is at least as strong as any other label). Hence, $a$ must be either $\ell(\ccc_1)$ or $\ell(\ccc_2)$.
		 
		 If $a = \ell(\ccc_1)$ and $b = \ell(\ccc_2)$, then by \Cref{lem:edgediagpi}, since $b \not \le a$, then $\ccc_1 \not\subseteq \ccc_2$, that, together with the fact that $|\ccc_1| \le |\ccc_2|$ implies that both $\ccc_1$ and $\ccc_2$ are not the set containing all colors. Hence, $x_1 \le x_2 \le \Delta-2$, implying that there exists some $i$ such that $\ell'_i = \ell(\ccc_2)$. Then, $\ell_i$ must be either $\X$ or $\ell(\ccc_1)$, but both of them are not as weak as $\ell(\ccc_2)$ because of \Cref{lem:edgediagpi} and the fact that $|\ccc_1| \le |\ccc_2|$.  
		 
		 If $a = \ell(\ccc_1)$ and $b = \X$, then there must exist some index  $2 \le i \le \Delta$ such that $\ell'_i = \ell(\ccc_2)$, but $\ell(\ccc_2)$ is not a successor of either $\ell(\ccc_1)$ or $\X$, by \Cref{lem:edgediagpi} and the fact that  $|\ccc_1| \le |\ccc_2|$.
		 
		 If $a = \ell(\ccc_2)$ and $b = \X$, then since $x_2 \ge x_1$, then there must exist some index $2 \le i \le \Delta$ such that $\ell_i = \X$ and $\ell'_i = \ell(\ccc_1)$, that is now allowed.
		 
		 If $a = \ell(\ccc_2)$ and $b = \ell(\ccc_1)$, then, if $a < b$, then by  \Cref{lem:edgediagpi}, we obtain that $x_1 < x_2$. This implies that there exists some $2 \le i \le \Delta$ such that $\ell_i = \X$, and $\ell'_i = \ell(\ccc_1)$, which is not allowed. If $a$ and $b$ are uncomparable, then there  must exist some $2 \le i \le \Delta$ such that $\ell_i = \ell(\ccc_2)$ or $\ell_i = \X$, and $\ell'_i = \ell(\ccc_1)$. In both cases, we have $\ell_i \not\le \ell'_i$.
		 
		 \item Both configurations are pointer configurations,  $C_1 = \P_s \s \U^{\Delta-1}_s$ and $C_2 = \P_t \s \U^{\Delta-1}_t$. Assume w.l.o.g.\ that $s \le t$. Note that, according to \Cref{lem:edgediagpi}, $\P_s \le \P_t < \U_t \le \U_s$. 
		 
		 If $b = \P_s$, then $a = \P_t$ or  $a = \U_t$. We obtain that $b \le a$, which is not allowed.
		 
		 If $b = \U_s$, then there exists some $2 \le i \le \Delta$ such that $\ell'_i = \P_s$. Then, $\ell_i$ is either $\P_t$ or $\U_t$. If $\ell_i = \U_t$, or $\ell_i = \P_t$ and $s \neq t$, then $\ell'_i < \ell_i$, which is not allowed. If $\ell_i = \P_t$ and $s = t$, then $a = \U_s$, implying $a=b$, which is not allowed.
		 
		 If $b = \P_t$, then $a$ cannot be $\U_s$, since $\P_t < \U_s$. Hence, $a = \P_s$, implying $s \neq t$, since otherwise we would get that $a=b$. In this case, $\ell'_2 = \U_t$ and $\ell_2 = \U_s$, and we obtain $\ell'_2 < \ell_2$, which is not allowed. 
		 
		 If $b = \U_t$, then $a$ cannot be $\U_s$, since $\U_t \le \U_s$. Hence, $a = \P_s$.
		 In this case, there must exist some $2 \le i \le \Delta$ such that $\ell'_i = \P_t$. Then, $\ell_i$ is $\U_s$, implying that $\ell'_i < \ell_i$, which is not allowed.

		 \item One configuration is a color configuration  $C_1 = \ell(\ccc)^{\Delta-x} \s \X^{x}$, and the other is a pointer configuration $C_2 = \P_s \s \U^{\Delta-1}_s$.
		 It is not possible that $a = \X$, since we need that $b \not\le a$, but $\X$ is at least as strong as any other label.
		 
		  If $a = \ell(\ccc)$ and $x \ge 1$, then there must be some $i$ such that $\ell_i = \X$, and since $\ell_i \le \ell'_i$, then $\ell'_i = \X$, which is a contradiction.
		 
		 If $a = \ell(\ccc)$ and $x = 0$, then for all $2 \le i \le \Delta$, $\ell_i = \ell(\ccc)$. In order to have $\ell_i \le \ell'_i$, by \Cref{lem:edgediagpi} it must hold that $\level(\ccc) \ge s$, since for some $i$, it holds that $\ell'_i = \U_s$. For the same reason, it must hold that for all $2 \le i \le \Delta$, $\ell'_i \neq \P_s$. Hence, $b = \P_s$. Since $x=0$, then $|\ccc| = 1$, implying that for all $C \in \ccc$, $\level(C) \ge s$, and hence that $\P_s < \ell(\ccc)$, and hence that $b < a$, which is not allowed.
		 
		 If $a = \P_s$, then, for all $2 \le i \le \Delta$, $\ell_i = \U_s$. By \Cref{lem:edgediagpi}, for all $\ccc$ it holds that $\U_s \not \le \ell(\ccc)$. Hence, for all $2 \le i \le \Delta$, $\ell'_i$ must be $\X$, but this is not possible, since the existence of the configuration $C_2$ implies that $\len(z) > 0$, and hence that $|z| \le \Delta-1$, and thus that $x \le \Delta-2$.
		 
		 If $a = \U_s$, by \Cref{lem:edgediagpi} we either have that $b = \X$, or that $b = \ell(\ccc)$ and $\level(\ccc) < s$. If $b = \X$ then there must exist some $i$ such that $\ell'_i = \ell(\ccc)$ such that $\ell_i = \U_s$ (since, for the same reason as above, it must hold that $x \le \Delta-2$), but this is not allowed. If $b = \ell(\ccc)$ and $\level(\ccc) < s$, then there must be some $i$ such that $\ell'_i = \ell(\ccc)$ and $\ell_i$ is either $\P_s$ or $\U_s$. Since  $\level(\ccc) < s$, then both are not as weak as $\ell(\ccc)$.
		 
	\end{itemize}

\end{proof}

Recall the definition of a relaxation of a node configuration in~\Cref{def:noderelax}.
The second and much more complicated step we have to perform is to show that $\rere(\re(\Pi_\Delta(z)))$ can be relaxed to $\Pi_\Delta(z')$, where $z' = \prefix(z)$.
Our key technical lemma is given in the following.
Recall that we set $\U_0 := \X$.

\begin{lemma}\label{lem:keytec}
	 Assume $|z| \le \Delta$ if $\len(z) = 0$, and $|z| \le \Delta-1$ otherwise. Then each node configuration of $\rere(\re(\Pi_\Delta(z)))$ can be relaxed to
	\begin{align*}
		&\gen{\gen{\P_i}} \s \gen{\gen{\U_i}}^{\Delta - 1} &\textrm{ for some }1 \leq i \leq \beta&, \textrm{ or}\\
		&\gen{\gen{\U_i, \ell(\ccc)}, \gen{\P_j}}^{\Delta - |\ccc| - i + j + 1} \s \gen{\gen{X}}^{|C| + i - j - 1} &\textrm{ for some }\ccc \in 2^\ccs \setminus \{\emptyset\}, 0 \leq i, j \leq \beta&\\ &&\textrm{satisfying }\level(\ccc) \leq j\textrm{ and }i \in \{ j, j+1 \}&,
	\end{align*}
	where we set $\gen{\gen{\U_i, \ell(\ccc)}, \gen{\P_j}} := \gen{\gen{\U_i, \ell(\ccc)}}$ if $j = 0$.
\end{lemma}

\subsection{Proof of Lemma~\ref{lem:keytec}}\label{sec:keytec}
Throughout this section, assume $|z| \le \Delta$ if $\len(z) = 0$, and $|z| \le \Delta-1$ otherwise.
Let $\nodeconst'$ denote the set of configurations given in \Cref{lem:keytec}, and $\nodeconst_{\rere}$ the node constraint of $\rere(\re(\Pi_\Delta(z)))$.
We prove \Cref{lem:keytec} by showing that any configuration with labels from $2^{\Sigma_{\re}} \setminus \{\emptyset\}$ that cannot be relaxed to any configuration from $\nodeconst'$ is not contained in $\nodeconst_{\rere}$.

Let $\fS = \S_1 \s \dots \s \S_\Delta$ be an arbitrary configuration with labels from $2^{\Sigma_{\re}} \setminus \{\emptyset\}$ that cannot be relaxed to any configuration from $\nodeconst'$.
If at least one of the $\S_k$ is not right-closed (w.r.t.\ $\nodeconst_{\re}$), then, by \Cref{obs:rcs}, $\fS$ is not contained in $\nodeconst_{\rere}$, and we are done.
Hence, assume throughout the remainder of the section that all $\S_k$ are right-closed.

In the following, we collect some properties that $\fS$ satisfies due to the fact that it cannot be relaxed to any configuration from $\nodeconst'$.
Observe that, by definition, the $\S_k$ are sets of sets of labels from $\Sigma_{\Delta}(\Pi)$.

\begin{lemma}\label{lem:confprop}
	The configuration $\fS$ satisfies the following two properties.
	\begin{enumerate}
		\item\label{item:color} For any $\ccc \in 2^\ccs \setminus \{\emptyset\}$, $0 \leq i, j \leq \beta$ satisfying $\level(\ccc) \leq j$ and $i \in \{ j, j+1 \}$, there exist at least $|\ccc| + i - j$ indices $k \in \{1, \dots, \Delta\}$ such that there exists some set $\Q_k \in \S_k$ satisfying
		\begin{enumerate}
			\item\label{item:colorone} $\U_i \notin \Q_k$ or $\ell(\ccc) \notin \Q_k$, \textbf{and}
			\item\label{item:colortwo} $\P_j \notin \Q_k$ if $j > 0$.
		\end{enumerate} 
		\item\label{item:pointer} For any $1 \leq i \leq \beta$,
		\begin{enumerate}
			\item there exist sets $\Q_1 \in \S_1, \dots, \Q_\Delta \in \S_\Delta$ such that $\P_i \notin \Q_k$ for all $1 \leq k \leq \Delta$, \textbf{or}
			\item there exists some index $1 \leq k \leq \Delta$ and a set $\Q_k \in \S_k$ such that $\U_i \notin \Q_k$.
		\end{enumerate}
	\end{enumerate}
\end{lemma}
\begin{proof}
	Suppose for a contradiction that the lemma does not hold.
	Then $\fS$ violates property~\ref{item:color} or property~\ref{item:pointer}.
	
	Consider first the case that $\fS$ violates property~\ref{item:color}, and let $\ccc \in 2^\ccs \setminus \{\emptyset\}$, $0 \leq i, j \leq \beta$ be a choice for which property~\ref{item:color} is violated, i.e., we have $\level(\ccc) \leq j$ and $i \in \{ j, j+1 \}$, and there exist at most $|\ccc| + i - j - 1$ indices $k \in \{1, \dots, \Delta\}$ such that there is some set $\Q_k \in \S_k$ satisfying properties~\ref{item:colorone} and~\ref{item:colortwo}.
	W.l.o.g., let these indices be $\Delta - |\ccc| - i + j + 2, \dots, \Delta$.
	Then, for any $1 \leq k \leq \Delta - |\ccc| - i + j + 1$, and any set $\Q_k \in \S_k$, property~\ref{item:colorone} is violated or property~\ref{item:colortwo} is violated, i.e., 1) $\U_i \in \Q_k$ and $\ell(\ccc) \in \Q_k$, or 2) $\P_j \in \Q_k$ and $j > 0$.
	If the former holds, then, since by \Cref{obs:rcs} $\Q_k$ is right-closed  w.r.t.\ $\edgeconst_{\Delta}(z)$, we have $\gen{\U_i, \ell(\ccc)} \subseteq \Q_k$, and we obtain, by \Cref{obs:subsetarrow}, that $\Q_k$ is at least as strong as $\gen{\U_i, \ell(\ccc)}$ according to $\nodeconst_{\re}$, which implies $\Q_k \in \gen{\gen{\U_i, \ell(\ccc)}}$.
	If the latter holds, then $\gen{\P_j} \subseteq \Q_k$ and $j > 0$, and we analogously obtain that $\Q_k \in \gen{\gen{\P_j}}$ (and $j > 0$).
	Since, for any $1 \leq j \leq \beta$, we have $\gen{\gen{\U_i, \ell(\ccc)}} \subseteq \gen{\gen{\U_i, \ell(\ccc)}, \gen{\P_j}}$ and $\gen{\gen{\P_j}} \subseteq \gen{\gen{\U_i, \ell(\ccc)}, \gen{\P_j}}$ (by the definition of $\gen{\cdot}$), it follows that $\Q_k \in \gen{\gen{\U_i, \ell(\ccc)}, \gen{\P_j}}$ (where we set $\gen{\gen{\U_i, \ell(\ccc)}, \gen{\P_j}} := \gen{\gen{\U_i, \ell(\ccc)}}$ if $j = 0$).
	Note that this property holds for any $\Q_k \in \S_k$.
	Hence, for any $1 \leq k \leq \Delta - |\ccc| - i + j + 1$, we have $\S_k \subseteq \gen{\gen{\U_i, \ell(\ccc)}, \gen{\P_j}}$, which implies that $\fS$ can be relaxed to the configuration $\gen{\gen{\U_i, \ell(\ccc)}, \gen{\P_j}}^{\Delta - |\ccc| - i + j + 1} \s \gen{\gen{X}}^{|C| + i - j - 1}$, yielding a contradiction to the fact that $\fS$ cannot be relaxed to any configuration from  $\nodeconst'$.
	Here we used that any set of labels from $\Sigma_{\re}$ is a subset of $\gen{\gen{\X}}$, which follows from \Cref{obs:subsetarrow}, the fact that by \Cref{obs:rcs} the sets in $\Sigma_{\re}$ are right-closed w.r.t\ $\edgeconst_{\Delta}(z)$, and the fact that $\X$ is at least as strong as any label in $\Sigma_{\Delta}(z)$ (which in turn follows from \Cref{lem:edgediagpi}).

	Now consider the second case, namely that  $\fS$ violates property~\ref{item:pointer}, and let $1 \leq i \leq \beta$ be a choice for which property~\ref{item:pointer} is violated, i.e., we have that 1) there exists an index $1 \leq k' \leq \Delta$ such that any set $\Q_{k'} \in \S_{k'}$ satisfies $\P_i \in \Q_{k'}$, and 2) for any $1 \leq k \leq \Delta$ and any $\Q_k \in \S_k$, we have $\U_i \in \Q_k$.
	W.l.o.g., assume that $k' = 1$.
	By an analogous argumentation to the one from the first case, we obtain $\Q_1 \in \gen{\gen{\P_i}}$ for any $\Q_1 \in \S_1$, and $\Q_k \in \gen{\gen{\U_i}}$ for any $k \in \{2, \dots, \Delta\}$ and any $\Q_k \in \S_k$.
	Hence, $\S_1 \subseteq \gen{\gen{\P_i}}$, and $\S_k \subseteq \gen{\gen{\U_i}}$ for any $2 \leq k \leq \Delta$.
	It follows that $\fS$ can be relaxed to the configuration $\gen{\gen{\P_i}} \s \gen{\gen{\U_i}}^{\Delta - 1}$, again obtaining a contradiction to the fact that $\fS$ cannot be relaxed to any configuration from  $\nodeconst'$.
\end{proof}

Recall that in order to prove Lemma~\ref{lem:keytec}, we want to show that certain configurations with labels from $2^{\Sigma_{\re}} \setminus \{\emptyset\}$ are not contained in $\nodeconst_{\rere}$.
To this end, the following lemma gives a useful sufficient condition for the property that a configuration is not contained in $\nodeconst_{\rere}$.

\begin{lemma}\label{lem:notsteptwo}
	Let $\D_1 \s \dots \s \D_\Delta$ be a configuration with labels from $2^{\Sigma_{\re}} \setminus \{\emptyset\}$.
	Assume that there exists some $(\B_1, \dots, \B_\Delta) \in \D_1 \times \dots \times \D_\Delta$ satisfying that
	\begin{enumerate}
		\item\label{item:intercolor} for each $\ccc \in 2^{\ccs} \setminus \{\emptyset\}$, there are at least $|\ccc|$ indices $k \in \{ 1, \dots, \Delta \}$ such that $\ell(\ccc) \notin \B_k$, \textbf{and}
		\item\label{item:interpointer} there exists some $0 \leq i \leq \beta$ such that
		\begin{enumerate}
			\item\label{item:interpointerone} if $i \geq 1$, then $\P_i \notin \B_k$ for all $1\leq k \leq \Delta$, \textbf{and}
			\item\label{item:interpointertwo} if $i \leq \beta - 1$, then there exists some index $1 \leq k \leq \Delta$ such that $\U_{i+1} \notin \B_k$.
		\end{enumerate}
	\end{enumerate}
	Then $\D_1 \s \dots \s \D_\Delta \notin \nodeconst_{\rere}$.
\end{lemma}
\begin{proof}
	Let $\fB = \B_1 \s \dots \s \B_\Delta$ be a configuration with labels from $\Sigma_{\re}$ satisfying properties~\ref{item:intercolor} and~\ref{item:interpointer}.
	By the definition of $\rere(\cdot)$, in order to prove the lemma, it suffices to show that $\B \notin \nodeconst_{\re}$.

	Assume for a contradiction that $\B \in \nodeconst_{\re}$, and recall that the labels in $\B$ are sets of labels from $\Sigma_{\Delta}(z)$.
	By the definition of $\re(\cdot)$, it follows that there exists some $(\A_1, \dots, \A_\Delta) \in \B_1 \times \dots \times \B_\Delta$ such that $\fA := \A_1 \s \dots \s \A_\Delta \in \nodeconst_\Delta(z)$.
	If $\fA = \ell(\ccc)^{\Delta - |\ccc| + 1} \s \X^{|\ccc| - 1}$ for some $\ccc \in 2^\ccs \setminus  \{\emptyset\}$, then $\fB$ violates property~\ref{item:intercolor}.
	Hence, $\fA$ is not of this form, which implies that $\fA = \P_j \s \U_j^{\Delta-1}$ for some $1 \le j \le \beta$.
	W.l.o.g., assume that $\A_1 = \P_j$ and $\A_k = \U_j$ for any $2 \leq k \leq \Delta$.
	By \Cref{obs:rcs} and \Cref{lem:edgediagpi}, it follows that $\P_i \in \B_1$ for all $j \leq i \leq \beta$, and $\U_i \in \B_k$ for all $1 \leq i \leq j$ and $1 \leq k \leq \Delta$.
	Thus, for each $0 \leq i \leq j-1$, configuration $\fB$ violates property~\ref{item:interpointertwo}, and, for each $j \leq i \leq \beta$, configuration $\fB$ violates property~\ref{item:interpointerone}.
	Hence, $\fB$ violates property~\ref{item:interpointer}, which yields the desired contradiction.
\end{proof}

For a set $V'$ of nodes of a graph $G$, let $N_G(V')$ denote the set of all nodes of $G$ that are adjacent to at least one vertex from $V'$.
In order to prove Lemma~\ref{lem:keytec}, we will make use of Hall's marriage theorem~\cite[Theorem 1]{hall} that can be stated as follows.
\begin{theorem}[Hall]\label{thm:hall}
	Let $G = (V \cup W, E)$ be a bipartite graph such that, for each subset $V' \subseteq V$, we have $|N_G(V')| \geq |V'|$.
	Then there exists a matching in $G$ that matches all vertices from $V$.
\end{theorem}

In the following, we will combine \Cref{lem:confprop,lem:notsteptwo} and \Cref{thm:hall} to prove \Cref{lem:keytec}.
Recall that to prove \Cref{lem:keytec} it suffices to show that $\fS = \S_1 \s \dots \s \S_\Delta \notin \nodeconst_{\rere}$.
To this end, we will show that there exists some $(\B_1, \dots, \B_\Delta) \in \S_1 \times \dots \times \S_\Delta$ satisfying properties~\ref{item:intercolor} and~\ref{item:interpointer} from \Cref{lem:notsteptwo}, and then apply \Cref{lem:notsteptwo}.

Let $1 \leq i^* \leq \beta$ denote the largest index such that, for each $\S_k$ and each $\Q_k \in \S_k$, we have $\U_{i^*} \in \Q_k$.
If no such index exists, i.e., if, for each $1 \leq i \leq \beta$, there exists some $1 \leq k \leq \Delta$ and some $\Q_k \in \S_k$ such that $\U_i \notin \Q_k$, then set $i^* := 0$.

Consider the bipartite graph $G = (V \cup W, E)$ where $V = \ccs$, $W = \{\S_1, \dots, \S_\Delta\}$, and there is an edge between some color $C \in V$ and some set $\S_k \in W$ if and only if there exists some $\Q_k \in \S_k$ satisfying 1) $\ell(\{ C \}) \notin \Q_k$ and 2) $\P_{i^*} \notin \Q_k$ if $i^* > 0$.  

Consider some arbitrary set $\ccc \in 2^\ccs \setminus \{ \emptyset \}$, and set $i_{\max} := \max(i^*, \level(\ccc))$.
By applying (property~\ref{item:color} of) \Cref{lem:confprop} with parameters $\ccc$ and $i = j = i_{\max}$, we obtain that there exist at least $|\ccc|$ indices $k \in \{1, \dots, \Delta\}$ such that there exists some set $\Q_k \in \S_k$ satisfying
\begin{enumerate}
	\item\label{item:newcolorone} $\U_{i_{\max}} \notin \Q_k$ or $\ell(\ccc) \notin \Q_k$, and
	\item\label{item:newcolortwo} $\P_{i_{\max}} \notin \Q_k$ if $i_{\max} > 0$.
\end{enumerate} 
Consider an arbitrary such $\Q_k$.
If $i_{\max} = i^*$, then, by the definition of $i^*$, we have $\U_{i_{\max}} \in \Q_k$, and hence $\ell(\ccc) \notin \Q_k$ (by property~\ref{item:newcolorone}).
If $i_{\max} \neq i^*$, then, by the definition of $i_{\max}$, we have $i_{\max} =\level(\ccc)$, which implies that $\ell(\ccc)$ is at least as weak as $\U_{i_{\max}}$ according to $\edgeconst_\Delta(z)$, by \Cref{lem:edgediagpi}.
Since, by \Cref{obs:rcs}, $\Q_k$ is right-closed (w.r.t.\ $\edgeconst_\Delta(z)$), it follows that $\ell(\ccc) \in \Q_k$ would imply $\U_{i_{\max}} \in \Q_k$.
Thus, by property~\ref{item:newcolorone}, we again obtain $\ell(\ccc) \notin \Q_k$.
Hence, irrespective of whether $i_{\max} = i^*$ or $i_{\max} \neq i^*$, we have $\ell(\ccc) \notin \Q_k$.

Moreover, observe that the characterization of $\Sigma_{\re}$ given in \Cref{lem:repi} implies, by \Cref{lem:edgediagpi}, that for every $\B \in \Sigma_{\re}$ there is a set $\ccc' \in 2^\ccs$ such that, for every $\ccc'' \in 2^\ccs \setminus \{ \emptyset \}$, we have $\ell(\ccc'') \in \B$ if and only if $\ccc'' \subseteq \ccc'$: if $\B$ is of the form $\gen{\P_{i'}, \ell(\ccc''')}$, then $\ccc' = \ccc'''$, and if $\B$ is of the form $\gen{\U_{i'}, \ell(\ccc''')}$, then also $\ccc' = \ccc'''$.
We claim that there exists some $C \in \ccc$ such that $\ell(\{C\}) \notin \Q_k$.
Assume for a contradiction that, for all $C \in \ccc$, we have $\ell(\{C\}) \in \Q_k$.
Then, since $\Q_k \in \Sigma_{\re}$, the aforementioned characterization implies that there exists some $\ccc' \supseteq \bigcup_{C \in \ccc} \{ C \}$ such that $\ell(\ccc'') \in \Q_k$ for any $\ccc'' \subseteq \ccc'$.
Since $\ccc \subseteq \ccc'$, we obtain $\ell(\ccc) \in \Q_k$, yielding the desired contradiction.
Hence, there exists some $C \in \ccc$ such that $\ell(\{C\}) \notin \Q_k$.

Furthermore, if $i^* > 0$, then, by \Cref{lem:edgediagpi}, $\P_{i^*} \leq \P_{i_{\max}}$, which implies $\P_{i^*} \notin \Q_k$, by property~\ref{item:newcolortwo} and the right-closedness of $\Q_k$ (which is due to \Cref{obs:rcs}).

Combining the above insights, we obtain that there exist at least $|\ccc|$ indices $k \in \{1, \dots, \Delta\}$ such that there exists some set $\Q_k \in \S_k$ satisfying that
\begin{enumerate}
	\item\label{item:obtainedone} there exists some color $C \in \ccc$ with $\ell(\{ C \}) \notin \Q_k$, and
	\item\label{item:obtainedtwo} $\P_{i^*} \notin \Q_k$ if $i^* > 0$.
\end{enumerate} 
As this statement holds for any $\ccc \in 2^\ccs \setminus \{ \emptyset \}$, we obtain, by the definition of $G$, that, for any subset $\ccc \subseteq V$, we have $N_G(\ccc) \geq |\ccc|$.
By applying \Cref{thm:hall}, i.e., Hall's marriage theorem, we obtain that there exists some matching $M$ such that each $C \in \ccs = V$ is matched.
W.l.o.g., assume that the $\S_k$ that are matched are $\S_1, \dots, \S_{|\ccs|}$.
For each $C \in \ccs$, denote the $\S_k$ color $C$ is matched to by $M(C)$.
Similarly, for each $\S_k$ with $1 \leq k \leq |\ccs|$, denote the color $C$ set $\S_k$ is matched to by $M(\S_k)$.

Now, let $\ccc_{\leq i^*} \in 2^\ccs$ denote the set of colors of level at most $i^*$.
W.l.o.g., assume that the colors in $\ccc_{\leq i^*}$ are matched with the sets from $\{ \S_1, \dots, S_{|\ccc_{\leq i^*}|}\}$.
Furthermore, observe that, if $i^* > 0$, then property~\ref{item:pointer} of \Cref{lem:confprop}, together with the definition of $i^*$, implies that for any $1 \leq k \leq \Delta$ there exists some set $\Q_k \in \S_k$ satisfying $\P_{i^*} \notin \Q_k$.
Also, note that it follows from the definition of $i^*$ that $i^* = \beta$ or there exist some $1 \leq k \leq \Delta$ and some $\Q_k \in \S_k$ such that $\U_{i^* + 1} \notin \Q_k$.
If $i^* < \beta$, denote the largest index $k$ such that there exists some $\Q_k \in \S_k$ satisfying $\U_{i^* + 1} \notin \Q_k$ by $k_{\max}$.
We consider two cases.
\newline

\noindent{\textbf{\boldmath Case I: $i^* = \beta$, or $i^* < \beta$ and $k_{\max} > |\ccc_{\leq i^*}|$}}

\noindent For each $1 \leq k \leq |\ccs|$, set $\B_k$ to be some $\Q_k \in \S_k$ such that $\ell(\{M(\S_k)\}) \notin \Q_k$ and $\P_{i^*} \notin \Q_k$ if $i^* > 0$.
Such a choice exists by the definition of $G$.
If $i^* < \beta$, set $\B_{k_{\max}}$ to be some $\Q_{k_{\max}} \in \S_{k_{\max}}$ such that $\U_{i^* + 1} \notin \Q_{k_{\max}}$ (overwriting the previous choice for $\B_{k_{\max}}$ if $k_{\max} \leq |\ccs|$).
If $i^* > 0$, then, for each $1 \leq k \leq \Delta$ for which $\B_k$ has not been set yet, set $\B_k$ to be some $\Q_k \in \S_k$ satisfying $\P_{i^*} \notin \Q_k$.
Such a choice exists as shown above.
If $i^* = 0$, then, for each $1 \leq k \leq \Delta$ for which $\B_k$ has not been set yet, set $\B_k$ to be some arbitrary $\Q_k \in \S_k$.

We claim that the obtained tuple $(\B_1, \dots, \B_{\Delta}) \in \S_1 \times \dots \times \S_\Delta$ satisfies properties~\ref{item:intercolor} and~\ref{item:interpointer} of \Cref{lem:notsteptwo}.
Consider first property~\ref{item:intercolor}, and consider some arbitrary $\ccc \in 2^\ccs \setminus \{\emptyset\}$.
Let $K$ denote the set of indices $1 \leq k \leq |\ccs|$ such that $M(\S_k) \in \ccc$.
As $M$ is a matching, we have $|K| = |\ccc|$.

Consider some arbitrary $k \in K$.
Since $1 \leq k \leq |\ccs|$, we know, by the definition of $\B_k$, that 1) $\ell(\{M(\S_k)\}) \notin \B_k$, or 2) $i^* < \beta$, $k = k_{\max} \leq |\ccs|$, and $\U_{i^* + 1} \notin \B_{k}$.
In the latter case, recall that $k_{\max} > |\ccc_{\leq i^*}|$, which implies that $\level(M(\S_k)) \geq i^* + 1$.
By \Cref{lem:edgediagpi}, it follows that $\ell(M(\S_k)) \leq \U_{i^* + 1}$, and, by the right-closedness of $\B_k$ (due to \Cref{obs:rcs}), we obtain $\ell(\{M(\S_k)\}) \notin \B_k$.
Hence, in either case, we have $\ell(\{M(\S_k)\}) \notin \B_k$.
Now, by \Cref{lem:edgediagpi} and the fact that $M(\S_k) \in \ccc$, we obtain $\ell(\ccc) \leq \ell(\{M(\S_k)\})$, which, by the right-closedness of $\B_k$, implies $\ell(\ccc) \notin \B_k$.
Hence, for each $k \in K$, we have $\ell(\ccc) \notin \B_k$.
Since $|K| = |\ccc|$, it follows that there are at least $|\ccc|$ indices $k \in \{1, \dots, \Delta\}$ such that $\ell(\ccc) \notin \B_k$, which shows that $(\B_1, \dots, \B_{\Delta})$ satisfies property~\ref{item:intercolor} of \Cref{lem:notsteptwo}.

Now consider property~\ref{item:interpointer} of \Cref{lem:notsteptwo}.
We show that this property holds for $(\B_1, \dots, \B_{\Delta})$ by showing that properties~\ref{item:interpointerone} and~\ref{item:interpointertwo} of \Cref{lem:notsteptwo} are satisfied when setting $i = i^*$.
Consider first property~\ref{item:interpointerone}.
If $i^* = 0$, then property~\ref{item:interpointerone} is trivially satisfied, hence assume $i^* > 0$.
By the definition of the $\B_k$, we obtain, for each $1 \leq k \leq \Delta$ that $\P_{i^*} \notin \B_k$ or that $\U_{i^* + 1} \notin \B_{k}$, which also implies $\P_{i^*} \notin \B_k$, due to the right-closedness of $\B_k$ and the fact that $\P_{i^*} \leq \U_{i^* + 1}$ (which follows from \Cref{lem:edgediagpi}).
Therefore, property~\ref{item:interpointerone} is satisfied.

Now consider property~\ref{item:interpointertwo} of \Cref{lem:notsteptwo}.
If $i^* = \beta$, then property~\ref{item:interpointertwo} is trivially satisfied, hence assume $i^* < \beta$.
By the definition of the $\B_k$, we obtain $\U_{i^* + 1} \notin \B_k$ for $k = k_{\max}$, proving that also  property~\ref{item:interpointertwo} is satisfied.
Hence, $(\B_1, \dots, \B_{\Delta})$ satisfies also property~\ref{item:interpointer} of \Cref{lem:notsteptwo}.

We conclude that $(\B_1, \dots, \B_{\Delta})$ satisfies the conditions given in \Cref{lem:notsteptwo}, and by applying \Cref{lem:notsteptwo}, we obtain $\S_1 \s \dots \s \S_\Delta \notin \nodeconst_{\rere}$, as desired.
\newline

\noindent{\textbf{\boldmath Case II: $i^* < \beta$ and $k_{\max} \leq |\ccc_{\leq i^*}|$}}

\noindent The challenge that we have to overcome in Case II is that we cannot simply choose $\B_{k_{\max}}$ to be some $\Q_{k_{\max}} \in \S_{k_{\max}}$ such that $\U_{i^* + 1} \notin \Q_{k_{\max}}$.
Such a choice would make sure that property~\ref{item:interpointertwo} of \Cref{lem:notsteptwo} is satisfied, but it would not guarantee that the color $C$ that is matched to $\S_{k_{\max}}$ (or more precisely, the set $\{C\}$) satisfies the condition given in property~\ref{item:intercolor} of \Cref{lem:notsteptwo}.
Our solution is to change the matching $M$ to another matching with the property that there is some $\S_k$ that is not matched to any color from $\ccc_{\leq i^*}$ but contains a set $\Q_k$ satisfying $\U_{i^* + 1} \notin \Q_{k}$ (which allows us to choose $\B_k = \Q_k$ and thereby satisfy property~\ref{item:interpointertwo} without violating property~\ref{item:intercolor}).
To this end, we essentially construct a suitable augmenting path along which we change the matching (while we will not use this exact terminology in the proof).

Let $\fE \subseteq \ccs$ denote the set of colors $C$ such that there exists some set $\Q \in M(C)$ satisfying $\U_{i^* + 1} \notin \Q$.
By the definition of $k_{\max}$ and the fact that $k_{\max} \leq |\ccc_{\leq i^*}|$, we obtain that $\emptyset \neq \fE \subseteq \ccc_{\leq i^*}$.
Moreover, recall that $E$ denotes the set of edges of $G$, and recursively define sets $\fF_1, \fF_2, \dots$ (which are subsets of $\ccc_{\leq i^*}$) as follows.
\begin{align*}
	\fF_1 := &\ \{ C \in \ccc_{\leq i^*} \mid \text{there exists some $|\ccc_{\leq i^*}| + 1 \leq k \leq \Delta$ such that $\{C, \S_k\} \in E$} \},\\
	\fF_y := &\ \{ C \in \ccc_{\leq i^*} \setminus (\fF_1 \cup \dots \cup \fF_{y-1}) \mid \text{there exists some $C' \in \fF_{y-1}$ such that $\{C, M(C')\} \in E$} \},\\
		 &\ \text{for all $y \geq 2$}.
\end{align*}
The following lemma shows that in this recursive definition, the $\fF_y$ do not become empty before some $\fF_y$ is reached that has a nonempty intersection with $\fE$.

\begin{lemma}\label{lem:eandf}
	Let $y$ be some positive integer.
	If $\fE \cap (\fF_1 \cup \dots \cup \fF_{y-1}) = \emptyset$, then $\fF_y \neq \emptyset$.
\end{lemma} 
\begin{proof}
	Set $\ccc := \ccc_{\leq i^*} \setminus (\fF_1 \cup \dots \cup \fF_{y-1})$ and observe that $\level(\ccc) \leq \level(\ccc_{\leq i^*}) \leq i^*$ since $\ccc \subseteq \ccc_{\leq i^*}$.
	Observe further that $\fE \cap (\fF_1 \cup \dots \cup \fF_{y-1}) = \emptyset$ implies $\ccc \neq \emptyset$ since $\emptyset \neq \fE \subseteq \ccc_{\leq i^*}$ as observed above.
	Hence, by property~\ref{item:color} of \Cref{lem:confprop} (applied with parameters $\ccc$, $j = i^*$, and $i = i^* + 1$), we know that there exist at least $|\ccc| + 1$ indices $k \in \{1, \dots, \Delta\}$ such that there exists some set $\Q_k \in \S_k$ satisfying 1) $\U_{i^* + 1} \notin \Q_k$ or $\ell(\ccc) \notin \Q_k$, and 2) $\P_{i^*} \notin \Q_k$ if $i^* > 0$.
	By the pigeonhole principle, it follows that there also exists such a $k$ additionally satisfying that $\S_k$ is not matched or $M(\S_k) \notin \ccc$; consider such a $k$.
	Since $\fE \subseteq \ccc$, we obtain that $\S_k$ is not matched or $M(\S_k) \notin \fE$.
	Observe that each $\S_{k'}$ such that there exists some $\Q_{k'} \in \S_{k'}$ with $\U_{i^* + 1} \notin \Q_{k'}$ is matched (as, by the definition of $k_{\max}$, we have $k' \le k_{\max} \le |\ccc_{\leq i^*}|$, and, for each $k'' \leq |\ccc_{\leq i^*}|$, all $\S_{k''}$ are matched (with colors from $\ccc_{\leq i^*}$)) and satisfies $M(\S_{k'}) \in \fE$ (by the definition of $\fE$).
	Hence, for all $\Q_k \in \S_k$, we have $\U_{i^* + 1} \in \Q_k$.

	By the definition of $k$, it follows that there exists some set $\Q_k \in \S_k$ satisfying 1) $\ell(\ccc) \notin \Q_k$, and 2) $\P_{i^*} \notin \Q_k$ if $i^* > 0$.
	Recall that, as observed before, $\ell(\ccc) \notin \Q_k$ implies that there exists some $C \in \ccc$ such that $\ell(\{ C \}) \notin \Q_k$ (due to the characterization of $\Sigma_{\re}$ given in \Cref{lem:repi}).
	Consider such a $C$.
	By the definition of $G$, there is an edge between $C$ and $\S_k$.
	
	Recall that $C \in \ccc$ and that $\S_k$ is not matched or $M(\S_k) \notin \ccc$.
	If $y = 1$, then $\ccc = \ccc_{\leq i^*}$, and we obtain $|\ccc_{\leq i^*}| + 1 \leq k \leq \Delta$, which implies $C \in \fF_1$.
	Thus, $\fF_y \neq \emptyset$, as desired.
	Hence, assume that $y \geq 2$.
	Observe that we have $1 \leq k \leq |\ccc_{\leq i^*}|$ as otherwise $C \in \fF_1$, which would contradict $C \in \ccc \subseteq \ccc_{\leq i^*} \setminus \fF_1$.
	In particular, $\S_k$ is matched, and $M(\S_k) \in \ccc_{\leq i^*}$.
	Since $M(\S_k) \notin \ccc$, we have $M(\S_k) \in \ccc_{\leq i^*} \setminus \ccc = \fF_1 \cup \dots \cup \fF_{y-1}$.
	As, by definition, $\fF_1, \dots, \fF_{y-1}$ are pairwise disjoint, it follows that there is precisely one $1 \leq y' \leq y-1$ such that $M(\S_k) \in \fF_{y'}$.
	Again by the definition of $\fF_1, \fF_2, \dots$, we obtain that $C \in \fF_{y' + 1}$.
	Since $C \in \ccc = \ccc_{\leq i^*} \setminus (\fF_1 \cup \dots \cup \fF_{y-1})$, we see that $y' = y - 1$.
	Hence, $C \in \fF_y$, and we obtain $\fF_y \neq \emptyset$, as desired.	
\end{proof}

Since the $\fF_y$ are subsets of (the finite set) $\ccc_{\leq i^*}$ and pairwise disjoint, there is some positive integer $y'$ with $\fF_{y'} = \emptyset$.
By \Cref{lem:eandf}, it follows that $\fE \cap (\fF_1 \cup \dots \cup \fF_{y'-1}) \neq \emptyset$ (which also implies that $y' \geq 2$).
Hence, there is some positive integer $y^*$ with $\fE \cap \fF_{y^*} \neq \emptyset$.
Let $C'$ be some color contained in $\fE \cap \fF_{y^*}$.


Let $C^{(1)}, \dots C^{(y^*)}$ be colors from $\fF_1, \dots, \fF_{y^*}$, respectively, such that $C^{(y^*)} = C'$ and, for any $1 \leq y \leq y^* - 1$, we have $\{ C^{(y+1)}, M(C^{(y)}) \} \in E$.
Such $C^{(y)}$ exist due to the definition of the $\fF_y$.
Moreover, let $|\ccc_{\leq i^*}| + 1 \leq k' \leq \Delta$ be some index such that $\{C^{(1)}, \S_{k'}\} \in E$.

If $\S_{k'}$ is matched in $M$, then $M(\S_{k'}) \in \ccs \setminus \ccc_{\leq i^*}$, which implies that $\level(M(\S_{k'})) \geq i^* + 1$.
Hence, in this case, $\ell(\{M(\S_{k'})\}) \leq \U_{i^* + 1}$, by \Cref{lem:edgediagpi}, and the fact that $M(C^{(y^*)})$ contains some set $\Q$ with $\U_{i^* + 1} \notin \Q$ (which follows from $C^{(y^*)} \in \fE$) implies that there is some $\Q \in M(C^{(y^*)})$ such that $\ell(\{M(\S_{k'})\}) \notin \Q$, by the right-closedness of $\Q$.
If $i^* > 0$, then for any such $\Q$, the fact that $\P_{i^*} \leq \U_{i^* + 1}$ analogously implies $\P_{i^*} \notin \Q$, and it follows that in $G$ there is an edge between $M(\S_{k'})$ and $M(C^{(y^*)})$.

Let $M'$ be the matching obtained from $M$ by
\begin{enumerate}
	\item removing the edges in $M$ that are incident to some vertex from $\{ C^{1}, \dots, C^{(y^*)} \}$,
	\item\label{item:step2} if $\S_{k'}$ is matched in $M$, removing the corresponding matching edge,
	\item if $\S_{k'}$ was matched in $M$, adding edge $\{M(\S_{k'}), M(C^{(y^*)})\}$
	\item adding edge $\{C^{(1)}, \S_{k'}\}$, and
	\item adding edge $\{ C^{(y+1)}, M(C^{(y)}) \}$, for each $1 \leq y \leq y^* - 1$.
\end{enumerate}
Note that all edges added to $M'$ are indeed edges in $G$, by the above discussion.
Note further that for each added edge $e$, the two endpoints of $e$ were unmatched after step~\ref{item:step2}; hence, $M'$ is indeed a matching.
Finally, observe that the construction of $M'$ ensures that each color $C \in \ccs$ is matched in $M'$. 

For each vertex $v$ of $G$ that is matched in $M'$, denote the vertex that $v$ is matched to by $M'(v)$.
Moreover, permute the indexing of the $\S_1, \dots, \S_\Delta$ such that the colors in $\ccs$ are matched (in $M'$) to the sets in $\{\S_1, \dots, \S_{|\ccs|}\}$, and the colors in $\ccc_{\leq i^*}$ are matched to the sets in $\{\S_1, \dots, \S_{|\ccc_{\leq i^*}|}\}$.
Whenever we refer to some $\S_k$ in the remainder of the section, it is w.r.t.\ this new indexing (unless indicated otherwise).
Let $k^*$ denote the new index of set $M(C^{(y^*)})$.
By the construction of $M'$ and the above discussion, we know that $\S_{k^*}$ is unmatched in $M'$ or matched to some color from $\ccs \setminus \ccc_{\leq i^*}$, implying that $k^* \geq |\ccc_{\leq i^*}| + 1$.
Note that, due to $\S_{k^*} = M(C^{(y^*)})$, $C^{(y^*)} \in \fE$, and by the definition of $\fE$, there is some set $\Q_{k^*} \in \S_{k^*}$ such that $\U_{i^* + 1} \notin \Q_{k^*}$. Hence, according to the new indexing, $k_{\max} \ge k^* \ge |\ccc_{\leq i^*}| + 1$.

Hence, we have essentially transformed the setting of Case II to the setting of Case I: by simply setting the $\B_k$ as in Case I (where we use matching $M'$ (instead of matching $M$), the new indexing, and the new definition of $k_{\max}$) and using the same proof as in Case I, we obtain $\S_1 \s \dots \s \S_\Delta \notin \nodeconst_{\rere}$, as before.
(The only difference is that we do not have to deal with the case $i^* = \beta$ (as we cannot have $i^* = \beta$ in Case II), which only works in our favor.)
We conclude that in either case we have $\S_1 \s \dots \s \S_\Delta \notin \nodeconst_{\rere}$, which shows that any configuration with labels from $2^{\Sigma_{\re}} \setminus \{\emptyset\}$ that cannot be relaxed to any configuration from $\nodeconst'$ is not contained in $\nodeconst_{\rere}$, and proves \Cref{lem:keytec}.

\subsection{Relaxing $\rere(\re(\Pi_{\Delta}(z)))$}\label{sec:relaxrerere}
Recall that we want to show that, if $\Pi_{\Delta}(z)$ can be solved in some number $T$ of rounds, then $\Pi_{\Delta}(z')$, where $z' = \prefix(z)$, can be solved in $T - 1$ rounds.
To this end, we show that $\rere(\re(\Pi_{\Delta}(z)))$ can be relaxed to $\Pi_{\Delta}(z')$, using \Cref{lem:keytec}, and then make use of the fact that $\rere(\re(\Pi_{\Delta}(z)))$ can be solved one round faster than $\Pi_{\Delta}(z)$.
Throughout this section, we assume that $|z| \le \Delta$ if $\len(z) = 0$, and $|z| \le \Delta-1$ otherwise (which is a requirement for applying \Cref{lem:keytec}).

Recall that $\nodeconst'$ denotes the set of configurations given in \Cref{lem:keytec}, and $\edgeconst_{\re}$ the edge constraint of $\re(\Pi_{\Delta}(z))$.
We start by defining a useful ``intermediate'' problem $\Pi^*$.
We define the label set $\Sigma^*$ of $\Pi^*$ as the set containing all labels that appear in at least one of the node configurations in $\nodeconst'$, and we set the node constraint of $\Pi^*$ to be $\nodeconst^* := \nodeconst'$.
Moreover, we define the edge constraint $\edgeconst^*$ to be the set of all cardinality-$2$ multisets $\D_1 \s \D_2$ with $\D_1, \D_2 \in \Sigma^*$ such that there exist labels $\B_1 \in \D_1$ and $\B_2 \in \D_2$ satisfying $\B_1 \s \B_2 \in \edgeconst_{\re}$.
Note that, by definition, the labels in $\Sigma^*$ are sets of sets of labels from $\Sigma_\Delta(z)$. We prove that $\Pi^*$ is at least $1$ round easier than $\Pi_\Delta(z)$, unless $\Pi_\Delta(z)$ is already $0$-round solvable.
  
\begin{lemma}\label{lem:relaxingworks}
	Assume that $|z| \le \Delta$ if $\len(z) = 0$, and $|z| \le \Delta-1$ otherwise, and let $T$ be some positive integer.
	If $\Pi_\Delta(z)$ can be solved in $T$ rounds in the deterministic port numbering model, then $\Pi^*$ can be solved in $T - 1$ rounds.
\end{lemma}
\begin{proof}
	Assume that $\Pi_\Delta(z)$ can be solved in $T$ rounds.
	Then, by \Cref{thm:rethm}, there is some $(T-1)$-round algorithm $\fA$ solving $\rere(\re(\Pi_{\Delta}(z)))$.
	For each node $v$, denote the cardinality-$\Delta$ multiset of labels that $\fA$ outputs at the half-edges incident to $v$ by $\fA(v) = \{\D_1(v), \dots, \D_{\Delta}(v)\}$.
	We define a new $(T-1)$-round algorithm $\fA'$ solving $\Pi^*$ as follows.
	Each node $v$ first executes $\fA$ (in $T - 1$ rounds), and then (without any further communication) chooses a node configuration $\D'_1 \s \dots \s \D'_{\Delta} \in \nodeconst^*$ satisfying $\D_k \subseteq \D'_k$ for each $1 \leq k \leq \Delta$, and replaces $\D_k$ by $\D'_k$, for each $1 \leq k \leq \Delta$.
	Such choices exist due to \Cref{lem:keytec} and the fact that $\nodeconst^* = \nodeconst'$.
	Observe that, for each half-edge $h$, the label (i.e., set) $\fA$ outputs at $h$ is a subset of the label $\fA'$ outputs at $h$.
	It follows, by the correctness of $\fA$ and the definitions of both $\edgeconst^*$ and the edge constraint of $\rere(\re(\Pi_{\Delta}(z)))$, that, for each edge $e$, the cardinality-$2$ multiset of labels that $\fA'$ outputs at the half-edges belonging to $e$ is contained in $\edgeconst^*$.
	Hence, $\fA'$ produces a correct output, which concludes the proof.  
\end{proof}

We now explicitly characterize the configurations allowed by the edge constraint of $\Pi^*$.

\begin{lemma}\label{lem:newedgestuff}
	The edge constraint $\edgeconst^*$ of $\Pi^*$ contains precisely the configurations
	\begin{itemize}
		\item $\gen{\gen{\U_i, \ell(\ccc)}, \gen{\P_j}} \s  \gen{\gen{\U_{i'}, \ell(\ccc')}, \gen{\P_{j'}}}$ for all $\ccc, \ccc' \in 2^{\ccs} \setminus \{ \emptyset \}, 0 \leq i, i', j, j' \leq \beta$ satisfying $\level(\ccc) \leq j, \level(\ccc') \leq j', i \in \{j, j+1\},  i' \in \{j', j'+1\}$, and
			\begin{enumerate}
				\item $\ccc \cap \ccc' = \emptyset$, or
				\item $i < j'$, or
				\item $i' < j$,
			\end{enumerate}
		\item $\gen{\gen{\U_i, \ell(\ccc)}, \gen{\P_j}} \s \gen{\gen{\U_{i'}}}$ for all $\ccc \in 2^{\ccs} \setminus \{ \emptyset \}, 0 \leq i, j \leq \beta, 1 \leq i' \leq \beta$ satisfying $\level(\ccc) \leq j, i \in \{j, j+1\}$
		\item $\gen{\gen{\U_i, \ell(\ccc)}, \gen{\P_j}} \s \gen{\gen{\P_{j'}}}$ for all $\ccc \in 2^{\ccs} \setminus \{ \emptyset \}, 0 \leq i, j \leq \beta, 1 \leq j' \leq \beta$ satisfying $\level(\ccc) \leq j, i \in \{j, j+1\}$, and $i < j'$,
		\item $\gen{\gen{U_i}} \s \gen{\gen{U_{i'}}}$ for all $1\leq i, i' \leq \beta$
		\item $\gen{\gen{U_i}} \s \gen{\gen{P_{j}}}$ for all $1\leq i < j \leq \beta$
		\item $\gen{\gen{\X}} \s \L$, for each $\L \in \Sigma^*$,
	\end{itemize}
	where, as before, we set $\U_0 := \X$ and $\gen{\gen{\U_i, \ell(\ccc)}, \gen{\P_0}} := \gen{\gen{\U_i, \ell(\ccc)}}$.
\end{lemma}
\begin{proof}
	Consider some arbitrary sets $\B_1, \B_2, \B'_1, \B'_2 \in \Sigma_{\re}$ satisfying $\gen{\B_1, \B_2}, \gen{\B'_1, \B'_2} \in \Sigma^*$.
	We claim that $\gen{\B_1, \B_2} \s \gen{\B'_1, \B'_2} \in \edgeconst^*$ if and only if at least one of the configurations $\B_1 \s \B'_1$, $\B_1 \s \B'_2$, $\B_2 \s \B'_1$, $\B_2 \s \B'_2$ is contained in $\edgeconst_{\re}$.
	The ``if'' part follows from the definitions of $\gen{\cdot}$ and $\edgeconst^*$.
	Now, assume that $\gen{\B_1, \B_2} \s \gen{\B'_1, \B'_2} \in \edgeconst^*$.
	Then, by the definition of $\edgeconst^*$, there is some $(\B_3, \B'_3) \in \gen{\B_1, \B_2} \times \gen{\B'_1, \B'_2}$ satisfying $\B_3 \s \B'_3 \in \edgeconst_{\re}$.
	By the definition of $\gen{\cdot}$, there exist some $q,r \in \{1, 2\}$ such that $\B_q \leq \B_3$ and $\B'_r \leq \B'_3$.
	By \Cref{lem:abc123}, it follows that $\B_q \subseteq \B_3$ and $\B'_r \subseteq \B'_3$.
	Since $\B_3 \s \B'_3 \in \edgeconst_{\re}$, we obtain $\B_q \s \B_r \in \edgeconst_{\re}$, by the definition of $\edgeconst_{\re}$.
	This proves the claim.

	Now, consider arbitrary labels $\A_1, \A_2, \A'_1, \A'_2 \in \Sigma_\Delta(z)$ satisfying $\gen{\A_1, \A_2}, \gen{\A'_1, \A'_2} \in \Sigma_{\re}$.
	By the definitions of $\gen{\cdot}$ and $\edgeconst_{\re}$, we obtain that $\gen{\A_1, \A_2} \s \gen{\A'_1, \A'_2} \in \edgeconst_{\re}$ if and only if all of the configurations $\A_1 \s \A'_1$, $\A_1 \s \A'_2$, $\A_2 \s \A'_1$, $\A_2 \s \A'_2$ are contained in $\edgeconst_{\Delta}(z)$.
	Note that all of the above considerations also hold if the two arguments in $\gen{\cdot, \cdot}$ are identical.

	The above insights provide a simple method for checking whether a given configuration $\D_1 \s \D_2 = \gen{\gen{A_{1,1},A_{1,2}},\gen{A_{2,1},A_{2,2}}} \gen{\gen{A'_{1,1},A'_{1,2}},\gen{A'_{2,1},A'_{2,2}}}$ with labels from $\Sigma^*$ is contained in $\edgeconst^*$: it is contained if and only if the following Boolean formula in disjunctive normal form is satisfied:
	\[
	\bigvee_{q,r \in \{1,2\}} \bigwedge_{s,t \in \{1,2\}} A_{q,s} \s A'_{r,t} \in  \edgeconst_\Delta(z).
	\]
	We illustrate this method for the arguably most complicated situation, namely the one where $\D_1$ and $\D_2$ are of the form $\gen{\gen{\U_i, \ell(\ccc)}, \gen{\P_j}}$.
	More precisely, consider the configuration $\D_1 \s \D_2 = \gen{\gen{\U_i, \ell(\ccc)}, \gen{\P_j}} \s  \gen{\gen{\U_{i'}, \ell(\ccc')}, \gen{\P_{j'}}}$ where $\ccc, \ccc' \in 2^{\ccs} \setminus \{ \emptyset \}, 0 \leq i, i', j, j' \leq \beta$ and $\level(\ccc) \leq j, \level(\ccc') \leq j', i \in \{j, j+1\},  i' \in \{j', j'+1\}$.
	By the above discussion, we obtain that $\D_1 \s \D_2 \in \edgeconst^*$ if and only if at least one of $\gen{\U_i, \ell(\ccc)}\s \gen{\U_{i'}, \ell(\ccc')}$, $\gen{\U_i, \ell(\ccc)} \s \gen{\P_{j'}}$, $\gen{\P_j} \s \gen{\U_{i'}, \ell(\ccc')}$, $\gen{\P_j} \s \gen{\P_{j'}}$ is contained in $\edgeconst_{\re}$, which in turn happens if we have
	\begin{enumerate}
		\item $\U_i \s \U_{i'} \in \edgeconst_\Delta(z)$ and $\U_i \s \ell(\ccc') \in \edgeconst_\Delta(z)$ and $\ell(\ccc) \s \U_{i'} \in \edgeconst_\Delta(z)$ and $\ell(\ccc) \s \ell(\ccc') \in \edgeconst_\Delta(z)$, or
		\item $\U_i \s \P_{j'} \in \edgeconst_\Delta(z)$ and $\ell(\ccc) \s \P_{j'} \in \edgeconst_\Delta(z)$, or
		\item $\P_j \s \U_{i'} \in \edgeconst_\Delta(z)$ and $\P_j \s \ell(\ccc') \in \edgeconst_\Delta(z)$, or
		\item $\P_j \s \P_{j'} \in \edgeconst_\Delta(z)$.
	\end{enumerate}
	If $j = 0$, resp.\ $j' = 0$, then the statements about the configurations containing $\P_j$, resp.\ $\P_{j'}$, do not appear.
	Now, by applying \Cref{lem:edgediagpi} to each of the $9$ substatements, we obtain that $\D_1 \s \D_2 \in \edgeconst^*$ if and only if
	\begin{enumerate}
		\item $\ccc \cap \ccc' = \emptyset$, or
		\item $i < j'$, or
		\item $i' < j$,
	\end{enumerate}
	Note that we used here that $\level(\ccc) \leq i$ and $\level(\ccc') \leq i'$.
	
	In general, since each label in $\D_1 \s \D_2$ is of one of the four forms $\gen{\gen{\P_i}}$, $\gen{\gen{\U_i}}$, $\gen{\gen{\U_i, \ell(\ccc)}, \gen{\P_j}}$, $\gen{\gen{X}}$, we obtain the lemma by applying an analogous argumentation to each pair of forms (which will yield the restrictions for $\ccc, \ccc', i, i', j, j'$ such that $\D_1 \s \D_2 \in \edgeconst^*$ where $\D_1$ and $\D_2$ are of the chosen forms).
\end{proof}

We now prove that $\Pi_\Delta(\prefix(z))$ is a relaxation of $\rere(\re(\Pi_\Delta(z)))$.
\begin{lemma}\label{lem:finallemma}
	Assume that $|z| \le \Delta$ if $\len(z) = 0$, and $|z| \le \Delta-1$ otherwise, and let $T$ be some positive integer.
	If $\Pi_\Delta(z)$ can be solved in $T$ rounds in the deterministic port numbering model, then $\Pi(z')$ can be solved in $T - 1$ rounds, where $z' = \prefix(z)$.
\end{lemma}
\begin{proof}
	Let $\ccs$ denote the set of colors in $\Pi_\Delta(z)$.
	We identify the set $\ccs'$ of colors in $\Pi_\Delta(z')$ with the set of pairs $(C,i)$ satisfying $C \in \ccs$ and $\level(C) \leq i \leq \beta$.
	For some color $(C,i) \in \ccs'$, set $\level((C,i)) = i$.
	This definition ensures that, for each $0 \leq j \leq \beta$, there are precisely $z'_j = z_0 + \dots + z_j$ many colors of level $j$ in $\Pi_\Delta(z')$, as required.

	Assume that $\Pi_\Delta(z)$ can be solved in $T$ rounds.
	Then, by \Cref{lem:relaxingworks}, there is some $(T-1)$-round algorithm $\fA$ solving $\Pi^*$.
	We define a new $(T-1)$-round algorithm $\fA'$ as follows.
	Each node $v$ first executes algorithm $\fA$ (in $T-1$ rounds).
	Then (without any further communication) $v$ replaces, for each incident half-edge $h$, the label $\ell$ that $\fA$ produced on $h$ by a new label $\ell'$ according to the following rules.
	
	If $\ell = \gen{\gen{\X}}$, then $\ell' := \X$.
	If $\ell = \gen{\gen{\U_i}}$ for some $1 \leq i \leq \beta$, then $\ell' := \U_i$.
	If $\ell = \gen{\gen{\P_i}}$ for some $1 \leq i \leq \beta$, then $\ell' := \P_i$.
	If $\ell = \gen{\gen{\U_i, \ell(\ccc)}, \gen{\P_j}}$ for some $\ccc \in 2^\ccs \setminus \{\emptyset\}, 0 \leq i,j \leq \beta$ satisfying $\level(\ccc) \leq j$ and $i \in \{ j, j+1 \}$, then $\ell' := \ell(\ccc^{\newy})$ where $\ccc^{\newy} := \{ (C,j) \mid C \in \ccc \} \cup \{ (C,i) \mid C \in \ccc  \}$.
	(Here, as before, we set $\U_0 := \X$ and $\gen{\gen{\U_i, \ell(\ccc)}, \gen{\P_0}} := \gen{\gen{\U_i, \ell(\ccc)}}$.)

	Finally, consider only those nodes $v$ such that there is a half-edge incident to $v$ on which $\ell' = \ell(\ccc^{\newy})$ for some $\ccc^{\newy} \in 2^{\ccs'} \setminus \{\emptyset\}$.
	(We remark that, by the design of the replacement rules and the definition of $\nodeconst^*$, for any such node $v$, each incident half-edge is labeled with either $\ell(\ccc^{\newy})$ or $\X$ at this point.)
	Each such node $v$ selects $|\ccc^{\newy}| - 1 - x$ half-edges currently labeled with $\ell(\ccc^{\newy})$ and replaces the labels on those half-edges with the label $\X$.
	Here $x$ denotes the number of half-edges incident to $v$ that were already labeled with $\X$.
	(We will see a bit later that $|\ccc^{\newy}| - 1 - x$ is always nonnegative.)
	This concludes the description of $\fA'$.

	We claim that $\fA'$ solves $\Pi_\Delta(z')$.
	We first consider the produced node configurations.
	Consider some node $v$, and let $N = \ell_1 \s \dots \s \ell_\Delta$ denote the node configuration on the half-edges incident to $v$ produced by $\fA'$. 
	The replacement rules, together with the definition of $\nodeconst^*$ and the correctness of $\fA$, imply that $N \in \nodeconst_\Delta(z')$, provided that $\fA'$ is well-defined.
	Hence, the only thing that we have to show is that ``$|\ccc^{\newy}| - 1 - x$ is always nonnegative'' as claimed in the algorithm description.

	Thus, assume that there is a half-edge $h$ incident to $v$ that, before the final step of the algorithm, was labeled with $\ell' = \ell(\ccc^{\newy})$ for some $\ccc^{\newy} \in 2^{\ccs'} \setminus \{\emptyset\}$.
	Moreover, let $x$ be as defined above.
	By the design of the replacement rules, we see that there are some $\ccc \in 2^\ccs \setminus \{\emptyset\}, 0 \leq i,j \leq \beta$ such that the output of $\fA$ on $h$ was $\ell = \gen{\gen{\U_i, \ell(\ccc)}, \gen{\P_j}}$.
	Furthermore, by the definition of $\nodeconst^*$, we have $x = |\ccc| + i - j - 1$, i.e., $x = |\ccc| - 1$ if $i = j$, and $x = |\ccc|$ if $i = j + 1$.
	Observe that the definition of $\ccc^{\newy}$ (in the replacement rule) ensures that $|\ccc^{\newy}| \geq |\ccc|$ if $i = j$, and $|\ccc^{\newy}| \geq |\ccc| + 1$ if $i = j + 1$.
	Hence, in either case $|\ccc^{\newy}| - 1 - x \geq 0$, and we obtain that each node configuration produced by $\fA'$ is contained in $\nodeconst_\Delta(z')$, as desired.
	
	Now consider the edge configurations produced by $\fA'$.
	Consider an arbitrary edge $e$ and let $D = \D_1 \s \D_2$ denote the edge configuration produced by $\fA'$ on the two half-edges $h_1, h_2$  belonging to $e$ (where output $\D_1$ is given on half-edge $h_1$).
	The replacement rules, together with \Cref{lem:newedgestuff} clearly imply that, if $\D_1, \D_2 \in \{ \U_1, \dots, \U_\Delta \} \cup \{ \P_1, \dots, \P_\Delta\}$, $\D_1 = \X$, or $\D_2 = \X$, then $D \in \edgeconst_\Delta(z')$.
	Hence, assume w.l.o.g., that $\D_1 = \ccc^{\newy}$ for some $\ccc^{\newy} \in 2^{\ccs'} \setminus \{\emptyset\}$, and that $\D_2$ is some color set from $2^{\ccs'} \setminus \{\emptyset\}$, some $\U_{i'}$ or some $\P_{j'}$.
	Let $\gen{\gen{\U_i, \ell(\ccc)}, \gen{\P_j}}$ be the output of $\fA$ on $h_1$, which implies that $\level(\ccc^{\newy}) = i$.
	If $\D_2 = \U_{i'}$ for some $1 \leq i' \leq \beta$, then $D \in \edgeconst_\Delta(z')$, simply by the definition of $\edgeconst_\Delta(z')$.

	If $\D_2 = \P_{j'}$ for some $1 \leq j' \leq \beta$, then we know that the output label of $\fA$ on $h_2$ is $\gen{\gen{\P_{j'}}}$.
	By \Cref{lem:newedgestuff}, we obtain $i < j'$, which together with the aforementioned $\level(\ccc^{\newy}) = i$ (and the definition of $\edgeconst_\Delta(z')$) implies $D \in \edgeconst_\Delta(z')$.

	Finally, assume that $\D_2 =  \ccc^{\newy}_2$, and let $\gen{\gen{\U_{i'}, \ell(\ccc')}, \gen{\P_{j'}}}$ be the output of $\fA$ on $h_2$.
	By \Cref{lem:newedgestuff}, we have $\ccc \cap \ccc' = \emptyset$, or $i < j'$, or $i' < j$.
	By the definitions of $\ccc^{\newy}$ and $\ccc^{\newy}_2$, we have that, if $\ccc \cap \ccc' = \emptyset$, or $i < j'$, or $i' < j$, then $\ccc^{\newy} \cap \ccc^{\newy}_2 = \emptyset$.
	Hence, $D \in \edgeconst_\Delta(z')$, and we conclude that the output produced by $\fA'$ is indeed a correct solution for $\Pi_\Delta(z)$.
\end{proof}

\subsection{Zero Round Solvability}\label{ssec:zero}
We now characterize for which parameters the problems in our family cannot be solved in $0$ rounds.
\begin{lemma}\label{lem:pnno0}
	Assume that $|z| \le \Delta$. The problem $\Pi_{\Delta}(z)$ is not $0$-round solvable in the deterministic port numbering model if the port numbering is unconstrained.
\end{lemma}
\begin{proof}
	Since for a $0$-round algorithm the view of each node is the same, then all nodes must output the same configuration. We show that, for all the allowed configurations of the node constraints it holds that they are not ``self compatible'', meaning that they contain two labels $\ell_1$, $\ell_2$, that are not edge-compatible, that is, $\ell_1 \s \ell_2 \notin \edgeconst_\Delta(z)$, implying that for some port numbering assignment the $0$ round algorithm must fail. 
	
	The allowed node configurations are the following:
	\begin{itemize}
		\item $\ell(\ccc)^{\Delta-x} \s \X^x$, for each $\ccc \in 2^\ccs \setminus  \{\emptyset\}$, where $x = |\ccc|-1$. 
		\item $\P_i \s \U_i^{\Delta-1}$, for each $1 \le i \le \beta$. 
	\end{itemize}
	Consider the configurations of the first type. Since $|z| \le \Delta$, then $|\ccs| \le \Delta$, implying that $x = |\ccc|-1 \le \Delta-1$, and hence that $\ell(\ccc)$ appears at least once in the configuration. Since, for all $\ccc \in 2^\ccs \setminus  \{\emptyset\}$ it holds that $\ccc \cap \ccc \neq \emptyset$, then $\ell(\ccc)$ is not edge-compatible with itself.
	Consider the configurations of the second type. For all $i$, $\P_i$ is not edge-compatible with itself.
\end{proof}

\subsection{A Problem Sequence}\label{ssec:sequence}
We now combine \Cref{lem:pn1step} and \Cref{lem:pnno0} to show that there exists a problem sequence that satisfies some desirable properties. Intuitively, if a problem in the sequence is $\Pi_\Delta(z)$, then the next problem is obtained by taking the prefix sum of $z$. If we apply this procedure many times, and we start from a  ``small'' vector it will take some steps before obtaining a ``large'' vector. The number of steps required to reach a vector $z'$ such that $|z'| \ge \Delta$ (or $> \Delta$ if $\len(z') = 0$) is a lower bound on the length of the sequence.
 
\begin{lemma}\label{lem:pnlowerbound}
	There exists a problem sequence $\Pi_0 \rightarrow \Pi_1 \rightarrow \ldots \rightarrow \Pi_t$, where $\Pi_1 = \Pi_{\Delta}(z)$, such that, for all $0 \le i < t$, the following holds:
	\begin{itemize}
		\item There exists a problem $\Pi'_i$ that is a relaxation of $\re(\Pi_i)$;
		\item $\Pi_{i+1}$ is a relaxation of $\rere(\Pi'_i)$;
		\item The number of labels of $\Pi_i$, and the ones of $\Pi'_i$, are upper bounded by $2^\Delta(1 + \len(z))$.
	\end{itemize}
	Also, $\Pi_t$ has at most $2^\Delta(1 + \len(z))$ labels and is not $0$-round solvable in the deterministic port numbering model if the port numbering assignment is unconstrained. Let $\beta = \len(z)$. The value of $t$ is at least the maximum $i$ such that $|\prefix^i(z)| < \Delta$ if $\len(z) > 0$ and such that $|\prefix^i(z)| \le \Delta$ if $\len(z) = 0$. If there is no maximum, then $t$ can be made arbitrarily large.
\end{lemma}
\begin{proof}
	We set $\Pi_i = \Pi_{\Delta}(\prefix^i(z))$, where $\prefix^i(z)$ is the result obtained by recursively applying the function $\prefix$ for $i$ times, and $\prefix^0(z)=z$. By \Cref{lem:pn1step} the problem sequence satisfies the requirements. We need to prove for which values of $t$, $\Pi_t$ is not $0$-rounds solvable. By \Cref{lem:pnno0} it is sufficient to satisfy $|\prefix^t(z)| \le \Delta$, which is a condition required by the statement.
\end{proof}

\section{Lifting Results to the \textsf{LOCAL} Model}\label{sec:lifting}
In this section we show how to lift the results that we obtained for the deterministic port numbering model to the \LOCAL model. For this purpose, we heavily rely on ideas already used in the literature~\cite{Balliu2019, balliurules, trulytight, binary}, by slightly improving some parts to support our use case. While in this work we only consider graphs, we show a general statement that works on hypergraphs as well.

In the following, we assume w.l.o.g.\ that $\delta$, the rank of the hyperedges, is at most $\Delta$, the degree of the nodes.  If not, we can just switch the roles of nodes and hyperedges, and the same claim still holds (up to an additive constant factor).

\begin{theorem}\label{thm:lifting}
	Let $\Pi_0 \rightarrow \Pi_1 \rightarrow \ldots \rightarrow \Pi_t$ be a sequence of problems.
	Assume that, for all $0 \le i < t$, and for some function $f$, the following holds:
	\begin{itemize}
		\item There exists a problem $\Pi'_i$ that is a relaxation of $\re(\Pi_i)$;
		\item $\Pi_{i+1}$ is a relaxation of $\rere(\Pi'_i)$;
		\item The number of labels of $\Pi_i$, and the ones of $\Pi'_i$, are upper bounded by $f(\Delta)$.
	\end{itemize}
	Also, assume that $\Pi_t$ has at most $f(\Delta)$ labels and is not $0$-round solvable in the deterministic port numbering model, even if the port numbering assignment satisfies some local constraints $C^{\mathrm{port}}$. Then, $\Pi_0$ requires $\Omega(\min\{t, \log_\Delta n - \log_\Delta \log f(\Delta)\})$ rounds in the deterministic \LOCAL model and $\Omega(\min\{t, \log_\Delta \log n - \log_\Delta \log f(\Delta)\})$ rounds in the randomized \LOCAL model, even if the port numbering satisfies some local constraints $C^{\mathrm{port}}$.
\end{theorem}
Since all the techniques required to prove \Cref{thm:lifting} are already present in the literature, and we only need to improve the analysis to support an arbitrary large number of labels, we defer the proof of \Cref{thm:lifting} to \Cref{apx:lifting}.

Note that, by proving that a problem can be relaxed to a fixed point, then we obtain an arbitrary long sequence of problems that satisfy the requirements. Hence, if the number of labels used in the fixed point problem is not too large, then we get that this problem requires $\Omega(\log_\Delta n)$ rounds for deterministic algorithms, and $\Omega(\log_\Delta \log n)$ rounds for randomized ones.

\section{Lower Bound Results}\label{sec:corollaries}
All the results that we present in this section hold already on $\Delta$-regular trees. We start by showing a general result, and then we instantiate it to obtain results for a number of \emph{natural} problems. By combining \Cref{lem:pnlowerbound} with \Cref{thm:lifting}, we obtain the following theorem.
\begin{restatable}{theorem}{lbgeneral}\label{thm:general}
	Let $z$ be a vector $[z_0, \ldots, z_\beta]$ such that $\beta \le 2^\Delta$. Let $t$ be the maximum value such that $|\prefix^t(z)| \le \Delta$ if $\len(z) = 0$, and such that  $|\prefix^t(z)| < \Delta$ if $\len(z) > 0$. If there is no maximum, let $t = \infty$. 
	Then, $\Pi_{\Delta}(z)$ requires $\Omega(\min\{t, \log_\Delta n\})$ rounds in the deterministic \LOCAL model and $\Omega(\min\{t, \log_\Delta \log n\})$ rounds in the randomized \LOCAL model.
\end{restatable}
Notice that the assumption on $\beta$ is only required in order to make the total number of labels small enough so that $\log_\Delta\log f(\Delta)=O(1)$, where $f(\Delta)$ is an upper bound on the number of labels of the problems in the sequence.

\paragraph{Arbdefective Colorings.}
We start by proving that $\Pi_{\Delta}([\Delta])$ is a relaxation of some variants of arbdefective coloring, including the $\Delta$-coloring problem itself.
\begin{restatable}{theorem}{lbgeneralizedarbdef}\label{cor:generalizedarbdef}
	Let $\vec{d} = (d_1, \ldots, d_c)$ be an arbdefect vector that is not $1$-relaxed. Then, the $\vec{d}$-arbdefective $c$-coloring problem requires $\Omega(\log_\Delta n)$ rounds in the deterministic \LOCAL model and $\Omega(\log_\Delta \log n)$ rounds in the randomized \LOCAL model.
\end{restatable}
\begin{proof}
	First of all, notice that $\prefix([\Delta]) = [\Delta]$, implying that for $\Pi_{\Delta}([\Delta])$ we can apply \Cref{thm:general} with $t= \infty$. We show that, if $\vec{d}$ is not $1$-relaxed, then the $\vec{d}$-arbdefective $c$-coloring problem can be relaxed to $\Pi_{\Delta}([\Delta])$.
	Since by assumption $\vec{d}$ is not $1$-relaxed, by definition we have that:
	\[
	\capa(\vec{d})=\sum_{i=1}^c (d_i+1) \le \Delta.
	\]
	Consider a color space $S$ of size $\sum_{i=1}^c (d_i+1)$. We partition the colors of $S$ into $c$ sets, where each set $\ccc_i$ contains $d_i+1$ colors. For each set $\ccc_i$, the problem $\Pi_{\Delta}([\Delta])$ contains the configuration $\ell(\ccc_i)^{\Delta-x} \s \X^x$, where $x = |\ccc_i|-1 = d_i$. Also, since all sets $\ccc_i$ are pairwise disjoint, then $\ell(\ccc_i) \s \ell(\ccc_j)$ is contained in the edge constraint, for all $i \neq j$. Moreover, $\X$ is compatible with every other label. Hence, given a solution for arbdefective coloring, nodes of color $i$ can label $\ell(\ccc_i)$ their edges, and then mark their arbdefective edges with $\X$. The number of arbdefective edges is at most $d_i$. Nodes can then mark other arbitrary edges to make the number of $\X$ exactly $d_i$. In this way, they obtain a solution for  $\Pi_{\Delta}([\Delta])$.
\end{proof}
By allowing the same arbdefect to each color, we obtain the following corollary.
\begin{corollary}\label{cor:arbdef}
	The $d$-arbdefective $c$-coloring problem requires $\Omega(\log_\Delta n)$ rounds in the deterministic \LOCAL model and $\Omega(\log_\Delta \log n)$ rounds in the randomized \LOCAL model, if $(d+1)c \le \Delta$.
\end{corollary}
Also, by setting $c=\Delta$ and $d=0$, we obtain the following.
\begin{corollary}\label{cor:lbdeltacol}
	The $\Delta$-coloring problem requires $\Omega(\log_\Delta n)$ rounds in the deterministic \LOCAL model and $\Omega(\log_\Delta \log n)$ rounds in the randomized \LOCAL model.
\end{corollary}

\paragraph{Arbdefective Colored Ruling Sets.}
We now consider $\alpha$-arbdefective $c$-colored $\beta$-ruling sets. Before showing our results, we prove some useful lemmas.
We start by showing a lemma proved in~\cite{balliurules}.

\begin{lemma}[Lemma 16 of \cite{balliurules}, arXiv version]\label{lem:prefix}
	Assume $z = [1,0,\ldots,0]$ and $k \le \beta = \len(z)$. For all $j \ge 1$, $\prefix^j(v)_k = {j+k-1 \choose k}$.
\end{lemma}

We now show a lower bound for $\Pi_\Delta(z)$ for the case in which only the first element of $z$ is non-zero.
\begin{lemma}\label{lem:firstnonzero}
	Assume $z = [k,0,\ldots,0]$, where $k\ge 1$, and $\beta = \len(z)\le 2^\Delta$. 
	Let $t$ be the maximum value such that ${t + \beta \choose \beta} < \Delta / k$.
	The problem $\Pi_{\Delta}(z)$ requires $\Omega(\min\{t, \log_\Delta n\})$ rounds in the deterministic \LOCAL model and $\Omega(\min\{t, \log_\Delta \log n\})$ rounds in the randomized \LOCAL model.
\end{lemma}
\begin{proof}
	Note that $|\prefix^t(z)| = \prefix^{t+1}(z)_\beta$, and $\prefix^{t+1}([k,0,\ldots,0]) = k \cdot \prefix^{t+1}([1,0,\ldots,0])$. Hence, $\prefix^{t+1}(z)_\beta  = \prefix^{t+1}([k,0,\ldots,0])_\beta  = k \cdot \prefix^{t+1}([1,0,\ldots,0])_\beta =  k {t + \beta \choose \beta}$, where the last equality holds by \Cref{lem:prefix}. Hence, since by assumption $t$ is the maximum value such that $k {t + \beta \choose \beta} < \Delta$, then, by \Cref{thm:general}, the claim follows.
\end{proof}
Note that $\Pi_\Delta([1,0,\ldots,0])$ is a relaxation of a variant of $(2,\beta)$-ruling sets where nodes also need to know an upper bound on the distance, and the direction, towards a ruling set node.
For this variant of ruling sets, in \cite[Lemma 15 (arXiv version)]{balliurules} it has been shown that it can be solved in $t$ rounds, where $t$ is the minimum value such that ${t + \beta \choose \beta} \ge c$, if nodes are provided a $c$-coloring in input. \Cref{lem:firstnonzero} implies that this result is tight, at least in the case in which a $(\Delta + 1)$-coloring is given in input to the nodes: it is not possible to solve this variant of ruling sets in $t$ rounds, if ${t + \beta \choose \beta} < \Delta$. In fact, from \Cref{lem:pnlowerbound} we actually get a lower bound of $t+1$ rounds, where $t$ is the maximum value such that ${t + \beta \choose \beta} < \Delta$.
 Hence, our result is tight in the case in which a $(\Delta+1)$-coloring is given to the nodes. Note that, improving this result for the case in which a larger coloring is given in input would imply that a $(\Delta+1)$-vertex coloring cannot be computed fast, but knowing whether $(\Delta+1)$-vertex coloring requires  $\omega(\log^* n)$ is a long standing open question.

 We now use the result of \Cref{lem:firstnonzero} to show lower bounds for arbdefective colored ruling sets. In particular, we show that a lower bound of $t$ rounds for $\Pi([c(1+\alpha),0,\ldots,0])$ implies a lower bound of $t-\beta$ rounds for $\alpha$-arbdefective $c$-colored $\beta$-ruling sets.
 
\lbcoldom*
 \begin{proof}
	Assume $z = [c(1+\alpha),0,\ldots,0]$ and $\len(z)=\beta$. 
	By \Cref{lem:firstnonzero}, the problem $\Pi_\Delta(z)$ requires $\Omega(\min\{t, \log_\Delta n\})$ rounds in the deterministic \LOCAL model and $\Omega(\min\{t, \log_\Delta \log n\})$ rounds in the randomized \LOCAL model, where $t$ is the maximum value such that $c(1+\alpha) {t + \beta \choose \beta} < \Delta$. 
	
	We show that, given a solution for the $\alpha$-arbdefective $c$-colored $\beta$-ruling set problem, nodes can spend $\beta$ rounds to solve $\Pi_{\Delta}(z)$. Hence, a lower bound of $t$ rounds for $\Pi_{\Delta}(z)$ implies a lower bound of $t-\beta$ for the $\alpha$-arbdefective $c$-colored $\beta$-ruling set problem.	
	In $\beta$ rounds, uncolored nodes can find the nearest colored node $v$, and nodes at distance $i$ from $v$ can output $\P_i$ on their port towards $v$ and $\U_i$ on all the other ports. Then, since the color space contains $c(1+\alpha)$ colors, then there must exist $c$ disjoint sets $\ccc_i$ of colors of size $1+\alpha$. For each set $\ccc_i$, there must exist an allowed configuration of the form $\ell(\ccc_i^{\Delta-x}) \s \X^x$, for $x = \alpha$. Nodes of color $i$ can label their edges $\ell(\ccc_i)$, and then mark their edges oriented towards other nodes of the same color with $\X$. The number of edges marked $\X$ is at most $\alpha$. Nodes can then mark other arbitrary edges to make the number of $\X$ exactly $\alpha$. It is easy to check that this labeling satisfies the constraints of $\Pi_{\Delta}(z)$.
\end{proof}
We now provide some more human-friendly versions of \Cref{thm:lbarbcolrs}. We start by making an observation.
\begin{observation}\label{obs:larget}
	For all $t \le  \frac{\beta}{2e} x^{1/\beta}$ and $\beta \le \frac{1}{2e} \log x$, it holds that ${t + \beta \choose \beta} \le x$. 
\end{observation}
\begin{proof}
	\[{t + \beta \choose \beta} \le \left(\frac{e(t + \beta)}{\beta}\right)^\beta \le  \left(\frac{e(\frac{\beta}{2e} x^{1/\beta} + \frac{1}{2e} \log x)}{\beta}\right)^\beta \le  \left(\frac{e(\frac{\beta}{e} x^{1/\beta} )}{\beta}\right)^\beta = x
	\]
\end{proof}

We are now ready to provide an explicit lower bound for $\alpha$-arbdefective $c$-colored $\beta$-ruling sets as a function of $\Delta$ and $n$.
\begin{corollary}\label{cor:arbdefrsdelta}
	There exists a constant $\eps_0>0$ such that the $\alpha$-arbdefective $c$-colored $\beta$-ruling set problem requires $\Omega\big(\min\big\{\beta (\frac{\Delta}{c(1+\alpha)})^{1/\beta}, \log_\Delta n\big\}\big)$ rounds, for $\beta \le \eps_0 \min\big\{\log \frac{\Delta}{c(1+\alpha)}, \log_\Delta n\big\}$, in the deterministic \LOCAL model, and it requires $\Omega\big(\min\big\{\beta (\frac{\Delta}{c(1+\alpha)})^{1/\beta}, \log_\Delta \log n\big\}\big)$ rounds, for $\beta \le \eps_0 \min\big\{\log \frac{\Delta}{c(1+\alpha)}, \log_\Delta \log n\big\}$, in the randomized \LOCAL model.
\end{corollary}
\begin{proof}
	Let $t$ be the maximum value such that ${t + \beta \choose \beta} < \frac{\Delta}{c(1+\alpha)}$. By \Cref{obs:larget} and the assumption on $\beta$, $t = \Omega(\beta (\frac{\Delta}{c(1+\alpha)})^{1/\beta})$. By the assumption on $\beta$, $t-\beta = \Omega(t)$. Also, by the assumption on $\beta$ in the case of deterministic algorithms, $\log_\Delta n - \beta = \Omega(\log_\Delta n)$, and in the case of randomized algorithms, $\log_\Delta \log n - \beta = \Omega(\log_\Delta \log n)$. By \Cref{thm:lbarbcolrs}, the claim follows.
\end{proof}
\begin{corollary}\label{cor:arbdefrsn}
	There exists some constant $\varepsilon>0$ such that the $\alpha$-arbdefective $c$-colored $\beta$-ruling set problem requires $\Omega(\frac{\log n}{\log (1+\alpha) + \log c + \beta \log \log n})$ rounds for $\beta \le \varepsilon \sqrt{\frac{\log n}{\log (1+\alpha) + \log c + \log \log n}}$ in the deterministic \LOCAL model, and $\Omega(\frac{\log \log n}{\log (1+\alpha) + \log c +\beta \log \log \log n})$ rounds for $\beta \le \varepsilon \sqrt{\frac{\log \log n}{\log (1+\alpha) + \log c +\log \log \log n}}$ in the randomized \LOCAL model.
\end{corollary}
\begin{proof}
	Apply \Cref{cor:arbdefrsdelta} by picking $\Delta \approx c(1+\alpha)(\frac{\log n}{\log \log n})^\beta$ in the deterministic case, and $\Delta \approx c(1+\alpha)(\frac{\log \log n}{\log \log \log n})^\beta$ in the randomized one.
\end{proof}

By setting $\alpha = 0$ and $c=1$, from \Cref{thm:lbarbcolrs} and \Cref{cor:arbdefrsdelta,cor:arbdefrsn} we obtain lower bounds for $(2,\beta)$-ruling sets.
\begin{corollary}\label{cor:rulingsets}
	Let $t$ be the maximum value such that ${t + \beta \choose \beta} < \Delta$, where $\beta\le2^\Delta$.
	The $(2,\beta)$-ruling set problem requires $\Omega(\min\{t, \log_\Delta n\} - \beta)$ rounds in the deterministic \LOCAL model and $\Omega(\min\{t, \log_\Delta \log n\} - \beta)$ rounds in the randomized \LOCAL model.
	
	In particular, there exists some constant $\varepsilon>0$ such that
	the $(2,\beta)$-ruling set problem requires  $\Omega(\min\{\beta \Delta^{1/\beta}, \log_\Delta n\})$ rounds for $\beta \le \varepsilon \min\{\log \Delta, \log_\Delta n\}$ in the deterministic \LOCAL model, and $\Omega(\min\{\beta \Delta^{1/\beta}, \log_\Delta \log n\})$ rounds for $\beta \le \varepsilon \min\{\log \Delta, \log_\Delta \log n\}$ in the randomized one.
	
	If we express the complexity of the problem solely as a function of $n$, then it requires $\Omega(\frac{\log n}{\beta \log \log n})$ rounds for $\beta \le \varepsilon \sqrt{\frac{\log n}{\log \log n}}$ for deterministic algorithms, and $\Omega(\frac{\log \log n}{\beta \log \log \log n})$ rounds for $\beta \le \varepsilon \sqrt{\frac{\log \log n}{\log \log \log n}}$ for randomized ones.
\end{corollary}
Then, by setting $\beta=1$, we obtain a lower bound for MIS.
\begin{restatable}{theorem}{lbmis}
	The maximal independent set problem requires  $\Omega(\min\{\Delta, \log_\Delta n\})$ rounds in the deterministic \LOCAL model and $\Omega(\min\{\Delta, \log_\Delta \log n\})$ rounds in the randomized \LOCAL model.
	If we express the complexity of the problem solely as a function of $n$, then it requires $\Omega(\frac{\log n}{\log \log n})$ rounds for deterministic algorithms and $\Omega(\frac{\log \log n}{\log \log \log n})$ rounds for randomized ones.
\end{restatable}
In particular, we get that the maximal independent set problem requires $\Omega(\log n / \log \log n)$ deterministic rounds \emph{on trees}, which is tight, as in \cite{BarenboimE10} it has been shown that MIS on trees can be solved in $O(\log n / \log\log n)$ deterministic rounds.

\paragraph{More Results.}
Until now, we only considered cases where the vector $z$ has non-zero values only in the first position.
We conclude the section by showing that \Cref{thm:general} can be used to prove lower bounds for even more variants of ruling sets and arbdefective colorings.

Consider the following problem. We have two disjoint color spaces $\ccs^1$ and $\ccs^2$, one containing $4$ colors and the other containing $5$ colors. Nodes can either be colored, or uncolored.
Uncolored nodes must have either at least one node colored from the space  $\ccs^1$  at distance at most $3$, or at least one node colored from the space  $\ccs^2$ at distance at most $2$, or both.
Nodes colored from the space  $\ccs^1$  must have arbdefect at most $1$, while nodes colored from the space $\ccs^2$ must have arbdefect at most $2$. We show how to use \Cref{thm:general} to prove that this problem requires $\Omega(\Delta^{1/3})$ rounds.
\begin{corollary}
	The example problem requires $\Omega(\min\{\Delta^{1/3}, \log_\Delta n\})$ rounds in the deterministic \LOCAL model and $\Omega(\min\{\Delta^{1/3}, \log_\Delta \log n\})$ rounds in the randomized \LOCAL model.
\end{corollary}
\begin{proof}
	Consider the problem $\Pi_{\Delta}(z)$ where $z = [8,15,0,0]$. We show that, given a solution for the example problem, nodes can solve in $O(1)$ rounds the problem $\Pi_{\Delta}(z)$.
	
	Uncolored nodes can spend $3$ rounds to see the color of the nearest colored node. Nodes at distance $i$, can output $\P_i$ on the port towards the colored node if the colored node has a color from $\ccs^1$, or $\P_{i+1}$ if the colored node has a color from $\ccs^2$, and if both cases apply, nodes from $\ccs^1$ are prioritized. Note that the assumptions guarantee that only labels from $\{\P_1,\P_2,\P_3\}$ are required. Uncolored nodes that output $\P_i$ on one port, then output $\U_i$ on all the other ports.
	
	Colored nodes proceed as follows. There are $8$ colors in level $0$, and $15$ colors in level $1$. We partition the $8$ colors in $4$ groups of $2$ colors, $\{\ccc^1_{i} ~|~ 1 \le i \le 4\}$, and we partition the $15$ colors in $5$ groups of $3$ colors, $\{\ccc^2_{i} ~|~ 1 \le i \le 5\}$. In $\Pi_{\Delta}(z)$ there must exist the following configurations:
	$\ell(\ccc^1_i)^{\Delta-1} \s \X$ such that $\level(\ccc^1_i) = 0$ for $i \in \{1,2,3,4\}$, and $\ell(\ccc^2_i)^{\Delta-2} \s \X^2$ such that $\level(\ccc^2_i) = 1$ for $i \in \{1,2,3,4,5\}$, such that the intersection between each distinct pair of color sets is disjoint. As in the arbdefective colored ruling set case, nodes can map their colors and put $\X$ on the right edges, such that the problem is solved correctly. 
	
	For the time lower bound, note that:
	\[
	\prefix^t([8,15,0,0]) \le \prefix^t([15,15,15,15]) \le 15 \cdot \prefix^t([1,1,1,1]) \le 15 \cdot \prefix^{t+1}([1,0,0,0]).
	\]
	Hence, by \Cref{lem:prefix}, \Cref{obs:larget}, and \Cref{thm:general}, the claim follows.
\end{proof}

\input{algorithms}

\section{Open Problems}\label{sec:open}
There is a natural distinction between locally checkable problems: those who are solvable in $O(f(\Delta) + \log^* n)$, that is, \emph{easy} problems, and those who require $\Omega(\log_\Delta n)$ for deterministic algorithms and $\Omega(\log_\Delta \log n)$ for randomized ones, that is, \emph{hard} problems. We also know from prior work that, for bounded-degree graphs, there cannot be problems in between \cite{chang16exponential,ChangP19}.
In this work, we proved lower bounds for many easy natural problems, by exploiting the hardness of hard problems, namely, $c$-coloring for $c \le \Delta$. 
For this purpose, we exploited the fact that there exists a fixed point proof for the hardness of $\Delta$-coloring. But what if there were no such proof? It would have been much more difficult to prove a tight lower bound for ruling sets and other problems. Hence, a first question that we ask is whether all hard problems admit a proof based on fixed points.
\begin{restatable}{oq}{oqfixedpoint}
	Prove or disprove that, for all locally checkable problems with deterministic complexity $\Omega(\log_\Delta n)$ in the \LOCAL model, there exists a relaxation that is a non-trivial fixed point.
\end{restatable}

Another interesting question is how broadly our technique can be applied. By applying the round elimination technique to $O(f(\Delta)+\log^*n)$-round solvable problems, we typically obtain problem sequences where the size of the problems grows really fast. Often, it seems that \emph{most} of the configurations allowed in the problems of the sequence do not actually make the problem easier; they just make the description of the problem more complicated, similarly to what happens with ruling sets, where these configurations are related to $c$-coloring. It would be very interesting to find other cases where we can get rid of these ``useless'' allowed configurations by constructing some fixed point problem that we can embed into the description of the problem to make these configurations disappear.
\begin{oq}
	Can we find other interesting problems that can be solved in $O(f(\Delta) + \log^* n)$ rounds for which we can prove lower bounds by exploiting the hardness of problems that require $\Omega(\log_\Delta n)$?
\end{oq}

Also, we now have three different lower bound proofs based on round elimination for the maximal independent set problem on trees. The first one appeared in \cite{balliurules}; in that proof the number of labels of the problems in the lower bound sequence grows linearly and is always upper bounded by $O(\Delta)$, and with similar techniques the authors also obtain lower bounds for ruling sets. The second one appeared in \cite{BBKOmis}, and in that case the number of labels is always $O(1)$, but the proof does not extend to ruling sets. In both cases, the obtained lower bounds are exponentially worse than the lower bounds presented in this work, but here we needed to use a number of labels that is exponential in $\Delta$. We would like to understand this phenomenon, and better understand the relation between the number of labels and the achievable lower bounds. For example, can we prove the same lower bound presented in our work by using a subexponential number of labels?
\begin{oq}
	Understand the tradeoff between the number of labels and the maximum lower bound achievable with a round elimination problem sequence where each problem does not exceed this number of labels.
\end{oq}

In our work we characterized exactly which variants of arbdefective coloring are easy, and which variants of arbdefective coloring are hard. The same question is open for another important subroutine in efficient coloring algorithms---defective coloring.
\begin{oq}
	Characterize which variants of the defective coloring problem can be solved in $O(f(\Delta) + \log^* n)$ rounds, for some function $f$, and which of them require $\Omega(\log_\Delta n)$ rounds.
\end{oq}

In order to apply the round elimination technique, we need to encode problems in the node-edge formalism. Some natural problems, for example $(2,\beta)$-ruling sets, cannot be encoded directly, and we need to require nodes to produce a node-edge checkable proof that they solved the problem. We call this variant of ruling sets an \emph{encoded} version of the problem. In the encoded variant of $(2,\beta)$-ruling sets, we require nodes to output pointer chains that point towards a ruling set node, such that these chains are of length at most $\beta$. 
 Unfortunately, this means that if we obtain a lower bound of $T$ rounds for the encoded problem, then we only obtain a $(T-\beta)$-round lower bound for the natural problem. Is the encoded version really $\beta$ rounds harder than the natural one?
For the encoded version of $(2,\beta)$-ruling sets we obtained truly tight bounds, whereas for the natural version we only obtained asymptotically tight bounds, and only for $\beta$ small enough.
While we think that encoded versions are interesting per se, we would also like to know the right complexity of the natural versions. In particular, for the natural versions, our lower bounds only hold for $\beta \le \varepsilon \log \Delta$ and some small enough $\varepsilon$, and we would like to know if they also hold for larger values of $\beta$.
We conjecture that $(2,O(\log \Delta))$-ruling sets require $\Omega(\log \Delta)$ rounds.
\begin{oq}
	Understand the difference between the complexity of encoded problems and natural problems.
\end{oq}

A limit of the current version of round elimination is that the best achievable lower bound, when $\Delta \ge 3$, is $\Omega(\log_\Delta n)$ rounds for deterministic algorithms and $\Omega(\log_\Delta \log n)$ rounds for randomized ones. We would like to understand whether this is really a limit of the technique or just of our current use of the general idea behind it.
\begin{oq}
	Can we use the round elimination technique to prove lower bounds of $\omega(\log_\Delta n)$ rounds for deterministic algorithms, or of $\omega(\log_\Delta \log n)$ rounds for randomized ones?
\end{oq}
In particular, in bounded degree graphs, we do not have any example of a problem that is known to be solvable in $O(\poly \log \log n)$ rounds with randomized algorithms, but also known to require $\omega(\log n)$ rounds with deterministic ones.
\begin{oq}\label{oq:morethanlogn}
	Prove or disprove that there exists a locally checkable problem that, in bounded degree graphs, is solvable in $O(\poly \log \log n)$ rounds with randomized algorithms in the \LOCAL model, but that requires $\omega(\log n)$ rounds for deterministic algorithms.
\end{oq}
We would like to mention that a related question has been asked in 2017 by Chang and Pettie~\cite{ChangP19}. There, the authors asked whether the distributed constructive Lov\'{a}sz Local Lemma problem can be solved in $O(\log \log n)$ rounds on bounded degree graphs. It is known that answering this question in the affirmative would imply that there are no locally checkable problems in bounded degree graphs with randomized complexity between $\omega(\log \log n)$ and $o(\log n)$ \cite{ChangP19}, and it is also known that, for locally checkable problems on bounded degree graphs, an $O(\log \log n)$ randomized algorithm implies an $O(\log n)$ deterministic one \cite{chang16exponential}.

One of the few natural problems that is currently a candidate for answering \Cref{oq:morethanlogn} is $\Delta$-coloring. This is a problem for which randomness helps, and it can in fact be solved in $O(\log^2 \log n)$ rounds with randomized algorithms in bounded degree graphs~\cite{GhaffariHKM18}. For deterministic algorithms, a lower bound of $\Omega(\log n)$ rounds is known, but the best known upper bound is $O(\log^2 n)$ rounds~\cite{GhaffariHKM18,panconesi95delta}. The currently best known algorithm uses ruling sets as a subroutine, and $O(\log^2 n)$ is in fact the time that is used to compute a $(2,\log n)$-ruling set in $G^{O(\log n)}$, the $O(\log n)$-th power of $G$. While it may be possible to improve the complexity of $\Delta$-coloring by exploiting the fact that the ruling set is computed on a power graph (where, for some reason, it may be easier to compute ruling sets), we would like to understand whether computing a ruling set is really necessary for solving $\Delta$-coloring fast, or if some other (easier) symmetry breaking primitive is sufficient (which, given our lower bounds for ruling sets, might provide a more promising way of improving the complexity of $\Delta$-coloring). 
\begin{oq}
	The best known algorithm for $\Delta$-coloring has a deterministic runtime of $O(\log^2 n)$ rounds in bounded degree graphs, and uses ruling sets as a subroutine. Can we find a genuinely different algorithm for $\Delta$-coloring, not based on ruling sets, with a runtime that is at least as good?
\end{oq}

Proofs based on the round elimination technique have been getting progressively more complicated. We have automatic tools to efficiently check them for specific values of $\Delta$ \cite{Olivetti2019}, but generalizing these proofs for any value of $\Delta$ usually requires a complex analysis (see e.g.\ the proof of \Cref{lem:pn1step}). It would be nice to find ways to prove such statements in some easier way, or to entirely automate the process. Also, it would be interesting to have machine-checked versions of existing proofs.
\begin{oq}
	Can we find an automatic way to prove statements that are based on round elimination?
\end{oq}

\urlstyle{same}
\bibliographystyle{alpha}
\bibliography{fixed-point}

\appendix
\section{Proof of Theorem \ref{thm:lifting}}\label{apx:lifting}
All proofs present in this section are already present in the literature \cite{Balliu2019, balliurules, trulytight, binary}. We report them here for completeness, but we also make them slightly more general: while the existing proofs only tolerate problem sequences containing $O(\Delta^2)$ labels, we remove this limitation and we assume that the number of labels is upper bounded by $f(\Delta)$, for some arbitrary function $f$.

The standard way for lifting lower bounds, obtained through round elimination, from the port numbering model to the \LOCAL model is the following:
\begin{enumerate}
	\item Show that the round elimination technique also works for randomized algorithms for the port numbering model. In particular, show that if $\Pi$ can be solved in $T$ rounds with local failure probability at most $p$ (where the local failure probability of a node $v$ is the probability that the constraints of some problem $\Pi$ are not satisfied on $v$), then $\rere(\re(\Pi))$ can be solved in $T-1$ rounds with some local failure probability at most $p'$, for $p'$ that is not too large compared to $p$.\label{point:randomre}
	\item Show that if a problem is not $0$ rounds solvable in the deterministic port numbering model, then any $0$ round algorithm for the problem must fail with large probability in the randomized port numbering model.\label{point:zerorounds}
	\item Combine the two parts above to show that any algorithm that runs ``too fast'', that is, faster than the deterministic port numbering model lower bound, must fail with large probability.\label{point:combine}
	\item Show for which values of $T$, as a function of $n$, the obtained local failure probability is too large. This gives a lower bound on the local failure probability of randomized algorithms that run in infinite $\Delta$-regular trees in the port numbering model, as a function of their running time.\label{point:withn}
	\item Show that the last point implies a lower bound in a finite $\Delta$-regular tree in the randomized port numbering model. This directly gives a lower bound in the randomized \LOCAL model as well. \label{point:randlocal}
	\item Show that a randomized lower bound implies an even stronger deterministic lower bound.\label{point:detlocal}
\end{enumerate}

While this was not emphasized in previous works, the standard techniques used to lift lower bounds from the deterministic port numbering model to the \LOCAL model do not work if there are inputs that can be adversarially assigned. In our setting, the inputs of the nodes are
\begin{itemize}
	\item the size of the graph,
	\item the random bits,
	\item the port numbering.
\end{itemize}
In order to make the lifting work, we need to assume that the port numbering is random among the ones allowed by the port numbering constraints.

We start by discussing \Cref{point:randomre}. In \cite{binary}, the authors show that, for a random port numbering assignment, the round elimination technique also works for randomized algorithms, and they provide a generic analysis, that does not depend on a specific problem, that shows how the local failure probability evolves when applying the round elimination technique. The same result holds also if the random port number assignment comes from a locally constrained port numbering assignment. This result holds also for hypergraphs, and in that case $\Delta$ should be an upper bound on the degree of the nodes and on the rank of the hyperedges.
\begin{lemma}[Lemma 41 of \cite{binary}]\label{lem:singlestep}
	Let $A$ be a randomized $t$-round algorithm for $\Pi$ with local failure probability at most $p$ (where $t>0$). Let $\Pi'$ be a relaxation of $\re(\Pi)$. Then there exists a randomized $(t-1)$-round algorithm $A'$ for $\rere(\Pi')$ with local failure probability  $p'' \le 2^\frac{1}{\Delta+1} (\Delta |\Sigma'|)^\frac{\Delta}{\Delta+1} {p'}^\frac{1}{\Delta+1} + p'$, where $p' \le 2^\frac{1}{\Delta+1} (\Delta |\Sigma|)^\frac{\Delta}{\Delta+1} p^\frac{1}{\Delta+1} + p$ and $\Sigma'$ is the label set of $\Pi'$.
\end{lemma}
This lemma essentially states that, if there exists an algorithm $A$ that solves $\Pi$ with some local failure probability, then we can make $A$ running one round faster and solve (a relaxation of) $\rere(\re(\Pi))$ with a local failure probability that is larger, but not too large.

We can apply \Cref{lem:singlestep} multiple times, to show an upper bound on the failure probability of an algorithm that solves (a relaxation of) the problem obtained by applying $\rere(\re(\cdot))$ for $j$ times, as a function of the failure probability of an algorithm solving $\Pi$ in $t$ rounds. This lemma (and its proof) already appeared in previous works \cite{Balliu2019, balliurules, trulytight, binary}, but we now make it parametric on the number of labels $f(\Delta)$ of the problem family.

\begin{lemma}\label{lem:multiplesteps}
	Let $\Pi_0 \rightarrow \Pi_1 \rightarrow \dots \rightarrow \Pi_t$ be a sequence of problems satisfying the conditions of \Cref{thm:lifting}.
	Let $A$ be a randomized $t$-round algorithm for $\Pi_0$ with local failure probability at most $p$. Then there exists a randomized $(t-j)$-round algorithm $A'$ for $\Pi_{j}$ with local failure probability at most $(2\Delta f(\Delta))^2 p^{1/(\Delta+1)^{2j}}$, for all $0<j \le t$.
\end{lemma}
\begin{proof}
	By assumption, $|\Sigma|$ and $|\Sigma'|$ are upper bounded by $f(\Delta)$. Hence, by \Cref{lem:singlestep}, we obtain: 
	\begin{align*}
	p'' &\le 2^\frac{1}{\Delta+1} (\Delta |\Sigma'|)^\frac{\Delta}{\Delta+1} {p'}^\frac{1}{\Delta+1} + p' \\
	&\le (\Delta |\Sigma'|)^\frac{1}{\Delta+1} (\Delta |\Sigma'|)^\frac{\Delta}{\Delta+1} {p'}^\frac{1}{\Delta+1} + p'\\
	&= (\Delta |\Sigma'|) {p'}^\frac{1}{\Delta+1} + p'\\
	&\le (\Delta f(\Delta)) {p'}^\frac{1}{\Delta+1} + p'\\
	&\le (2 \Delta f(\Delta))  {p'}^\frac{1}{\Delta+1}.
	\end{align*}
	Similarly,
	\begin{align*}
	p' &\le 2^\frac{1}{\Delta+1} (\Delta |\Sigma|)^\frac{\Delta}{\Delta+1} {p}^\frac{1}{\Delta+1} + p \\
	&\le(\Delta |\Sigma|) {p}^\frac{1}{\Delta+1} + p \\
	&\le (2\Delta f(\Delta))  {p}^\frac{1}{\Delta+1}.
	\end{align*}
	Summarizing, we get that:
	\begin{align*}
	p''&\le (2\Delta f(\Delta)) p'^{\frac{1}{\Delta+1}}, \text{ where }\\
	p' &\le (2\Delta f(\Delta)) p^{\frac{1}{\Delta+1}}.
	\end{align*}
	By recursively applying \Cref{lem:singlestep}, we get the following:
	\[
	p_j \le  (2\Delta f(\Delta)) p_{j-1}^{\frac{1}{\Delta+1}},
	\]
	where $p_0=p$ and $p_{2j}$, are, respectively, the local failure probability bounds for $\Pi_{0}$ and $\Pi_{j}$. We prove by induction that for all $j>0$,
	\[
	p_j \le (2\Delta f(\Delta))^2 p^{\frac{1}{(\Delta+1)^j}}.
	\]
	For the base case where $j=1$, we get that $p_1 \le(2\Delta f(\Delta))^2 p^{\frac{1}{\Delta+1}}$, which holds, as we showed above. Let us assume that the claim holds for $j$, and let us prove it for $j+1$. We obtain the following, where the second inequality holds by the inductive hypothesis.
	\begin{align*}
	p_{j+1} &\le 2\Delta f(\Delta) p_{j}^{\frac{1}{\Delta+1}} \le 2\Delta f(\Delta) \left(  (2\Delta f(\Delta))^2 p^{\frac{1}{(\Delta+1)^j}} \right)^{\frac{1}{\Delta+1}} \\
	&\le (2\Delta f(\Delta))^{1+\frac{2}{\Delta+1}} p^{\frac{1}{(\Delta+1)^{j+1}}} \le  (2\Delta f(\Delta))^2 p^{\frac{1}{(\Delta+1)^{j+1}}}.
	\end{align*}
\end{proof}

We now handle \Cref{point:zerorounds}. We show that, the mere fact that a problem $\Pi$ cannot be solved in $0$ rounds in the deterministic port numbering model implies a lower bound on the local failure probability of any algorithm that tries to solve $\Pi$ in the randomized port numbering model. Also in this case, we consider hypergraphs, and we assume $\Delta$ to be an upper bound on the degree of the nodes and on the rank of the hyperedges.
\begin{lemma}\label{lem:zerorounds}
	Let $\Pi$ be a problem that cannot be solved in $0$ rounds with deterministic algorithms in the port numbering model given a port numbering assignment that satisfies some constraints $C^{\mathrm{port}}$, such that $|\Sigma| \le f(\Delta)$, for some function $f$. Then any randomized $0$-round algorithm solving $\Pi$ must fail with probability at least $\frac{1}{\Delta^{3\Delta^2} f(\Delta)^{\Delta^2}}$, if the port numbering is chosen uniformly at random among the ones allowed by $C^{\mathrm{port}}$.
\end{lemma}
\begin{proof}
	Since in $0$ rounds of communication all nodes have the same information, we can see any $0$-round algorithm as a probability assignment to each node configuration. That is, for each $\c_i \in \nodeconst$, the algorithm outputs $\c_i$ with probability $p_i$, such that $\sum p_i = 1$. Since the number of labels is at most $f(\Delta)$, then the number of possible configurations appearing in $\nodeconst$ is at most $f(\Delta)^\Delta$. 
	Hence, by the pigeonhole principle, there exists some configuration $\bar{\c}$ that all nodes output with probability at least $\frac{1}{f(\Delta)^\Delta}$.  
	
	Also, since there are at most $\Delta!$ different permutations of $\bar{c}$, a node can map the labels of $\bar{\c}$ to its $\Delta$ ports in at most $\Delta! \le \Delta^\Delta$ different ways. Hence, conditioned on a node to output $\bar{c}$, there is at least one permutation $\sigma$ that is chosen with probability at least $1/\Delta^\Delta$. Hence, there is at least one configuration $\bar{\c}$ that is output using some permutation $\sigma$ with probability at least $\frac{1}{\Delta^\Delta f(\Delta)^\Delta}$.
	The probability that all the nodes incident to the same hyperedge output the configuration $\bar{\c}$ using the same permutation $\sigma$ is at least $(\frac{1}{\Delta^\Delta f(\Delta)^\Delta})^\Delta = \frac{1}{\Delta^{\Delta^2} f(\Delta)^{\Delta^2}}$.
	
	Since the problem is not $0$-round solvable in the deterministic setting, we know that there is a port numbering assignment allowed by $C^{\mathrm{port}}$, such that, if all nodes around the same hyperedge output $\bar{c}$ with the same permutation $\sigma$, then the constraint on the hyperedge is not satisfied. If we consider a hyperedge $h$, and $\delta$ nodes incident to it, $v_1, \ldots v_\delta$, there are $2 \Delta \delta \le 2 \Delta^2$ ports in total: for each $v_i$ there are $2 \Delta$ ports, the ones connecting them to their incident hyperedges, and the ports connecting the hyperedges to them. Since the number of port assignments allowed by $C^{\mathrm{port}}$ is upper bounded by the total number of port assignments, that is at most $\Delta^{2\Delta^2}$, then the probability that all nodes incident to the same hyperedge output the same configuration using the same permutation such that the port numbering assignment makes this output fail on the hyperedge is at least $\frac{1}{\Delta^{2\Delta^2} \Delta^{\Delta^2} f(\Delta)^{\Delta^2}} = \frac{1}{\Delta^{3\Delta^2} f(\Delta)^{\Delta^2}}$.
\end{proof}

We now handle \Cref{point:combine}. We show that any algorithm that runs in strictly less than $t$ rounds, where $t$ is a lower bound for the deterministic port numbering model, must fail with large probability. This lemma also appeared in previous works \cite{Balliu2019, balliurules, trulytight, binary}, and here we just make it parametric on $f(\Delta)$.
\begin{lemma}\label{lem:randomizedpnlb}
	Let $\Pi_0 \rightarrow \Pi_1 \rightarrow \dots \rightarrow \Pi_t$ be a sequence of problems satisfying the conditions of \Cref{thm:lifting}. Any algorithm for $\Pi_0$ running in strictly less than $t$ rounds must fail with probability at least $(\frac{1}{f(\Delta)^{\Delta^{10t}}})$.
\end{lemma}
\begin{proof}
	By applying \Cref{lem:multiplesteps} we get that an algorithm solving $\Pi_0$ in $t' < t$ rounds with local failure probability at most $p$ implies an algorithm solving $\Pi_{t'}$ in $0$ rounds with local failure probability at most $(2\Delta f(\Delta))^2 p^{1/(\Delta+1)^{2t'}}$. Then, since  $\Pi_{t'}$ is not $0$-round solvable in the port numbering model using deterministic algorithms by assumption, by applying \Cref{lem:zerorounds} we get the following:
	\[
	(2\Delta f(\Delta))^2 p^{1/(\Delta+1)^{2t'}} \ge \frac{1}{\Delta^{3\Delta^2} f(\Delta)^{\Delta^2}},
	\]
	that implies the following:
	\begin{align*}
	p &\ge \left(\frac{1}{\Delta^{3\Delta^2} f(\Delta)^{\Delta^2} (2\Delta f(\Delta))^2 }\right)^{(\Delta+1)^{2t'}}\\ 
	&\ge \left(\frac{1}{4 \Delta^{3\Delta^2 + 2} f(\Delta)^{\Delta^2+2}}\right)^{(\Delta+1)^{2t'}}\\ 
	&\ge \left(\frac{1}{f(\Delta)^{\Delta^2 + 2 + (3\Delta^2+4)\log\Delta}}\right)^{(\Delta+1)^{2t'}}\\ 
	&\ge \left(\frac{1}{f(\Delta)^{6\Delta^3}}\right)^{(\Delta+1)^{2t'}}\\ 
	&\ge \left(\frac{1}{f(\Delta)^{6\Delta^3}}\right)^{\Delta^{4t'}}\\ 
	&\ge \frac{1}{f(\Delta)^{6\Delta^3 \Delta^{4t'}}}\\ 
	&\ge \frac{1}{f(\Delta)^{\Delta^{10t'}}}\\
	&\ge \frac{1}{f(\Delta)^{\Delta^{10t}}}.
	\end{align*}
\end{proof}

We now handle \Cref{point:withn}, and show that if an algorithm runs ``too fast'' then it must fail with probability strictly larger than $1/n$. This lemma essentially says that, if an algorithm, on infinite $\Delta$-regular trees, runs in $o(\min\{t,\log_\Delta \log n\})$, for some given $n$, where $t$ is a lower bound for the deterministic port numbering model, then it cannot locally succeed with high probability.
\begin{lemma}\label{lem:randomfail}
	Let $\Pi_0 \rightarrow \Pi_1 \rightarrow \dots \rightarrow \Pi_t$ be a sequence of problems satisfying the conditions of \Cref{thm:lifting}. Any randomized algorithm running in strictly less than $\min\{t, \frac{1}{10} (\log_\Delta \log n -  \log_\Delta \log f(\Delta) )\}$ rounds must fail with probability $> 1/n$.
\end{lemma}
\begin{proof}
	By \Cref{lem:randomizedpnlb}, any algorithm running in strictly less than $t$ rounds must fail with probability $p \ge \frac{1}{f(\Delta)^{\Delta^{10t}}}$. We show that, for all $t < \frac{1}{10} (\log_\Delta \log n -  \log_\Delta \log f(\Delta) )$, it holds that $p > 1/n$.
	
	\begin{align*}
	& t < \frac{1}{10} (\log_\Delta \log n -  \log_\Delta \log f(\Delta) ) & \Rightarrow \\
	& 10t\log\Delta <  \log \log n - \log \log f(\Delta) & \Rightarrow \\
	& 10t\log\Delta + \log \log f(\Delta)  < \log \log n & \Rightarrow \\
	& \Delta^{10t} \log f(\Delta) < \log n & \Rightarrow \\
	& f(\Delta)^{\Delta^{10t}} < n  & \Rightarrow \\
	& \frac{1}{f(\Delta)^{\Delta^{10t}}} > 1/n &
	\end{align*}
\end{proof}

We now handle \Cref{point:randlocal}. We show that if an algorithm fails with large probability in infinite trees, then it also fails with large probability in finite trees. Again, this lemma also appeared in previous works \cite{Balliu2019, balliurules, trulytight, binary}, and here we just make it parametric on $f(\Delta)$.
\begin{lemma}\label{lem:randlocal}
	Let $\Pi_0 \rightarrow \Pi_1 \rightarrow \dots \rightarrow \Pi_t$ be a sequence of problems satisfying the conditions of \Cref{thm:lifting}. Any randomized algorithm running in the \LOCAL model, in $\Delta$-regular balanced trees of $n$ nodes, that fails with probability at most $1/n$, requires $\Omega(\min\{t, \log_\Delta \log n -  \log_\Delta \log f(\Delta)\})$ rounds.
\end{lemma}
\begin{proof}
	We prove the result by contradiction. Assume that there is an algorithm for the \LOCAL model that contradicts our claim. Note that this algorithm can be simulated in the randomized port numbering model with essentially the same failure probability, since an algorithm in the randomized port numbering model can first generate unique IDs with high probability of success, and then simulate the \LOCAL algorithm. This algorithm is for finite $\Delta$-regular balanced trees of $n$ nodes. We show that we can use such an algorithm in infinite  $\Delta$-regular trees and obtain the same local failure probability, contradicting \Cref{lem:randomfail}.
	
	The claim follows by a standard argument: we can run this algorithm on infinite trees, and if it fails with too large local failure probability, then we can consider the finite tree of $n$ nodes obtained by considering the subgraph induced by the $O(\log_\Delta n)$ neighborhood of the failing node. By running the algorithm on the obtained tree of $n$ nodes, the view of the failing node itself, and its neighbors, is the same as in the original instance, since the running time is $o(\log_\Delta n)$. Hence, by a standard indistinguishability argument, the algorithm fails with probability strictly larger than $1/n$, contradicting the correctness of the given algorithm.
\end{proof}

We can finally handle \Cref{point:detlocal}, and show that a strong lower bound for the randomized \LOCAL model implies an even stronger lower bound for the deterministic \LOCAL model. The proof uses a standard ``lie about $n$'' argument, and again, it has been used in the context of round elimination in~\cite{Balliu2019, balliurules, trulytight, binary}, and we report the same proof here modified to handle an arbitrary number of labels~$f(\Delta)$.
\begin{lemma}\label{lem:detlocal}
	Let $\Pi_0 \rightarrow \Pi_1 \rightarrow \dots \rightarrow \Pi_t$ be a sequence of problems satisfying the conditions of \Cref{thm:lifting}. Any deterministic algorithm running in the \LOCAL model, in $\Delta$-regular balanced trees of $n$ nodes, requires $\Omega(\min\{t, \log_\Delta n -  \log_\Delta \log f(\Delta)\})$ rounds.
\end{lemma}
\begin{proof}
	The proof follows by a standard ``lie about $n$'' argument. If we run an algorithm $A$ and lie about the size of the graph, if the algorithm cannot notice the lie, then it must work correctly, with a running time that now depends on the value of $n$ provided to the algorithm.
	In particular, in \cite[Theorem 6]{chang16exponential}, it is shown that, given an algorithm that runs in $T(n)$ rounds, if we lie about the size of the graph and say that it has size $N$, then the obtained algorithm is correct and runs in $T = T(N)$ rounds if the following conditions are satisfied:
	\begin{itemize}
		\item The IDs in each $(T+1)$-radius neighborhood are unique and from $\{1,2,\dotsc, N\}$;
		\item Each $(T+1)$-radius neighborhood contains at most $N$ nodes.
	\end{itemize}
	
	Assume for a contradiction that we have a deterministic algorithm $A$ that contradicts the claim. In particular, this means that it terminates in $o(\log_\Delta n)$ rounds. Let us fix $N = \log n$, and let us construct a new algorithm $A'$ that runs $A$ by lying about the size of the graph, and claims that it has size $N$. In order to satisfy the first requirement, we need to create a new ID assignment from a smaller space. For this purpose, we can compute an $N$-coloring of $G^{2T+2}$, the $(2T+2)$th power of $G$, using three iterations of Linial's coloring algorithm~\cite[Corollary 4.1]{Linial1992}, and use the obtained coloring as a new ID assignment. This algorithm can be used to color a $k$-colored graph of maximum degree $\bar \Delta$ with e.g.
	\[
	O\bigl(\bar \Delta^2 \bigl(\log \log \log k + \log \bar \Delta\bigr)\bigr)
	\]
	colors in three rounds. The power graph $G^{2T+2}$ has maximum degree $\bar \Delta \leq \Delta^{2T+2}$ and thus we get a coloring of $G^{2T+2}$ with
	\[
	O\bigl(\Delta^{2(2T+2)} \bigl(\log \log N + \log \Delta^{2T+2}\bigr)\bigr)
	\]
	colors. Since $T(N) \in o(\log_\Delta N)$, then each $(T+1)$-radius neighborhood contains $O(\Delta^{T+1}) < N$ nodes (and hence the second requirement is satisfied) and that the number of colors is less than $N$ (and hence the first requirement is satisfied). The coloring can be computed in time $O(T)$ in the original graph $G$. Hence, if $A$ is correct, then $A'$ is also correct and terminates in $o(\min\{t, \log_\Delta N -  \log_\Delta \log f(\Delta)\}) = o(\min\{t, \log_\Delta \log n -  \log_\Delta \log f(\Delta)\})$ rounds, but this contradicts the claim of \Cref{lem:randlocal}.
\end{proof}
By combining \Cref{lem:randlocal} and \Cref{lem:detlocal}, we obtain \Cref{thm:lifting}.

\end{document}

%% file: introduction.tex

\section{Introduction}
\label{sec:intro}

In the area of distributed graph algorithms, we usually assume that the graph $G=(V,E)$ on which we want to solve some graph problem also represents a communication network. The nodes of $G$ are autonomous agents that communicate over the edges of $G$. Each node $v\in V$ is equipped with a unique $O(\log n)$-bit identifier (where $n=|V|$). More precisely, the nodes typically interact in synchronous rounds and in each round, every node $v\in V$ can send a message to each of its neighbors in $G$. In the standard \LOCAL model~\cite{Linial1987,Peleg2000}, those messages can be of arbitrary size. Except for information that we assume to be publicly known (such as, e.g., the family of graphs from which $G$ is chosen), at the beginning, nodes do not know anything about the graph $G$. At the end of an algorithm, every node $v\in V$ has to output its local part of the solution of the considered graph problem. The time or round complexity of such a distributed algorithm is defined as the number of rounds that are necessary until all nodes terminate and output their part of the solution.

Since the early works on distributed and parallel graph algorithms~\cite{Alon1986,Linial1987,Luby1986}, two problems that have been at the very core of the attention are the problems of computing a vertex coloring of $G$ and of computing a maximal independent set (MIS) of $G$. Both problems, as well as their many variants, have been studied intensively throughout the last more than 30 years.
A well-studied variant of MIS is the problem of computing an $(\alpha,\beta)$-ruling set, introduced in the late 1980s~\cite{Awerbuch89}. The problem asks for a subset $S$ of the nodes satisfying that nodes in $S$ are at distance at least $\alpha$ from each other, and nodes not in $S$ are at distance at most $\beta$ from a node in $S$. Note that $(2,1)$-ruling sets are equivalent to MIS, and by increasing $\beta$ we obtain relaxed variants of MIS. Interestingly, it is often the case that whenever the computation of an MIS can be used as a subroutine to solve some problem of interest, it is actually enough to compute a $(2,\beta)$-ruling set for some $\beta>1$, which is easier to compute. In fact, the fastest known distributed algorithms for several fundamental distributed problems use the computation of ruling sets as subroutines and the complexity of those algorithms critically depends on having fast $(2,\beta)$-ruling set algorithms for some values $\beta>1$. Examples are deterministic and randomized algorithms for computing $\Delta$-colorings\footnote{As common when considering the distributed $\Delta$-coloring problem, we will always implicitly assume that $\Delta \geq 3$ and that the input graph does not contain a connected component that is $\Delta$-regular (which in particular implies, by Brooks' Theorem~\cite{brooks1941colouring}, that the input graph is $\Delta$-colorable).} \cite{GhaffariHKM18,panconesi95delta}, randomized distributed algorithms for the Lov{\'{a}}sz Local Lemma~\cite{FischerG17,FOCS18-derand}, as well as randomized distributed algorithms for computing an MIS~\cite{ghaffari16improved,Barenboim2016}. In some cases, improved algorithms for ruling sets would directly lead to improved algorithms for such problems. Hence, understanding if it is possible to obtain faster algorithms for ruling sets is of great importance. The first lower bounds on computing ruling sets (except MIS) were proven in \cite{balliurules} and as one of the main contributions of our work, we improve those lower bounds exponentially and we in fact show that some existing ruling set algorithms are optimal. This for example implies that in order to improve the current best algorithm for $\Delta$-coloring~\cite{GhaffariHKM18}, we need to find a genuinely different algorithm, which is not based on ruling sets.

\paragraph{Round Elimination.} As many recent lower bounds for local distributed graph problems, the lower bounds of \cite{balliurules,BBKOmis} are based on the \emph{round elimination} technique~\cite{Brandt2016,Brandt2019}. The high-level idea of the round elimination technique is as follows. Assume that the time needed to solve some desired problem, such as computing an MIS, is $T$ rounds. Then, every algorithm that solves the problem in $T$ rounds must already have computed something non-trivial after $T-1$ rounds. And it turns out that one can exactly characterize the weakest problem that needs to be solved after $T-1$ rounds if the desired problem has to be computed after $T$ rounds. Starting from our target problem, we can then construct a sequence of problems where each problem in the sequence is exactly one round easier than the previous one. And if we can construct a sequence of length $R$ such that the last problem of the sequence still cannot be solved in $0$ rounds, then we know that our target problem cannot be solved in $R$ rounds.

In the last few years, round elimination has been used to prove an impressive number of strong new distributed lower bounds~\cite{Brandt2016,chang16exponential,chang18complexity,Brandt2019,Balliu2019,BalliuHOS19,trulytight,binary,balliurules,BBKOmis}. A highlight of this line of work certainly are the lower bounds proven in \cite{Balliu2019}, where it is shown that the number of rounds required to compute a maximal matching is at least $\Omega(\min\set{\Delta, \log n/\log\log n})$ deterministically and $\Omega(\min\set{\Delta, \log\log n/\log\log\log n})$ in the randomized \LOCAL model, where $\Delta$ is the largest degree of $G$. Because a maximal matching of a graph $G$ is an MIS of the line graph of $G$, the same bounds also hold for computing an MIS. Since an MIS can be computed deterministically in $O(\Delta+\log^* n)$ rounds~\cite{barenboim14distributed}, both bounds are tight as a function of $\Delta$. In this paper we prove new tight lower bounds for a large family of distributed graph problems. As one special case of our general lower bound result, we obtain that the same lower bounds that were proven in \cite{Balliu2019} for MIS in general graphs also hold for computing an MIS in regular trees, which is an exponential improvement over the existing lower bounds~\cite{balliurules,BBKOmis}.

\paragraph{New Technique.} 
There are some problems that behave in a special way  in the round elimination framework: if we compute the problem that is exactly one round easier, we obtain the problem itself. We call these problems \emph{fixed points}. Such a behavior is not contradictory and such problems actually exist; in fact, proving that a problem is a fixed point is a powerful tool: we know that being a fixed point implies that the problem requires $\Omega(\log_\Delta n)$ rounds for deterministic algorithms, and $\Omega(\log_\Delta \log n)$ rounds for randomized ones, in the \LOCAL model. More generally, fixed points provide a conceptually simple method to prove lower bounds: if we can show that a relaxed version of a given problem\footnote{A problem $\Pi'$ is a relaxation of $\Pi$ if any solution for $\Pi$ is also a solution for $\Pi'$.} is a fixed point, then we directly obtain a lower bound for the original problem. Understanding whether all problems that have a deterministic complexity of $\Omega(\log_\Delta n)$ rounds admit a relaxation to some fixed point is a major open question. Clearly, fixed point relaxations cannot exist for problems that can be solved deterministically in $o(\log_\Delta n)$ rounds, and in particular they cannot exist for problems that, for some function $f$, can be solved in $O(f(\Delta) + \log^* n)$ rounds.

Despite this apparent incompatibility, we show in this work that, perhaps surprisingly, fixed points can also be used for proving lower bounds for problems that can be solved in $ O(f(\Delta) + \log^* n)$ rounds. Problems of this kind are difficult to understand in the context of round elimination: typically, when we apply a round elimination step to such a problem, we obtain a problem with a description that is exponentially larger than the description of the original problem. Hence, if we apply multiple round elimination steps, we obtain a problem sequence where the description size grows like a power tower, and this makes it extremely hard to characterize this sequence, or even merely determine which problems in the sequence are $0$-round solvable. We show that we can sometimes avoid this growth by relaxing the problem of interest by \emph{embedding a fixed point} in it, and this allows us to prove lower bounds for a large family of problems whose behavior was too complex to deal with previously.
In the following we provide a high-level explanation of this concept where, for simplicity, we restrict attention to one specific problem of this family---the MIS problem.

By applying round elimination to the MIS problem for $k-1$ times, we obtain a problem $\Pi_k$ that can be described by two main parts, $\Pi^{\mathrm{good}}_k$ and $\Pi^{\mathrm{bad}}_k$. Part $\Pi^{\mathrm{good}}_k$ allows two types of nodes, colored ones and uncolored ones. The requirements are that uncolored nodes must have a colored neighbor, colored nodes have a color from a palette of size $k$, and the graph induced by colored nodes must be colored properly (note that the MIS problem is equivalent to $\Pi_1 = \Pi^{\mathrm{good}}_1$). Part $\Pi^{\mathrm{bad}}_k$ allows some other types of outputs that cannot be easily understood. The number of ``types of outputs'' of this form is roughly a power tower of height $k$.

The above directly implies an $O(\Delta + \log^* n)$-round upper bound for MIS: a $(\Delta+1)$-coloring of the nodes (which can be found in $o(\Delta) + O(\log^* n)$ rounds~\cite{fraigniaud16local}) is a solution to $\Pi^{\mathrm{good}}_{\Delta + 1}$ and therefore also to the more relaxed $\Pi_{\Delta + 1}$, which in turn has a complexity of precisely $\Delta$ rounds less than $\Pi_1 = \text{MIS}$. In order to prove that this complexity is tight by providing a matching lower bound, we need to show that none of the outputs allowed by $\Pi^{\mathrm{bad}}_k$ actually contribute to making the problem easier, but this is extremely challenging, since the number of outputs allowed by $\Pi^{\mathrm{bad}}_k$ is a power tower of height $k$.

In this work, we show how to prevent $\Pi^{\mathrm{bad}}_k$ from appearing altogether in the sequence of problems obtained by round elimination. Let $\Pi_k^{\mathrm{rel}}$ be the problem obtained by replacing the $k$-coloring part of $\Pi_k^{\mathrm{good}}$ with a fixed point relaxation of the $k$-coloring problem (for which we prove the existence). Note that $\Pi_k^{\mathrm{rel}}$ is at least as easy as $\Pi_k^{\mathrm{good}}$. We show that, by applying the round elimination technique to $\Pi_k^{\mathrm{rel}}$, we obtain a problem that is at least as hard as $\Pi_{k+1}^{\mathrm{rel}}$. This allows us to obtain a sequence of problems of length $\Omega(\Delta)$, where each problem is at least one round easier than the previous one, the last problem (and all previous ones) is not $0$-round solvable, and $\Pi^{\mathrm{bad}}_k$ does not appear at all. 

In other words, the reason why we obtain $\Pi^{\mathrm{bad}}_k$ is that the $\Delta$-coloring problem is hidden in the sequence of problems obtained by starting from the MIS problem, and 
we find it fascinating that the hardness of a hard problem (the relaxed $\Delta$-coloring problem) that requires $\Omega(\log n)$ rounds on bounded-degree graphs can be used to prove lower bounds for problems that can be solved in time $O(f(\Delta)+\log^* n)$ for some function $f$, i.e., with a much better $n$-dependency.
More generally, we believe that our new technique of embedding the $\Omega(\log n)$-round hardness in the form of a suitable fixed point into a problem of interest provides a promising (perhaps even generic) way to tackle the main challenge of round elimination---finding a method to make the obtained problem sequences understandable.
After all, on an intuitive level, the particular property of fixed points---being invariant, and in particular non-growing, under round elimination---perfectly fits the objective of keeping the problem sequence in check by restricting the growth of the problem descriptions, and it seems plausible to us that this is precisely why we obtain several tight and close-to-tight lower bounds.

\paragraph{New \boldmath$\Delta$-Coloring Fixed Point.} At the core of our argument is a new relaxation of $k$-coloring, which we prove to be a fixed point in the round elimination framework in $\Delta$-regular trees as long as $k\leq \Delta$. That is, if we perform a round elimination step for this relaxed $k$-coloring problem, then we arrive at the same relaxed $k$-coloring problem. In the relaxed $k$-coloring problem, every node $v$ is allowed to be colored with a subset $A_v$ of the $k$ colors. Nodes that choose more than one color can mark some of their edges as \emph{``don't care''} edges. Specifically, a node $v$ that chooses $|A_v|=a_v$ different colors is allowed to mark $a_v-1$ of its edges as \emph{``don't care''} edges. For every edge $\set{u,v}$ of the graph, if $u$ or $v$ mark $\set{u,v}$ as \emph{``don't care''} edge, then there are no constraints for this edge. Otherwise, the color sets $A_u$ and $A_v$ of $u$ and $v$ must be disjoint. Since this problem is a round elimination fixed point for $k=\Delta$, we directly get an alternative proof for the result of \cite{Brandt2016,chang16exponential} that $\Delta$-coloring $\Delta$-regular trees requires $\Omega(\log_\Delta n)$ rounds for deterministic and $\Omega(\log_\Delta\log n)$ rounds for randomized distributed algorithms. Further, apart from proper $\Delta$-coloring, the described relaxed $\Delta$-coloring problem is also a relaxation of many arbdefective coloring problems.

\paragraph{Arbdefective Coloring.} Arbdefective colorings were introduced in \cite{BarenboimE11} and they have since then become an indispensable tool in the development of efficient distributed coloring algorithms, e.g., \cite{BarenboimE11,barenboim16sublinear,fraigniaud16local,Kuhn20,MausTonoyan20,GhaffariKuhn20}. An arbdefective $c$-coloring of a graph $G=(V,E)$ is a coloring of the nodes $V$ with $c$ colors together with an orientation of the edges $E$. The arbdefect of a node $v$ of color $x$ is the number of outneighbors of $v$ of color $x$. Consider the set of nodes that are colored with some set $A$ of colors in the above relaxed $k$-coloring problem. The specification for the nodes with color set $A$ corresponds exactly to the specification of a color class with arbdefect at most $|A|-1$. If we want to compute a $c$-coloring of a $\Delta$-regular tree such that for each of the colors $x\in\set{1,\dots,c}$, we have a bound $d_x\geq 0$ on the arbdefect of the nodes with color $x$, then our lower bound on the relaxed $\Delta$ coloring problem directly shows that as long as $\sum_{x=1}^c(d_x+1)\leq \Delta$, this arbdefective coloring problem requires $\Omega(\log_\Delta n)$ rounds deterministically and $\Omega(\log_\Delta\log n)$ rounds with randomization. We also show that as soon as $\sum_{x=1}^c(d_x+1)\geq \Delta+1$, the problem can be solved in $O(\Delta+\log^* n)$ rounds by a simple distributed greedy algorithm. In other words, we \emph{exactly} characterize when the arbdefective coloring problem is ``easy'' and when it is ``hard''.

\paragraph{A General Problem Family.} In this work, we show that our new technique can be used to prove lower bounds for a large family of problems that contains many natural relaxations of MIS and coloring problems. For these problems we prove tight lower bounds as a function of $\Delta$ (modulo assuming an appropriate initial coloring of the nodes) and improved lower bounds as a function of $n$. 
In this way, we obtain a unified lower bound proof for many problems for which (possibly) worse lower bounds were already known \cite{Brandt2016, chang16exponential, balliurules, BBKOmis}, as well as lower bounds for new problems.
For a relation between (a simplified version of) our problem family and various well-known problems, please see \Cref{fig:problemFamily}.

We next describe a simplified variant of our problem family. A full and formal definition of the problem family appears in \Cref{sec:family}. In the simplified variant, a problem has three integer parameters $\alpha \ge 0$, $\beta\geq 0$, and $c\geq 1$, and we call these problems $\alpha$-arbdefective $c$-colored $\beta$-ruling sets. For a graph $G=(V,E)$, the objective is to compute a set $S\subseteq V$ of nodes together with an $\alpha$-arbdefective $c$-coloring of the induced subgraph $G[S]$. Moreover, all nodes $v\in V\setminus S$ must be within distance at most $\beta$ from a node in $S$ (see \Cref{fig:2-2-2} for an example).

For $\beta =0$, we obtain the $\alpha$-arbdefective $c$-coloring problem. Starting from arbdefective coloring, if we set $\alpha=0$ and $c=\Delta$, we obtain the $\Delta$-coloring problem. If we instead set $\alpha = \Delta-1$ and $c=1$, we obtain the sinkless orientation problem.
For $\alpha = 0$ and $c=1$, we obtain the $(2,\beta)$-ruling set problem, and by setting $\beta=1$ we obtain the MIS problem.
By setting $\beta=1$ and $c=1$, we obtain the $\alpha$-outdegree dominating set problem, previously defined and discussed in \cite{BBKOmis}.

\begin{figure}[t]
	\centering
	\includegraphics[width=0.7\textwidth]{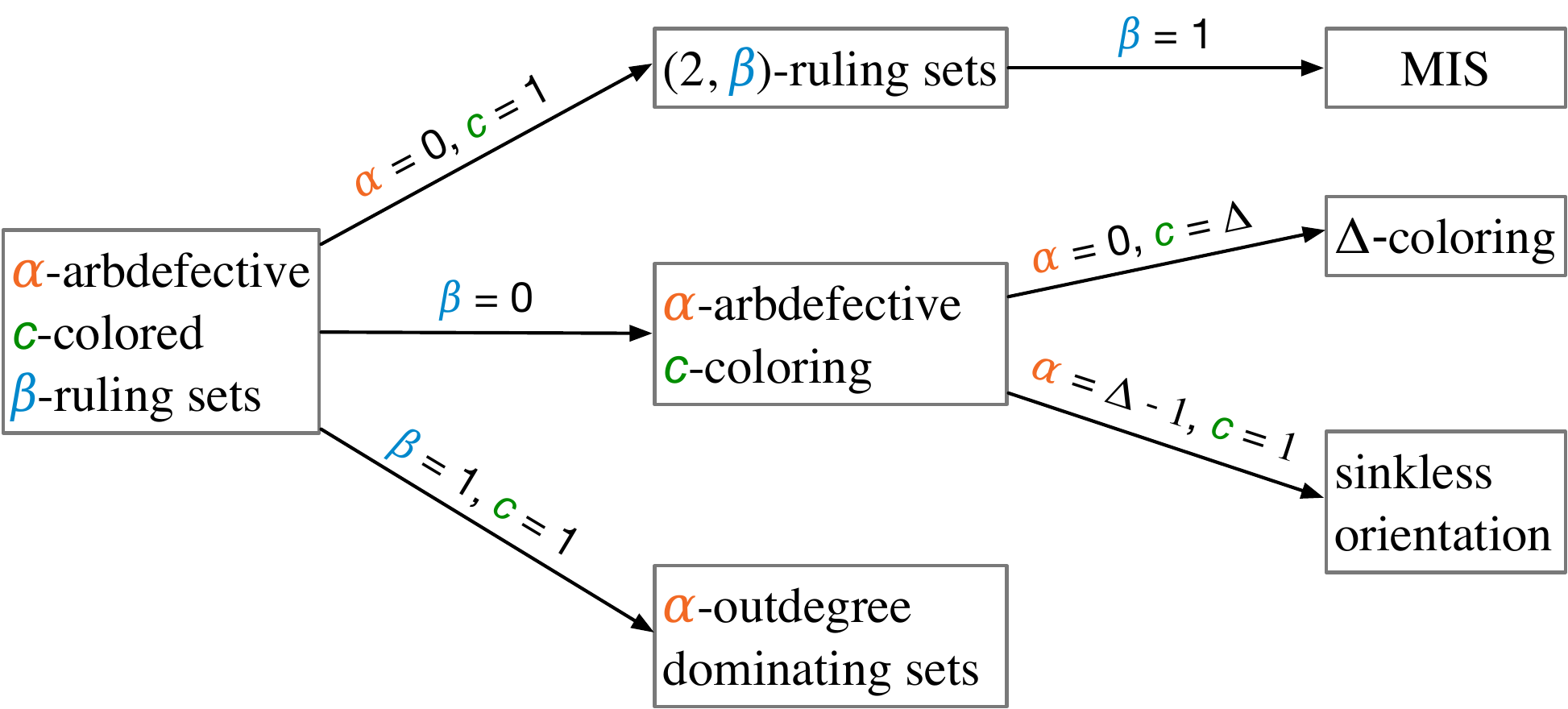}
	\caption{Relation between (a simplified version of) our problem family and concrete problems.}
	\label{fig:problemFamily}
\end{figure}

\paragraph{Main results.}
While our results apply to a large family of problems, we summarize our main results in the following.
For precise formal statements, we refer to Sections~\ref{sec:ourresults}, \ref{sec:corollaries}, and \ref{sec:algorithms}.

Our first contribution is related to the problem of computing ruling sets, which are used as a subroutine in many fundamental problems. For these problems, we obtain tight lower bounds as a function of $\Delta$, improving exponentially on the results presented in \cite{balliurules}. Also, we obtain improved lower bounds as a function of $n$. Our lower bounds hold not only on general graphs, but even on trees, and it turns out that our lower bounds as a function of $n$ are tight on trees for deterministic algorithms. This in particular implies that, for the MIS problem on trees, we now know its exact deterministic complexity as a function of $n$.

Our second contribution regards the arbdefective coloring problem, which is a powerful subroutine in many efficient coloring algorithms. We exactly characterize which variants are easy, in the sense that they can be solved in $O(f(\Delta) + \log^* n)$ rounds, and which of them are hard, in the sense that they require $\Omega(\log_\Delta n)$ rounds.
Our third contribution consists in introducing a promising new technique, i.e., using fixed points for proving lower bounds for problems that are not in $\Omega(\log_\Delta n)$.

Finally, our approach provides a unified proof for various results---both new and already present in the literature---by showing lower bounds for a problem family that interpolates between coloring problems, orientation problems, and independent set problems.
The tightness of the obtained results shows that considering these types of problems as special cases of the presented broader problem family is a very useful approach for understanding their complexities, and we hope that similar generalizations will lead to tight bounds for other problems.

\subsection{Problem Definitions}
\label{sec:problems}

We next formally define some natural problems for which we prove upper and lower bounds. As discussed above, the full general (and somewhat technical) family of problems for which we prove our lower bounds is formally defined in \Cref{sec:family}. In the following, for a positive integer $k$, we use $[k]$ to denote the set $[k]:=\set{1,\dots,k}$.





\paragraph{Generalized Arbdefective Colorings.}
We define a generalized notion of arbdefective colorings. We first
define the notion of an \emph{arbdefect vector}, which specifies the
requirements that a node has to satisfy if it is given a certain
color. If there are $C$ colors available (for some integer $C\geq 1$),
an arbdefect vector is a vector of length $C$, where the
$i^{\mathit{th}}$ coordinate indicates the allowed arbdefect of a node
having color $i$. Concretely, an arbdefect vector $\vec{d}$ of length
$C$ has the following form $\vec{d}=(d_1,\dots,d_C)\in\ARB^C$. A
\emph{generalized arbdefective coloring} of a graph is now defined as
follows.

\begin{definition}[Generalized Arbdefective Coloring]\label{def:generalarbdefect}
  Let $G=(V,E)$ be a graph, $C\geq 1$ an integer, and
  $\vec{d}=(d_1,\dots,d_C)\in \ARB^C$ an arbdefect vector of
  length $C$. An assignment $\varphi:V\to[C]$ of colors in $[C]$ to
  the nodes $V$ is called a $\vec{d}$-arbdefective $C$-coloring
  w.r.t.\ a given orientation of the edges $E$ of $G$ if for every
  node $v\in V$, if $v$ has color $\varphi(v)=x$, then $v$ has at most
  $d_x$ outneighbors of color $x$.
\end{definition}

If we have an arbdefect vector $\vec{d}=(d_1,\dots,d_C)$ with $d_i=\alpha$ for all $i\in [C]$, we call the corresponding coloring an \emph{$\alpha$-arbdefective $C$-coloring}. Note that this is the standard notion of arbdefective coloring as defined in \cite{BarenboimE11}.  When writing that an algorithm computes an
arbdefective coloring of a graph $G$, we always implicitly mean that
the algorithm computes a coloring of the nodes of $G$ \emph{and} a
corresponding orientation of the edges of $G$. We define the \emph{capacity} $\capa(\vec{d})$
of an arbdefect vector $\vec{d}$ as $\capa(\vec{d}):=\sum_{i=1}^C(d_i+1)$.

As discussed above, computing a $\vec{d}$-arbdefective coloring in a
$\Delta$-regular tree is as hard as computing a $\Delta$-coloring as
long as $\capa(\vec{d})\leq \Delta$ and it can be computed by a simple
greedy algorithm if $\capa(\vec{d})>\Delta$. For a value $\sigma>0$
and graphs of maximum degree $\Delta$, we say that an arbdefect vector
$\vec{d}$ is \emph{$\sigma$-relaxed} iff
$\capa(\vec{d})>\sigma\cdot\Delta$. Hence, a given generalized
arbdefective coloring instance can be solved by a greedy algorithm iff
the given arbdefect vector is $1$-relaxed.

\begin{figure}[t]
	\centering
	\includegraphics[width=0.6\textwidth]{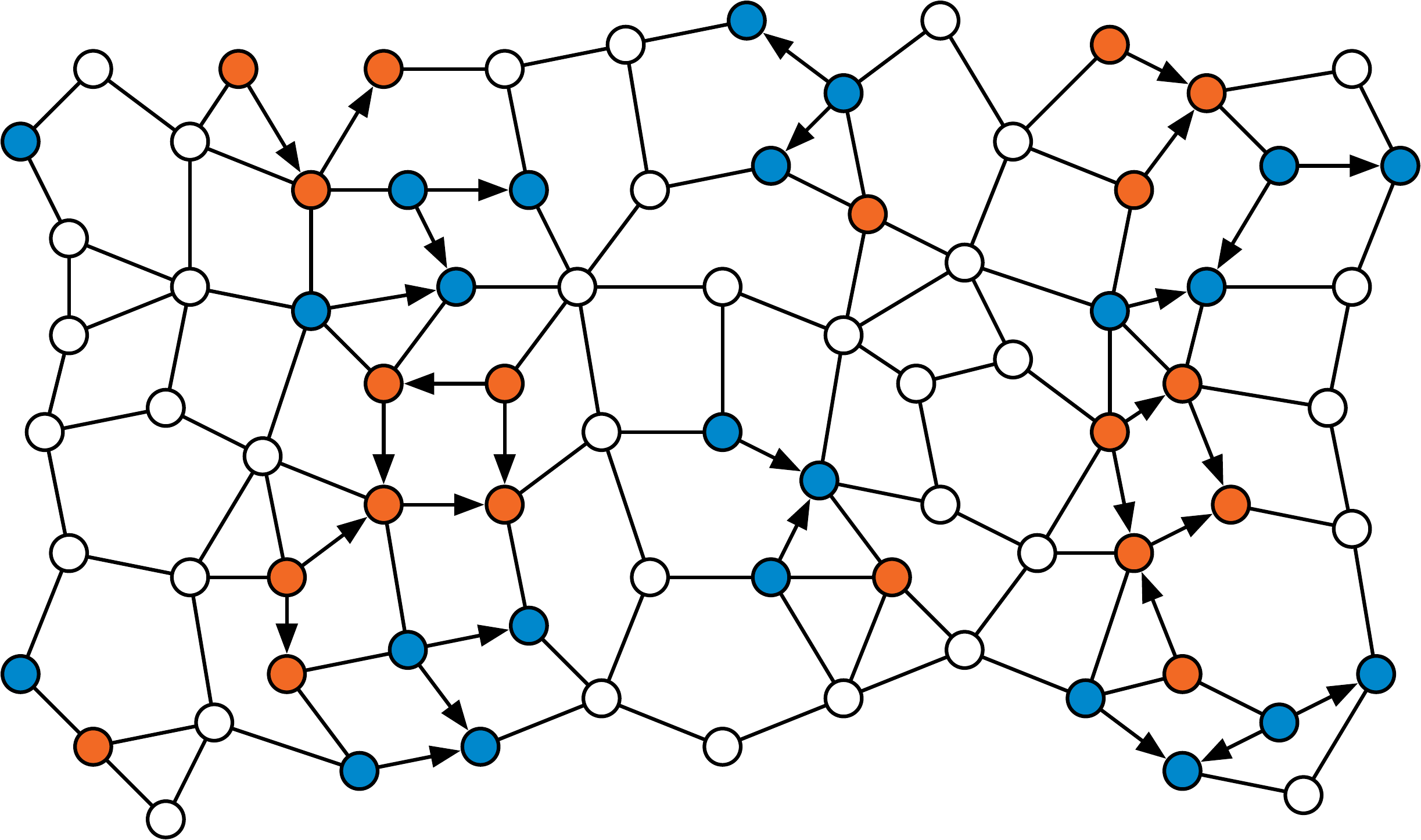}
	\caption{An example of $2$-arbdefective $2$-colored $2$-ruling set. Nodes in blue and orange are in the set, and they all have arbdefect at most $2$. White nodes are not in the set, and there is at least one node in their radius-$2$ neighborhood that is in the set.}
	\label{fig:2-2-2}
      \end{figure}
      
\paragraph{Arbdefective Colored Ruling Sets.}
We next define a family of problems that can be seen both as variants
of arbdefective colorings and of ruling sets, and they provide a way
to interpolate between them. \Cref{fig:2-2-2} provides an illustration of the following definition.

\begin{definition}[$\alpha$-Arbdefective $c$-Colored $\beta$-Ruling Set]
	Let $G = (V,E)$ be a graph, and $\alpha \ge 0$, $c \ge 1$, $\beta \ge 0$ be integers. A set $S \subseteq V$ together with an orientation of the edges between nodes in $S$ is an $\alpha$-arbdefective $c$-colored $\beta$-ruling set if it satisfies the following:
	\begin{itemize}
		\item The induced subgraph $G[S]$ is colored with an $\alpha$-arbdefective $c$-coloring.
		\item For all $v\in V\setminus S$, there is a node $u\in S$ at distance $\leq \beta$. Note that if $\beta=0$, then
                   $S=V$.
	\end{itemize}
\end{definition}



%% file: algorithms.tex
\section{Upper Bounds}
\label{sec:algorithms}

We first discuss the problem of computing generalized arbdefective
colorings as defined in \Cref{def:generalarbdefect}. We start with a
simple lemma, which states that if we are given a proper vertex $C$
coloring of a graph, as long as a given arbdefective coloring instance
is $1$-relaxed, we can compute a valid solution by iterating over
the $C$ color classes.

\begin{lemma}\label{lemma:simpletightarbcoloring}
  Let $G=(V,E)$ be a properly $m$-colored graph. Any $1$-relaxed generalized
  arbdefective coloring problem on $G$ can be solved in $O(m)$ rounds
  deterministically.
\end{lemma}
\begin{proof}
  Let $C$ be the number of colors of the given generalized
  arbdefective coloring instance and let $\vec{d}=(d_1,\dots,d_C)$ be
  the arbdefect vector.  The coloring of the nodes $V$ and the
  orientation of the edges $E$ is computed greedily in $m$ phases by
  iterating over the $m$ colors of the given proper $m$-coloring of
  $G$. Assume that the proper initial coloring assigns color
  $\gamma(v)\in [m]$ to some node $v\in V$. The edge orientation is
  defined as follows. For every edge $\set{u,v}\in E$, the edge is
  oriented from $u$ to $v$ if $\gamma(u)>\gamma(v)$ and it is oriented
  from $v$ to $u$ otherwise.  In phase $p\in\set{1,\dots,m}$, the
  nodes $v$ of initial color $\gamma(v)=p$ are colored. Note that the
  given edge orientation implies that when coloring a node $v$, all
  its outneighbors have already been colored in earlier phases and
  because the initial $m$-coloring is proper, we also never color two
  adjacent nodes in the same phase. Each node $v\in V$ that picks a
  color can now just pick a color $x\in [C]$ for which the number of
  outneighbors of color $x$ is at most $d_x$. Note that such a color
  must exist because $v$ has at most $\Delta$ outneighbors and because
  the arbdefect vector $\vec{d}$ is $1$-relaxed and therefore
  $\sum_{i=1}^C(d_i+1)>\Delta$.
\end{proof}

\Cref{thm:ubarbdef} now follows directly from the above lemma.

\begin{theorem}\label{thm:ubarbdef}
	Let $G=(V,E)$ be an $n$-node graph. Any $1$-relaxed
	generalized arbdefective coloring problem on $G$ can be solved in
	$O(\Delta+\log^* n)$ rounds deterministically.
\end{theorem}
\begin{proof}
  We can first compute a proper $(\Delta+1)$-coloring of $G$ in time
  $O(\Delta+\log^* n)$ (e.g., by using the algorithm of
  \cite{barenboim14distributed}) and then apply
  \Cref{lemma:simpletightarbcoloring}.
\end{proof}

The following lemma provides a simple generalization the $O(\Delta/d + \log^* n)$-round  $d$-arbdefective
$O(\Delta/d)$-coloring algorithm of \cite{BEG18}.

\begin{lemma}\label{lemma:relaxedarbcoloring}
  For every $\alpha\geq 0$ and every $\eta\geq 1$, an $\alpha$-arbdefective
  $O(\eta\Delta/\alpha)$-coloring of an $n$-node graph $G$ with
  maximum degree $\Delta$ can be computed deterministically in
  $O\big(\frac{\Delta}{\eta\alpha}+\log^* n\big)$ rounds.
\end{lemma}
\begin{proof}
  We first compute $\Delta/\eta$-defective $O(\eta^2)$-coloring of $G$
  in time $O(\log^* n)$ by using the distributed defective coloring
  algorithm of \cite{Kuhn2009}. Each of the $O(\eta^2)$ color classes
  now has a maximum degree of $\leq \Delta/\eta$ and we can therefore
  compute an $\alpha$-arbdefective $O(\Delta/(\eta\alpha))$-coloring
  in time $O(\Delta/(\eta\alpha)+\log^* n)$. In combination with the
  $O(\eta^2)$ colors of the initial $\Delta/\eta$-defective coloring,
  we obtain the desired $\alpha$-arbdefective coloring of $G$ with
  $O(\eta\Delta/\alpha)$ colors.
\end{proof}

\subsection{Algorithms for Arbdefective Colored Ruling Sets}

We first show that if we are already given an appropriate arbdefective coloring, an $\alpha$-arbdefective $c$-colored $\beta$-ruling set can be computed in time matching the lower bound given in \Cref{thm:lbarbcolrs}.

\ubcoldom*
\begin{proof}
  We partition the $C$ colors of the given $\alpha$-arbdefective $C$-coloring into $K:=\lceil C/c\rceil$ groups of size at most $c$. Within each of the $K$ groups, we rename the at most $c$ colors from $1$ to $c$. For every $i\in [K]$ let $V_i$ be the set of nodes in color group $i$. Note that the given coloring proviced an $\alpha$-arbdefective $c$-coloring for each of the induced subgraphs graphs $G[V_i]$. Let $E'\subseteq E$ be the set of edges between two nodes of the same color group $i\in [K]$ and consider the subgraph $H=(V,E\setminus E')$ of $G$. Note that each set $V_i$, $i\in [K]$ is an independent set of $H$ and therefore the partitioning into $K$ color groups provides a proper $K$-coloring of $H$. We now compute a $(2,\beta)$-ruling set $S\subseteq V$ of $H$. By using the algorithm of \cite{balliurules} (Lemma 15 in the full version of \cite{balliurules}), a $(2,\beta)$-ruling set $S$ of a $K$-colored graph can be computed in the minimum time $t$ such that ${t+\beta \choose \beta}\geq K$. The set $S$ is an independent set of $H$ and therefore the only edges of the subgraph $G[S]$ of $G$ induced by the nodes in $S$ are edges between two nodes $u,v\in V_i$ for some $i$ (i.e., between nodes of the same color group). We therefore know that the initial $\alpha$-arbdefective coloring provides an $\alpha$-arbdefective $c$-coloring of $G[S]$ and $S$ therefore is an $\alpha$-arbdefective $c$-colored $\beta$-ruling set of $G$.
\end{proof}

By using the fact that ${a\choose b}\geq (a/b)^b$, we directly obtain the following corollary.
\begin{corollary}\label{cor:basicarbruling}
  Let $G=(V,E)$ be a graph and let $\alpha\geq 0$, $c\geq 1$, $\beta\geq 1$, and $C\geq c$. If one is given an $\alpha$-arbdefective $C$-coloring of $G$, then an $\alpha$-arbdefective $c$-colored $\beta$-ruling set of $G$ can deterministically be computed $O\big(\beta\cdot(C/c)^{1/\beta}\big)$ rounds in the \LOCAL model.
\end{corollary}

In order to compute the initial arbdefective coloring that is required to apply \Cref{cor:basicarbruling}, we can use \Cref{lemma:relaxedarbcoloring}. This gives the following complexity for computing an $\alpha$-arbdefective $c$-colored $\beta$-ruling set.

\begin{corollary}\label{cor:basicarbruling2}
Let $G=(V,E)$ be a graph and let $\alpha\geq 0$, $c\geq 1$, $\beta\geq 1$, and $C\geq c$. If one is given an $\alpha$-arbdefective $C$-coloring of $G$, then an $\alpha$-arbdefective $c$-colored $\beta$-ruling set of $G$ can deterministically be computed $O\big(\beta\big(\frac{\Delta}{(\alpha+1)\sqrt{c}}\big)^{2/(\beta+1)}\big)$ rounds in the \LOCAL model.
 \end{corollary}
 \begin{proof}
   We first compute an $\alpha$-arbdefective $C:=O(\eta\Delta/\alpha)$-coloring. By \Cref{lemma:relaxedarbcoloring}, this can be done in $O(\Delta/(\eta\alpha)+\log^* n)$ rounds. We choose $\eta$ such that $\Delta/(\eta\alpha)=(C/c)^{1/\beta}=(\eta\Delta/(\alpha c))^{1/\beta}$. We obtain $\eta=(\Delta/\alpha)^{\frac{\beta-1}{\beta+1}}\cdot c^{\frac{1}{\beta+1}}$, which together with \Cref{cor:basicarbruling} gives the required round complexity.
 \end{proof}
